%% file: main_map.tex
\documentclass[a4paper,UKenglish,cleveref, autoref,numberwithinsect]{lipics-v2019}
\input{Preamble.tex}

\nolinenumbers

\newcommand{\defproblemu}[3]{
  \vspace{1mm}
\noindent\fbox{
  \begin{minipage}{0.95\textwidth}
  #1 \\
  {\bf{Input:}} #2  \\
  {\bf{Question:}} #3
  \end{minipage}
  }
  \vspace{1mm}
}

\title{Decomposition of Map Graphs with Applications}

\author{Fedor V. Fomin}{University of Bergen, {Norway}, fomin@ii.uib.no}{}{}{}

\author{Daniel Lokshtanov}{University of California, Santa Barbara, {USA}, daniello@ucsb.edu}{}{}{}
\author{Fahad Panolan}{University of Bergen, {Norway}, fahad.panolan@ii.uib.no}{}{}{}
\author{Saket Saurabh}{The Institute of Mathematical Sciences, HBNI,  {Chennai, India}, saket@imsc.res.in}{}{}{}

\author{Meirav Zehavi}{Ben-Gurion University of the Negev, {Beer-Sheva, Israel}, meiravze@bgu.ac.il}{}{}{}

\authorrunning{F.V. Fomin, D. Lokshatnov, F. Panolan, S. Saurabh and M. Zehavi}

\Copyright{Fedor V.~Fomin, Daniel Lokshtanov, Fahad Panolan, Saket Saurabh and Meirav Zehavi}


\begin{document}
\maketitle

\input{abstract}

\newpage
\input{introduction}

\input{prelims}

\input{decomposition}
\input{FVS}

\input{cycle}

\input{cycle-packing}

\input{conclusion}

\bibliography{book_kernels_fvf,refs}


\end{document}

%% file: Preamble.tex
\usepackage[dvipsnames,svgnames,x11names]{xcolor}
\usepackage{amssymb,amsmath,amsthm,amsfonts}
\usepackage{mathrsfs}
\usepackage[dvipsnames,svgnames,x11names]{xcolor}
\usepackage{soul}

\usepackage{todonotes}
\usepackage{microtype,xspace}
\usepackage{graphicx,tikz}
\usepackage{framed}

\usepackage{thmtools} 
\usepackage{thm-restate}

\theoremstyle{plain}
\newtheorem{observation}[theorem]{Observation}
\usepackage{xspace}

\newcommand{\introducefakevertex}{{\bf fake introduce}}
\newcommand{\org}{{\sf Original}}
\newcommand{\fake}{{\sf Fake}}
\newcommand{\cliques}{{\sf Cliques}}
\newcommand{\NSTD}[1]{{\sf #1-NFewCliTD}}
\newcommand{\NSTDlong}[1]{nice #1-few cliques tree decomposition}
\newcommand{\STD}[1]{{\sf #1-FewCliTD}}

\newcommand{\STDlong}[1]{{#1-few cliques tree decomposition}}

\newcommand{\tw}{{\sf tw}}

\newcommand{\TT}{{\mathcal T}}
\newcommand{\sTT}{\scriptscriptstyle {\mathcal T}}
\newcommand{\QQ}{{\cal Q}}

\newcommand{\ZZ}{{\cal Z}}

\newcommand{\D}{{\cal D}}
\newcommand{\Yes}{{\sc Yes}}
\newcommand{\No}{{\sc No}}

\newcommand{\cO}{\mathcal{O}}
\newcommand{\OO}{\mathcal{O}}

\newcommand{\probKPath}{\textsc{Longest Path}\xspace}
\newcommand{\probKCycle}{\textsc{Longest Cycle}\xspace}
\newcommand{\probKexactCycle}{\textsc{Exact $k$-Cycle}\xspace}
\newcommand{\probFVS}{\textsc{Feedback Vertex Set}\xspace}

\newcommand{\probCycPacking}{\textsc{Cycle Packing}\xspace}

\newcommand{\ProblemFormat}[1]{{\sc #1}}

\newcommand{\ProblemIndex}[1]{\index{problem!\ProblemFormat{#1}}}

\newcommand{\ProblemName}[1]{\ProblemFormat{#1}\ProblemIndex{#1}\xspace}

\newcommand{\probTWdel}{\ProblemName{Treewidth-$\eta$ Vertex Deletion}}
\newcommand{\probPWdel}{\ProblemName{Pathwidth-$\eta$ Vertex Deletion}}
\newcommand{\probTDepthdel}{\ProblemName{Treedepth-$\eta$ Vertex Deletion}}



%% file: abstract.tex
\begin{abstract}

  Bidimensionality is the most common technique to design subexponential-time parameterized algorithms on special classes of graphs, particularly planar graphs. The core engine behind it is a combinatorial lemma of Robertson, Seymour and Thomas that states that every planar graph either has a $\sqrt{k}\times \sqrt{k}$-grid as a minor, or its treewidth is $\OO(\sqrt{k})$. 
 However, bidimensionality theory cannot be extended directly to several well-known classes of geometric graphs. 
 The reason is very simple: a clique on $k-1$ vertices has no $\sqrt{k}\times \sqrt{k}$-grid as a minor and its treewidth is $k-2$, while classes of geometric graphs such as unit disk graphs or  map graphs can have arbitrarily large cliques. Thus, the combinatorial lemma of Robertson, Seymour and Thomas is inapplicable to these classes of geometric graphs.   Nevertheless, a relaxation of this lemma has been proven useful for unit disk graphs. Inspired by this, we prove a new decomposition lemma for map graphs, the intersection graphs of finitely many simply-connected and interior-disjoint regions of the Euclidean plane. 
 Informally, our lemma  states the following.  For any map graph $G$, there exists a  collection  $(U_1,\ldots,U_t)$ of cliques of $G$ with  the following property:  
{\em 
$G$ either contains a $\sqrt{k}\times \sqrt{k}$-grid as a minor, or it admits a tree decomposition where every bag is the union of $\OO(\sqrt{k})$ of the cliques in the above collection. 
}

The new lemma appears to be a handy tool in the design of subexponential parameterized algorithms on map graphs. We demonstrate its usability  by designing algorithms on map graphs with running time $2^{\cO({\sqrt{k}\log{k}})} \cdot n^{\cO(1)}$ for {\sc Connected Planar $\cal F$-Deletion} (that encompasses problems such as {\sc Feedback Vertex Set} and {\sc Vertex Cover}). Obtaining subexponential algorithms for \probKCycle/{\sc Path}
and \probCycPacking is more challenging. We have to construct   tree  decompositions with more powerful properties and to prove sublinear bounds on the number of ways an optimum solution could  ``cross''   bags in these decompositions.

For 
\probKCycle/{\sc Path}, these are  the first subexponential-time parameterized algorithms on map graphs.
For \probFVS and \probCycPacking, we improve upon known $2^{\cO({k^{0.75}\log{k}})} \cdot n^{\cO(1)}$-time algorithms on map graphs.

\end{abstract}

%% file: introduction.tex
\section{Introduction}\label{sec:intro}
In this paper, we develop new proof techniques to design parameterized subexponential-time algorithms for problems on map graphs, particularly problems that involve hitting or connectivity constraints. The class of map graphs was introduced by Chen, Grigni, and Papadimitriou \cite{ChenGP98,ChenGP02} as a modification of the class of planar graphs. Roughly speaking, map graphs are graphs whose vertices represent countries in a map, where two countries are considered adjacent if and only if their boundaries have at least one point in common; this common point can be a single common point rather than necessarily an edge as standard planarity requires. Formally,  a {\em map} $\cal M$ is a pair  $(\mathscr{E},\omega)$ defined as follows (see
Figure~\ref{map_graph}): $\mathscr{E}$ is a plane graph\footnote{That is, a planar graph with a drawing in the plane.} where each connected component of $\mathscr{E}$ is  biconnected, and $\omega$ is a function that maps each face $f$ of $\mathscr{E}$ to $0$ or $1$.  A face $f$ of $\mathscr{E}$ is called {\em nation} if $\omega(f)=1$ and {\em lake} otherwise. The graph associated with $\cal M$ is the simple graph $G$ where $V(G)$ consists of the nations of $\cal M$, and $E(G)$ contains $\{f_1,f_2\}$ for every pair of faces $f_1$ and $f_2$ that are adjacent (that is, share at least one vertex). Accordingly, a graph $G$ is called a map graph if there exists a map $\cal M$ such that $G$ is the graph associated with~$\cal M$.

Every planar graph is a map graph~\cite{ChenGP98,ChenGP02}, but the converse does not hold true. Moreover, map graphs can have cliques of any size and thus  they can be ``highly non-planar''. These two properties of map graphs can be contrasted with those of  $H$-minor free graphs and unit disk graphs: the class of $H$-minor free graphs generalizes the class of planar graphs, but can only have cliques of constant size (where the constant depends on $H$), while the class of unit disk graphs does not generalize the class of planar graphs, but can have cliques of any size. At least in this sense, map graphs offer the best of both worlds. Nevertheless, this comes at the cost of substantial difficulties in the design of efficient algorithms on them.


Arguably, the two most natural and central algorithmic questions concerning map graphs are as follows. First, we would like to efficiently recognize map graphs, that is, determine whether a given graph is a map graph. In 1998, Thorup \cite{Thorup98a} announced the existence of a polynomial-time algorithm for map graph recognition. Although this algorithm is complicated and its running time is about $\OO(n^{120})$, where $n$ is the number of vertices of the input graph, no improvement has yet been found; the existence of a simpler or faster algorithm for map graph recognition has so far remained an important open question in the area (see, e.g., \cite{ChenGP06}).
 
The second algorithmic question---or rather family of algorithmic questions---concerns the design of efficient algorithms for various optimization problems on map graphs. Most well-known problems that are NP-complete on general graphs remain NP-complete when restricted to planar (and hence on map) graphs. Nevertheless, a large number of these problems can be solved faster or ``better'' when restricted to planar graphs.  For example, nowadays we know of many problems that are APX-hard on general graphs, but which admit polynomial time approximation schemes (PTASes) or even efficient PTASes (EPTASes) on planar graphs (see, e.g., \cite{Baker94,DemaineHaj05,Demaine:2008mi,FominLS17}).  
Similarly,  many parameterized problems that on general graphs  cannot be solved in time  $2^{o(k)} \cdot n^{\cO(1)}$
 unless the Exponential Time Hypothesis (ETH) of  Impagliazzo, Paturi and Zane 
\cite{ImpagliazzoPZ01} fails, admit parameterized algorithms with running times subexponential in $k$ on planar graphs (see, e.g., \cite{AlberBFKN02,AlberFN04,DemaineHaj05,DBLP:conf/icalp/Marx13}). It is compelling to ask whether the algorithmic results and techniques for planar graphs can be extended to~map~graphs. 
 
 \begin{figure} 
\begin{center}
\includegraphics[scale=.21]{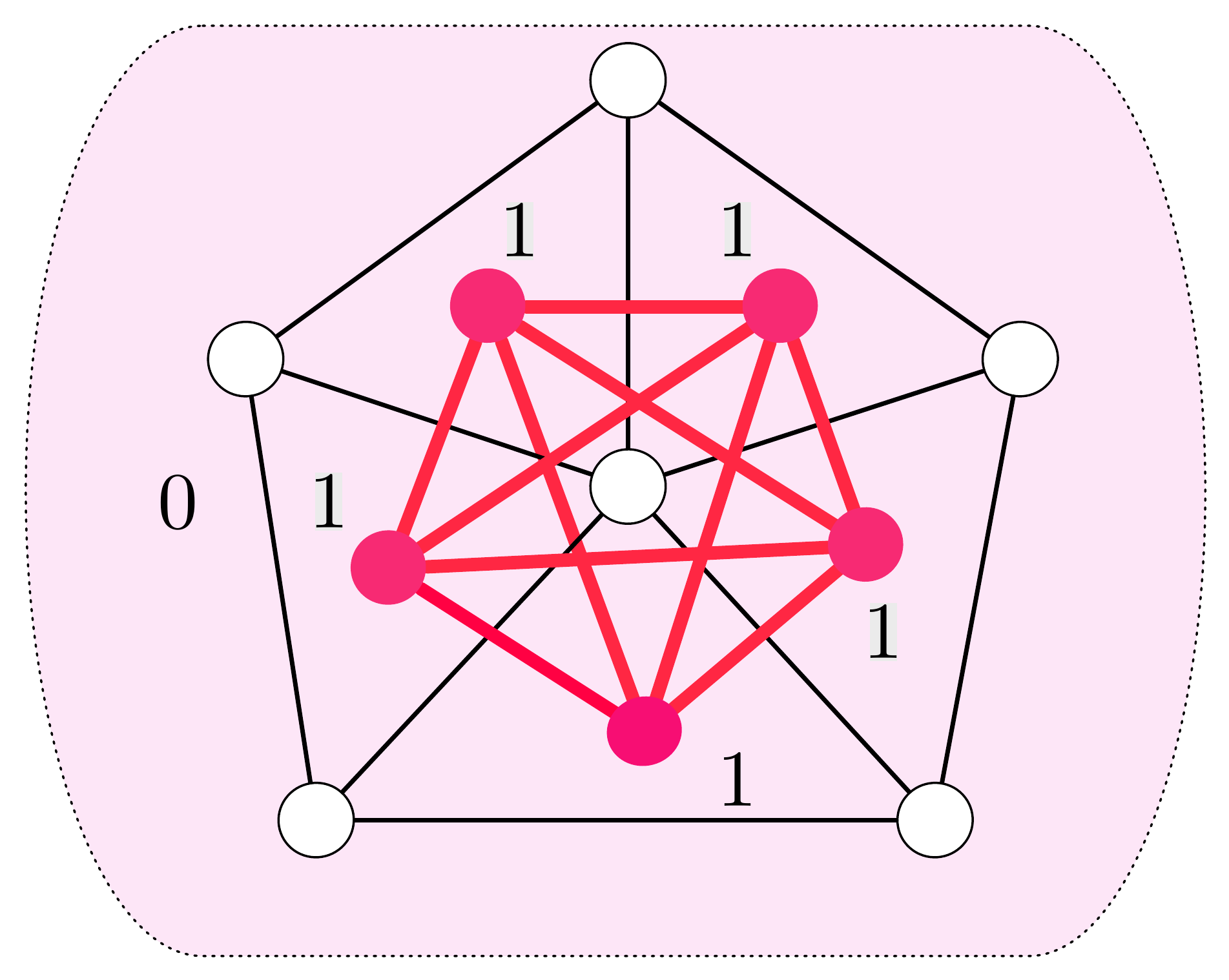} 
\caption{ A map $\cal M=(\mathscr{E},\omega)$ where the internal $1$-faces are nations and the $0$-exterior face is a lake. The corresponding map graph is a complete graph on five vertices.}\label{map_graph}
\end{center}
\end{figure}

 For approximation algorithms, Chen~\cite{Chen01} and Demaine et al.~\cite{DFHT05talg} developed PTASes for the \textsc{Maximum Independent Set} and \textsc{Minimum $r$-Dominating Set} problems on map graphs. Moreover, Fomin et al.~\cite{FominLS12,FominLS17} developed an EPTAS for \textsc{Treewidth-$\eta$ Modulator} for any fixed constant $\eta\geq 0$, which encompasses {\sc Feedback Vertex Set} and {\sc Vertex Cover}. 
For parameterized subexponential-time algorithms on map graphs, the situation is less explored.   While on planar graphs there are general algorithmic methods---in particular, the powerful theory of bidimensionality \cite{Demaine:2008mi,DemaineFHT05jacm}---to design parameterized subexponential-time algorithms, we are not aware of any general algorithmic method that can be easily adapted to map graphs. 
Demaine et al. 
 \cite{DFHT05talg} gave a  parameterized algorithm for \textsc{Dominating Set}, and more generally for $(k,r)$-\textsc{Center}, with running time  $2^{\OO(r\log{r}\sqrt{k})}n^{\cO(1)}$ on map graphs. 
 Moreover,  Fomin et al.~\cite{FominLS12,FominLS17} gave $2^{\cO(k^{0.75}\log k)} n^{\cO(1)}$-time parameterized algorithms for 
\probFVS and \probCycPacking on map graphs. Additionally, Fomin et al.~\cite{FominLS12,FominLS17} noted that the same approach yields  $2^{\cO(k^{0.75}\log k)} n^{\cO(1)}$-time parameterized algorithms for 
{\sc Vertex Cover} and {\sc Connected Vertex Cover} on map graphs.  However, the existence of a parameterized subexponential-time algorithm for  \probKPath/\textsc{Cycle}\footnote{In the \probKPath/\textsc{Cycle}  problem, we ask whether a given graph $G$ contains a path/cycle on at least $k$ vertices. Here, the parameter is $k$.} on map graphs was left open. Furthermore, time complexities of $2^{\cO(k^{0.75}\log k)} n^{\cO(1)}$, although having subexponential dependency on $k$, remain far from time complexities of $2^{\cO(\sqrt{k}\log k)} n^{\cO(1)}$ and $2^{\cO(\sqrt{k})} n^{\cO(1)}$ that commonly arise for planar graphs~\cite{DBLP:conf/icalp/Marx13}. We remark that time complexities of $2^{\cO(\sqrt{k}\log k)} n^{\cO(1)}$ and $2^{\cO(\sqrt{k})} n^{\cO(1)}$ are particularly important since they are often known to be essentially optimal under the aforementioned ETH~\cite{DBLP:conf/icalp/Marx13}.

In the field of Parameterized Complexity,  \probKPath/\textsc{Cycle} ,  \probFVS and     \probCycPacking serve  as   testbeds for development of  fundamental algorithmic techniques such as color-coding~\cite{AlonYZ},  methods based on polynomial identity testing~\cite{Koutis08,KoutisW16,Williams09,BjHuKK10}, cut-and-count~\cite{Cygan11}, and methods based on matroids \cite{FominLPS16}. 
We refer to  
\cite{cygan2015parameterized}
  for an extensive overview of the literature on parameterized algorithms for  these three problems on general graphs.  
  By combining  the bidimensionality theory  of Demaine et al. \cite{DemaineFHT05jacm} with efficient   algorithms on graphs of bounded treewidth \cite{DornPBF10,cygan2015parameterized}, 
 \probKPath/\textsc{Cycle}, \probCycPacking\ and  \probFVS
  are solvable in time $2^{\cO(\sqrt{k})}n^{\cO(1)}$ on planar graphs.
 Furthermore, the parameterized subexponential-time  ``tractability'' of these problems can be extended to  graphs excluding some fixed graph as a minor~\cite{Demaine:2008mi}.

\subsection*{Our results and methods} 
\noindent\textbf{Our results.}
We design parameterized subexponential-time algorithms with running time $2^{\cO({\sqrt{k}\log{k}})} \cdot n^{\cO(1)}$ for 
a number of natural and well-studied problems on map graphs.

Let $\cal F$ be a family of connected graphs that contains  at least one planar graph. Then {\sc Connected Planar $\cal F$-Deletion} (or just  {\sc  $\cal F$-Deletion}) is defined as follows. 

\defproblemu{ {\sc  $\cal F$-Deletion}}%
{A graph $G$ and a non-negative integer $k$.}%
{ Is there  a set $S$ of at most $k$ vertices such that $G - S$ does not contain any of the graphs in $\cal F$ as a minor?}

 {\sc  $\cal F$-Deletion} is a general problem and several 
 problems such as {\sc Vertex Cover}, {\sc Feedback Vertex Set}, {\sc \probTWdel}, \probPWdel, \probTDepthdel, {\sc Diamond Hitting Set} and {\sc Outerplanar Vertex Deletion} are its special cases.
 We give the first parameterized subexponential algorithm for this problem on map graphs, which runs in time 
 $2^{\cO({\sqrt{k}\log{k}})} \cdot n^{\cO(1)}$.
%
%
%
  Our approach for  {\sc  $\cal F$-Deletion} also directly extends to yield $2^{\cO({\sqrt{k}\log{k}})} \cdot n^{\cO(1)}$-time parameterized algorithms for  {\sc Connected Vertex Cover} and {\sc Connected Feedback Vertex Set} on map graphs. (In this versions we are asked if there is a \emph{connected} vertex cover or a  feedback vertex set of size at most $k$.)
  
  With additional ideas,  we derive the first subexponential-time parameterized algorithm on map graphs for \probKPath/\textsc{Cycle}. 
  Our technique also allows to improve the running time for  \probCycPacking (does a map graph contains at least $k$ vertex-disjoint cycles) from $2^{\cO({k^{0.75}\log{k}})} \cdot n^{\cO(1)}$ to $2^{\cO({\sqrt{k}\log{k}})} \cdot n^{\cO(1)}$.
  

Our results are summarized in Table~\ref{tabl:compl}.

\begin{table}[ht]
\begin{center}
{\small
\begin{tabular}{|l|c|c|}
\hline
& Our results & Previous work \\
\hline
{\sc Vertex Cover} &   $2^{\cO({\sqrt{k}\log{k}})} \cdot n^{\cO(1)}$ [Thm~\ref{thm:fdeletion}] &  $2^{\cO({k^{0.75}\log{k}})} \cdot n^{\cO(1)}$ \cite{FominLS17} \\
\hline
{\sc Connected Vertex Cover} &   $2^{\cO({\sqrt{k}\log{k}})} \cdot n^{\cO(1)}$ [Thm~\ref{thm:confvsvc}] &  $2^{\cO({k^{0.75}\log{k}})} \cdot n^{\cO(1)}$  \cite{FominLS17} \\
\hline
\probFVS &   $2^{\cO({\sqrt{k}\log{k}})} \cdot n^{\cO(1)}$  [Thm~\ref{thm:fvs}] &  $2^{\cO({k^{0.75}\log{k}})} \cdot n^{\cO(1)}$  \cite{FominLS17} \\
\hline
{\sc Connected Feedback Vertex Set} &   $2^{\cO({\sqrt{k}\log{k}})} \cdot n^{\cO(1)}$ [Thm~\ref{thm:confvsvc}] &  $2^{\cO({k^{0.75}\log{k}})} \cdot n^{\cO(1)}$ \cite{FominLS17}  \\
\hline
 {\sc  $\cal F$-Deletion} &   $2^{\cO({\sqrt{k}\log{k}})} \cdot n^{\cO(1)}$ [Thm~\ref{thm:fdeletion}] &  $2^{\cO(k)}n^{\cO(1)}$ \cite{FominLMS12}   \\
\hline
 \probKPath &   $2^{\cO({\sqrt{k}\log{k}})} \cdot n^{\cO(1)}$ [Thm~\ref{thm:path}] &  $2^{\cO(k)}n^{\cO(1)}$   \cite{cygan2015parameterized} \\
\hline
{\sc Longest Cycle}  &   $2^{\cO({\sqrt{k}\log{k}})} \cdot n^{\cO(1)}$ [Thm~\ref{thm:cycle}] &  $2^{\cO(k)}n^{\cO(1)}$   \cite{cygan2015parameterized}  \\
\hline
 \probCycPacking &   $2^{\cO({\sqrt{k}\log{k}})} \cdot n^{\cO(1)}$ [Thm~\ref{thm:cycPack}] &$2^{\cO({k^{0.75}\log{k}})} \cdot n^{\cO(1)}$  \cite{FominLS17}\\
\hline

                         \end{tabular}
 }
\caption{Parameterized complexity of problems on map graphs. For   {\sc  $\cal F$-Deletion}, {\sc Longest Cycle}, and  \probKPath no faster (than on general graphs) algorithms were known for map graphs.}  \label{tabl:compl}
\end{center}
\end{table}



\medskip
\noindent\textbf{Methods.}
The starting point of our study is the technique of bidimensionality~\cite{Demaine:2008mi,DemaineFHT05jacm}. The core engine behind this technique is a combinatorial lemma of Robertson, Seymour and Thomas~\cite{ROBERTSON1994323} that states that every planar graph either has a $\sqrt{k}\times \sqrt{k}$-grid as a minor, or its treewidth is $\OO(\sqrt{k})$. Unfortunately, a clique on $k-1$ vertices has no $\sqrt{k}\times \sqrt{k}$-grid as a minor and its treewidth is $k-2$. Because classes of geometric graphs such as unit disk graphs and map graphs can have arbitrarily large cliques, the combinatorial lemma is inapplicable to them. Nevertheless, a relaxation of this lemma has been proven useful for unit disk graphs. Specifically, every unit disk graph $G$ has a natural partition $(U_1,\ldots,U_t)$ of $V(G)$ such that each part induces a clique with ``nice'' properties---in particular, it has neighbors only in a {\em constant number} (to be precise, this constant is at most 24) of other parts; it was shown that $G$ either has a $\sqrt{k}\times \sqrt{k}$-grid as a minor, or it has a tree decomposition where every bag is the union of $\OO(\sqrt{k})$ of these cliques~\cite{subexpudg}. In particular, given a parameterized problem where any two cliques have constant-sized ``interaction'' in a solution, it is implied that any bag has $\OO(\sqrt{k})$-sized ``interaction'' with all other bags in a solution.
For any map graph $G$, there also exists a natural collection of subsets of $V(G)$ that induce cliques with ``nice'' properties. However, not only are these cliques not vertex disjoint, but each of these cliques can have neighbors in {\em arbitrarily many} other cliques.

In this paper, we first prove that every map graph either has a $\sqrt{k}\times \sqrt{k}$-grid as a minor, or it has a tree decomposition where every bag is the union of $\OO(\sqrt{k})$ of the cliques in the above collection. For {\sc $\cal F$-Deletion}, {\sc Connected Vertex Cover}, and {\sc Connected Feedback Vertex Set}, this combinatorial lemma alone already suffices to design $2^{\cO({\sqrt{k}\log{k}})} \cdot n^{\cO(1)}$-time algorithms on map graphs. Indeed, we can choose a fixed constant $c>0$ so that in case we have a $c\sqrt{k}\times c\sqrt{k}$-grid as a minor, there does not exist a solution, and otherwise we can solve the problem by using dynamic programming over the given tree decomposition. Specifically, since every bag is the union of $\OO(\sqrt{k})$ cliques, and the size of each clique is upper bounded by $\OO(k)$ (once we know that no $c\sqrt{k}\times c\sqrt{k}$-grid exists), only $\OO(\sqrt{k})$ vertices in the bag are not to be taken into a solution---there are only $2^{\cO({\sqrt{k}\log{k}})}$ choices to select these vertices, and once they are selected, the information stored about the remaining vertices is the same as in normal dynamic programming over a tree decomposition of $\OO(\sqrt{k})$ width. 

This  approach already substantially improves upon the previously best known algorithms for {\sc Feedback Vertex Set},  {\sc Vertex Cover} and {\sc Connected Vertex Cover} of Fomin et al.~\cite{FominLS12,FominLS17}. However, 
  $2^{\cO({\sqrt{k}\log{k}})} \cdot n^{\cO(1)}$-time algorithms for \probKPath/\textsc{Cycle}  and \probCycPacking on map graphs require more efforts. 
%
 The main reason why we cannot apply the same arguments as for unit disk graphs is the following.  Recall that for unit disk graphs, given a parameterized problem where any two cliques have constant-sized ``interaction'' in a {\em solution} (in our case, this means a path/cycle on at least $k$ vertices, or a cycle packing of $k$ cycles), it is implied that any bag has $\OO(\sqrt{k})$-sized ``interaction'' with all other bags in a solution.  Here, interaction between two cliques refers to the number of edges in a solution ``passing'' between these two cliques; similarly, interaction between a bag $B$ and a collection of other bags refers to the number of edges in a solution that have one endpoint in $B$ and the other endpoint in some bag in the collection.  In this context, dealing with map graphs is substantially more difficult than dealing with unit disk graphs. In map graphs vertices in a clique can have neighbors in arbitrarily many other cliques in the collection rather than only in a constant number as in unit disk graphs. This is why
it is difficult to obtain an $\OO(\sqrt{k})$-sized ``interaction'' as before.

This is the reason why we are forced to take a different approach for map graphs by  bounding  ``the interaction within  a clique across all the bags of a decomposition''.   
Towards this, 
we first need to strengthen our tree decomposition. To explain the new properties required, we note that every clique in the aforementioned collection of cliques, say $\cal K$, is either a single vertex or the neighborhood of some ``special vertex'' in an exterior bipartite graph (see Section \ref{sec:prelim}). Further, every vertex of $G$ occurs as a singleton in $\cal K$. 
We construct our decomposition in a way such that every bag is not necessarily a union of $\OO(\sqrt{k})$ cliques in $\cal K$, but a union of carefully chosen subcliques of $\OO(\sqrt{k})$ cliques in $\cal K$ (with one subclique for each of these $\OO(\sqrt{k})$ cliques); subcliques of the same clique chosen in different bags may be different. We then prove properties that roughly state that, if we look at the collection of bags that include some vertex $v$ of $G$, then this collection induces a subtree and a path as follows: $(\clubsuit)$ {\em the subtree consists of the bags that correspond to the singleton clique $v$, and the path goes ``upwards'' (in the tree decomposition) from the root of this subtree}. We thereby implicitly derive that in every bag $B$, every subclique of size larger than $1$ can only have as neighbors vertices that are {\em (i)} in the bag $B$ itself or in one of its descendants, or {\em (ii)} in cliques that have a subclique in the bag $B$. In particular, this means that if we prove that there exists a solution such that 
for any clique $K$ in $\cal K$, the number of edges in $E(K)$ that ``cross any bag $B$'' (i.e., the edges in $E(K)$ with one endpoint in  $B$ and the other in the collection of all bags that are not descendants of $B$) is a constant,  then we obtain a bound of $\OO(\sqrt{k})$ on the interaction between any bag $B$ and the collection of all bags that are not descendants of $B$. We prove the mentioned statement using property $(\clubsuit)$. 
The  proof that  such a property simultaneously holds for all cliques and all bags is the most challenging part of the proof. 

We remark that we discussed  above two types of tree decompositions, first the special one and then its stronger variant which is used for \probKPath/\textsc{Cycle} and  \probCycPacking. Since the stronger variant of the decomposition can be used to work with 
 {\sc ${\cal F}$-Deletion} too, in the technical part of this paper we  derive only the stronger variant of the  decomposition.  In Section~\ref{sec:prelim}, we give definitions, notations and some known results which we use throughout the paper. In Section~\ref{sec:fewclde}, we design a tree decomposition of map graphs which we call as few clqiues tree decomposition, and in Section~\ref{sec:fvs} we explain its direct applicability for \probFVS and {\sc ${\cal F}$-Deletion}.  In  Sections~\ref{sec:exactCyc} and \ref{sec:cycPack} we design subexponential-time parameterized algorithms for \probKPath/\textsc{Cycle} and  \probCycPacking on map graphs, respectively. For  these two problems we need  additional, somehow technically involved, combinatorial  ``sublinear crossing'' lemmata.

 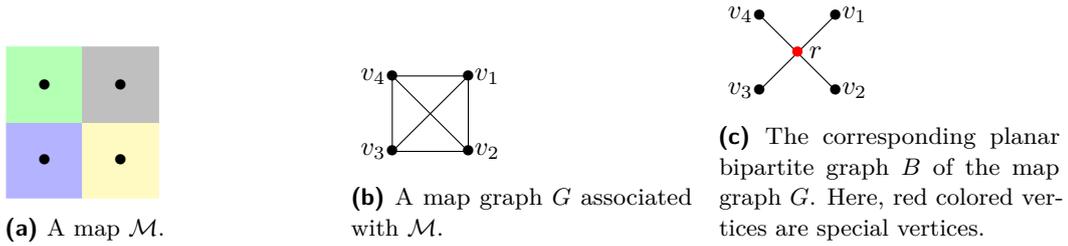
\begin{figure}
    \centering
    \begin{subfigure}[b]{0.3\textwidth}
       \begin{tikzpicture}
         \filldraw[draw=black,color=blue!30!white] (0,0) rectangle (-1,-1);
         \filldraw[draw=black,color=green!30!white] (0,0) rectangle (-1,1);
         \filldraw[draw=black,color=yellow!30!white] (0,0) rectangle (1,-1);
         \filldraw[draw=black,color=lightgray] (0,0) rectangle (1,1);
         \node[]  at (0.5,0.5) (a) {$\bullet$};
         \node[]  at (-0.5,-0.5) (a) {$\bullet$};
         \node[]  at (-0.5,0.5) (a) {$\bullet$};
         \node[]  at (0.5,-0.5) (a) {$\bullet$};
       \end{tikzpicture}
        \caption{A map $\cal M$.}
        \label{fig:map}
    \end{subfigure}
    ~ 
    \begin{subfigure}[b]{0.32\textwidth}
       \begin{tikzpicture}
         \node[]  at (0.5,0.5) (a) {$\bullet$};
         \node[]  at (-0.5,-0.5) (a) {$\bullet$};
         \node[]  at (-0.5,0.5) (a) {$\bullet$};
         \node[]  at (0.5,-0.5) (a) {$\bullet$};
         \node[]  at (0.75,0.5) (a) {$v_1$};
         \node[]  at (-0.75,-0.5) (a) {$v_3$};
         \node[]  at (-0.75,0.5) (a) {$v_4$};
         \node[]  at (0.75,-0.5) (a) {$v_2$};
         \draw (0.5,0.5)--(0.5,-0.5)--(-0.5,-0.5)--(-0.5,0.5)--(0.5,0.5)--(-0.5,-0.5);
         \draw (-0.5,0.5)--(0.5,-0.5);
       \end{tikzpicture}
        \caption{A map graph $G$ associated with $\cal M$.}
        \label{fig:mapgraph}
    \end{subfigure}
    ~ 
    \begin{subfigure}[b]{0.32\textwidth}
       \begin{tikzpicture}
               \draw (0.5,0.5)--(-0.5,-0.5);
         \draw (-0.5,0.5)--(0.5,-0.5);
                  \node[]  at (0.75,0.5) (a) {$v_1$};
         \node[]  at (-0.75,-0.5) (a) {$v_3$};
         \node[]  at (-0.75,0.5) (a) {$v_4$};
         \node[]  at (0.75,-0.5) (a) {$v_2$};
      \node[]  at (0.24,0) (a) {$r$};
         \node[]  at (0.5,0.5) (a) {$\bullet$};
         \node[]  at (-0.5,-0.5) (a) {$\bullet$};
         \node[]  at (-0.5,0.5) (a) {$\bullet$};
         \node[]  at (0.5,-0.5) (a) {$\bullet$};
         \node[red] at (0,0) (c) {{$\bullet$}};
       \end{tikzpicture}

        \caption{The corresponding planar bipartite graph $B$ of the map graph $G$. Here, red colored vertices are special vertices.}
        \label{fig:corrplanargraph}
    \end{subfigure}
    \caption{Example of a map graph $G$ obtained from a corresponding planar bipartite graph $B$.}\label{fig:mapgraphderviation}
\end{figure}

%% file: prelims.tex

\section{Preliminaries}\label{sec:prelim}

The set of natural numbers is denoted by ${\mathbb N}$. For any $t\in {\mathbb N}$, we use $[t]$ and $[t]_0$ as  shorthands for $\{1,2,\ldots,t\}$ and $\{0,1,\ldots,t\}$, respectively.  
For a set $U$, we use $2^U$ to denote the power set of $U$. Two disjoint sets $A$ and $B$, we use $A\uplus B$ to denote the disjoint union of $A$ and $B$. 
For a sequence $\sigma=x_1x_2\ldots x_n$ and any $1\leq i\leq j\leq n$, the sequence 
$\sigma' =x_i\ldots x_j$  is called a {\em segment} of $\sigma$.  
For a sequence $\sigma=x_1x_2\ldots x_n$ and a subset $Z\subseteq \{x_1,\ldots,x_n\}$, the restriction 
of $\sigma$ on $Z$, denoted by $\sigma|_Z$, is the sequence 
obtained from $\sigma$ by deleting the elements of $\{x_1,\ldots,x_n\}\setminus Z$.

\subparagraph*{Standard graph notations.} We use standard notation and terminology from the book of Diestel~\cite{NewDiestel} for graph-related terms that are not explicitly defined here. 
 Given a graph $G$, let $V(G)$ and $E(G)$ denote its vertex-set and edge-set, respectively. When the graph $G$ is clear from context, we denote $n=|V(G)|$ and $m=|E(G)|$. 
 For a set ${\cal Q}$ of graphs we slightly abuse terminology and let $V({\cal Q})$ and $E({\cal Q})$  denote the union of the sets of vertices and edges of the graphs in $\cal Q$, respectively. 
A graph is {\em simple} if it contains neither loops nor multiple edges between pairs of vertices.  
Throughout the paper, when we use the term {\em graph} we refer to a simple graph. 
Given $U\subseteq V(G)$, let $G[U]$ denotes the subgraph of $G$ induced by $U$.
For an edge subset $E\subseteq E(G)$, let $V(E)$ denotes the set of endpoints of the edges in $E$, and $G[E]$  denotes the graph with vertex set $V(E)$ and edge set $E$. 
Given $X\subseteq V(G)$, let $E(X)$  denotes the 
edge set $\{\{u,v\}\in E(G)\colon u,v\in X\}$.
Moreover, let $N_G(U)$ denotes the open neighborhood of $U$ in $G$; we omit the subscript $G$ when the graph is clear from context. In case $U=\{v\}$, we slightly abuse terminology and use $N_G(v)=N_G(U)$. 
For a graph $G$ and a vertex $v\in V(G)$, let 
$d_G(v)=\vert N_G(v)\vert$. 
A graph $H$ is called a {\em minor} of $G$ if $H$ can be obtained from $G$ by a sequence of edge deletions, 
edge contractions, and vertex deletions. 
For a  graph $G$ and a degree-$2$ vertex $v\in V(G)$, by {\em contracting $v$}, we mean 
deleting  $v$ from $G$ and adding an edge between the two neighbors of $v$ in $G$.

In a  graph $G$, a sequence of vertices $[u_1u_2\ldots u_{\ell}]$ is  a path in 
$G$ if  for any  distinct $i,j\in [\ell]$, $u_i\neq u_j$, and for any $r\in [\ell-1]$, $\{u_r,u_{r+1}\}\in E(G)$.  
We also call the path $P=[u_1u_2\ldots u_{\ell}]$ as {\em $u_1$-$u_{\ell}$ path}, and its 
 internal vertices are $u_2,u_3,\ldots,u_{\ell-1}$. 
For any two paths $P_1=[u_1\ldots u_{i}]$ and $P_2=[u_i\ldots u_{\ell}]$ with $\{u_1,\ldots,u_{i-1}\}\cap \{u_{i+1},\ldots,u_{\ell}\}=\emptyset$, let  
$P_1P_2$ denotes the path  $[u_1u_2\ldots u_{\ell}]$. 
A sequence of vertices $[u_1u_2\ldots u_{\ell}]$ is a cycle in 
$G$ if  $u_1=u_{\ell}$, $[u_1u_2\ldots u_{\ell-1}]$ is a path, and $\{u_{\ell-1},u_{\ell}\}\in E(G)$.
Since in a multi graph there can be more than one edges between a pair of vertices, we use sequence $[u_0e_0u_1e_1\ldots e_{\ell}u_{1}]$ to denote a cycle. In that context,
for each $i\in [\ell]_0$, $e_i$ is an edge between $u_i$ and $u_{(i+1) \mod \ell}$.  
For a graph $G$, we say that $U\subseteq V(G)$ is a clique if $G[U]$ is a complete graph.  
Given $a,b\in\mathbb{N}$, an $a\times b$ grid is a graph on $a\cdot b$ vertices, $v_{i,j}$ for $(i,j)\in[a]\times[b]$, such that for all $i\in[a-1]$ and $j\in[b]$, it holds that $v_{i,j}$ and $v_{i+1,j}$ are neighbors, and for all $i\in[a]$ and $j\in[b-1]$, it holds that $v_{i,j}$ and $v_{i,j+1}$ are neighbors. 

A binary tree is a rooted tree where each node  has 
at most two children. In a labelled binary tree, for each node with two children one of the children is labelled as ``left child'' 
and the  other child is labelled as ``right child''. A {\em postorder transversal} of a labelled binary tree $T$ is the sequence $\sigma$ of $V(T)$ where for each node $t\in V(T)$, $t$ appears after all its descendants, and if $t$ has two children, then the nodes in the subtree rooted 
at the left child appear before  the nodes in the subtree rooted 
at the right child.  
For a binary tree $T$, we say that a sequence $\sigma$ of $V(T)$ is a postorder 
transversal if there is a labelling of $T$ such that $\sigma$ is its postorder transversal.

A tree decomposition of a graph $G$, which is defined as follows, measures how close the graph $G$ is to a tree like structure.

\begin{definition}[Treewidth]\label{def:treeDecomp}
A {\em tree decomposition} of a graph $G$ is a pair $\TT=(T_{\sTT},\beta_{\sTT})$, where $T$ is a rooted tree and $\beta_{\sTT}$ is a function from $V(T_{\sTT})$ to $2^{V(G)}$, that satisfies the following three conditions. (We use the term {\em nodes} to refer to the vertices 
of $T_{\sTT}$.) 
\begin{itemize}
\item[$(a)$] $\bigcup_{x\in V(T_{\sTT})}\beta_{\sTT}(x)=V(G)$.
\item[$(b)$] For every edge $\{u,v\}\in E(G)$, there exists $x\in V(T_{\sTT})$ such that  $\{u,v\}\subseteq \beta_{\sTT}(x)$.
\item[$(c)$] For every vertex $v\in V(G)$, the set of nodes $\{t\in V(T_{\sTT})~:~v\in \beta_{\sTT}(t)\}$ induces a (connected) subtree of $T_{\sTT}$.
\end{itemize}
The {\em width} of $\TT$ is $\max_{x\in V(T_{\sTT})} |\beta_{\sTT}(x)|-1$. Each set $\beta_{\sTT}(x)$ is called a {\em bag}. Moreover,  $\gamma_{\sTT}(x)$ denotes the union of the bags of $x$ and its descendants. The {\em treewidth} of $G$ is the minimum width among all possible tree decompositions of $G$, and it is denoted by $\tw(G)$.
\end{definition}


A {\em nice tree decomposition} is a tree decomposition of a form that simplifies the design of dynamic programming (DP) algorithms. Formally,

\begin{definition}
A tree decomposition $\TT=(T_{\sTT},\beta_{\sTT})$ of a graph $G$ is {\em nice} if 
for the root $r$ of $T_{\sTT}$, it holds that $\beta_{\sTT}(r)=\emptyset$, and 
each node $v\in V(T_{\sTT})$ is of one of the following types.
\begin{itemize}
\item {\bf Leaf}: $v$ is a leaf in $T_{\sTT}$ and $\beta_{\sTT}(v)=\emptyset$. This bag is labelled with {\bf leaf}. 
\item {\bf Forget vertex}: $v$ has exactly one child $u$, and there exists a vertex $w\in\beta_{\sTT}(u)$ such that $\beta_{\sTT}(v)=\beta_{\sTT}(u)\setminus\{w\}$. This bag is labelled with {\bf  forget$(w)$}. 
\item {\bf Introduce vertex}: $v$ has exactly one child $u$, and there exists a vertex $w\in\beta_{\sTT}(v)$ such that $\beta_{\sTT}(v)\setminus\{w\}=\beta_{\sTT}(u)$. This bag is labelled with {\bf  introduce$(w)$}.
\item {\bf Join}: $v$ has exactly two children, $u$ and $w$, and $\beta_{\sTT}(v)=\beta_{\sTT}(u)=\beta_{\sTT}(w)$. 
This bag is labelled with {\bf join}.
\end{itemize}
\end{definition}

We will use the following two folklore observations in Section~\ref{sec:fewclde}. The correctness of these  observations follows from  Condition $(c)$ of a tree decomposition.
\begin{observation}
\label{obs:onlyoneforget}
Let $\TT$ be a nice tree decomposition of a graph $G$. For any $v\in V(G)$, there is exactly one node 
$t\in V(T_{\sTT})$ such that $t$ is labelled with {\bf forget}$(v)$. 
\end{observation}

\begin{observation}
\label{obs:forgetsubtree}
Let $\TT$ be a nice tree decomposition of a graph $G$, $v\in V(G)$, and $t\in V(T_{\sTT})$ be the node labelled with {\bf forget}$(v)$. For any node $t'$ in the subtree of $T_{\sTT}$ rooted at  $t$ and $t'\neq t$, either $v\in \beta_{\sTT}(t')$ or 
$v\notin \gamma_{\sTT}(t')$.  
\end{observation}

The following proposition concerns the computation of a nice tree decomposition. 

\begin{proposition}[\cite{Bodlaender96}]\label{prop:nice}
Given a graph $G$ and a tree decomposition $\TT$ of $G$, a nice tree decomposition $\TT'$ of the same width as $\TT$ can be computed in linear time.
\end{proposition}

\subparagraph*{Planar Graphs and Map Graphs.}  
A graph $G$ is {\em planar} if there is a mapping of every vertex of $G$ to a point on the Euclidean plane, and of every edge $e$ of $G$ to a curve on the Euclidean plane where the extreme points of the curve are the points mapped to the endpoints of $e$, and all curves are disjoint except on their extreme points. 

\begin{lemma}[Theorem 7.23 in~\cite{cygan2015parameterized},\cite{Gu2012,ROBERTSON1994323}]
\label{lem:planargridtw}
For any $t\in {\mathbb N}$, every planar graph $G$ of treewidth at least $9t/2$
contains  a $t\times t$ grid minor. Furthermore, for every $\epsilon > 0$, there exists an 
$\OO(n^2)$ time algorithm that  given an $n$-vertex planar graph $G$ and $t\in {\mathbb N}$, either 
outputs a tree decomposition of $G$ of width at most $(9/2 + \epsilon)t$, or constructs a $t\times t$ grid minor 
 in $G$.
\end{lemma}

By substituting $\epsilon=1/3$ in Lemma~\ref{lem:planargridtw}, we get the following corollary. 

\begin{corollary}
\label{cor:planargridtws}
There exists an 
$\OO(n^2)$ time algorithm that given an $n$-vertex planar graph $G$ and $t\in {\mathbb N}$, either 
outputs a tree decomposition of $G$ of width less than $5t$, or constructs a $t\times t$ grid minor  
 in $G$.
\end{corollary}

Map graphs are the intersection graphs of finitely many connected and interior-disjoint regions of the Euclidean plane. 
Any number of regions can meet at a common corner
which results (in the map graph) in 
a clique on the vertices corresponding to these regions. 
Map graphs can be represented as the {\em half-squares of planar bipartite graphs}.  
For a bipartite graph $B$ with bipartition $V(B)=W\uplus U$,  the half-square of $B$ is the graph $G$ with vertex set  $W$ and  edge set  is defined as follows: two vertices in $W$ are adjacent in $G$ if they are at distance $2$ in $B$. 
It is known that the half-square of a planar bipartite graph is a map graph~\cite{ChenGP98,ChenGP02}. 
Moreover, for any map graph $G$, there exists a planar bipartite graph $B$ such that $G$ is a half-square of $B$~\cite{ChenGP98,ChenGP02};  
we refer to such $B$ as a planar bipartite graph {\em corresponding} to the map graph $G$ (see Figure~\ref{fig:mapgraphderviation}). 

Throughout this paper, we assume that any input map graph $G$ 
is  given with a corresponding planar bipartite graph $B$~\footnote{This assumption is made without loss of generality in the sense that if $G$ is given with an embedding instead to witness that it is a map graph, then $B$ is easily computable in linear time~\cite{ChenGP98,ChenGP02}.}. We remark that we consider map graphs as simple 
graphs, that is, there  are no multiple edges between two vertices $u$ and $v$, even if there are two or more  internally vertex disjoint paths of length $2$ between $u$ and $v$ in the corresponding planar bipartite graph.  
For a map graph $G$ with a corresponding planar bipartite graph $B$ having bipartition $V(B)=W\uplus U$, we refer to the vertices in $W=V(G)$ simply as {\em vertices} and the vertices in $U$ as  {\em special vertices}.  Moreover,  we denote the special vertices by $S(G)$. Notice that for any $s\in S(G)$, $N_B(s)$ forms a clique in $G$; we refer to these cliques as {\em special cliques} of $G$.  
We remark that the collection $\cal K$ of cliques mentioned in Section~\ref{sec:intro} refers to $\{N_B(s)\colon s\in S(G)\}\cup \{\{v\}\colon v\in V(G)\}$.

%% file: decomposition.tex

\section{Few Cliques Tree Decomposition of Map Graphs}
\label{sec:fewclde}

In this section, we define a special tree decomposition for map graphs. This decomposition will be derived from 
a tree decomposition of the bipartite planar graph corresponding to the given map graph. 
Once we have defined our new decomposition,  we will gather a few of its structural properties.
These properties will be useful in 
designing  {fast subexponential time} algorithms on map graphs.   

\begin{definition}\label{def:specialTreeDecomp}
Let $G$ be a map graph with a corresponding planar bipartite graph $B$.  
Let ${\cal D}=(T_{{\cal D}},\beta_{{\cal D}})$ be a tree decomposition of $B$ of width less than $\ell$.  
A pair ${\cal D}'=(T_{{\cal D}'},\beta_{{\cal D}'})$  is called the {\em \STDlong{$\ell$} derived from ${\cal D}$}, or simply an {\em \STD{$(\ell,{\cal D})$}},  if it is constructed as follows (see Figure~\ref{fig:TDderivation}).
\begin{enumerate}
\item The tree $T_{{\cal D}'}$ is equal to $T_{{\cal D}}$. Whenever ${\cal D}'$ and ${\cal D}$  are clear from context, we denote both $T_{{\cal D}'}$ and $T_{{\cal D}}$ by $T$.
\item For each node $t\in V(T)$,  $\beta_{{\cal D}'}(t)=(\beta_{{\cal D}}(t)\cap V(G))\cup(\bigcup_{s\in \beta_{{\cal D}}(t)\cap S(G)}N_B(s)\cap \gamma_{{\cal D}}(t))$. That is, for each node $t\in V(T)$, we derive $\beta_{{\cal D}'}(t)$ from $\beta_{{\cal D}}(t)$ by replacing every special vertex $s\in \beta_{{\cal D}}(t)\cap S(G)$ by $N_B(s)\cap \gamma_{{\cal D}}(t)$.
\end{enumerate}
\end{definition}

In  words, the second item states that for every vertex $v\in V(G)$ and node $t\in V(T)$, we have that $v\in \beta_{{\cal D}'}(t)$ if and only if either $(i)$ $v\in \beta_{{\cal D}}(t)\cap V(G)$ or $(ii)$ $v\in N_B(s)$ for some $s\in S(G)\cap \beta_{{\cal D}}(t)$ and $v\in\beta_{\cal D}(t')$ for some node $t'$ in the subtree of $T$ rooted at $t$.

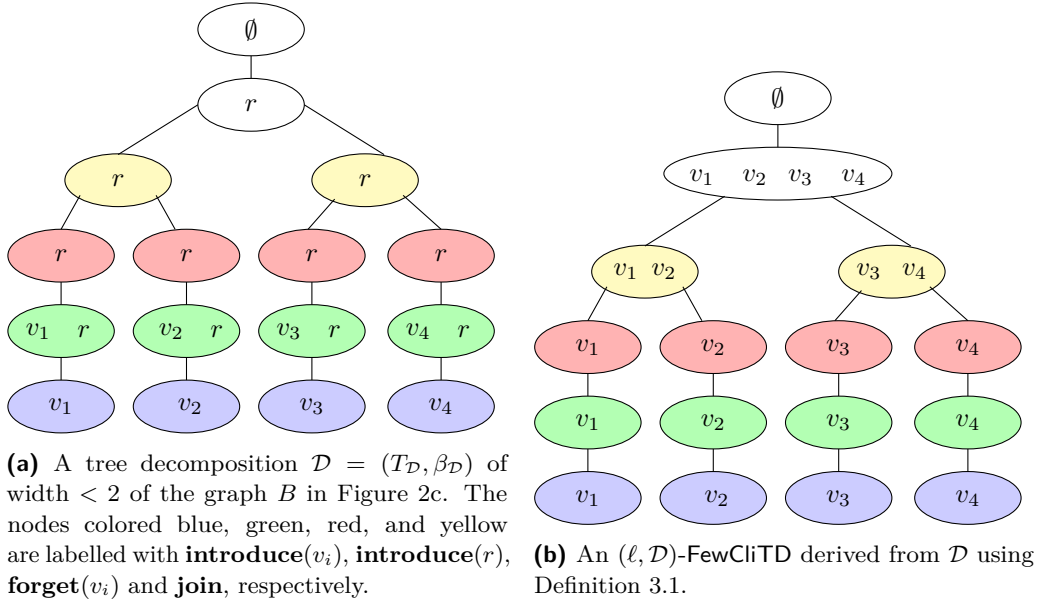
\begin{figure}
    \centering
    \begin{subfigure}[b]{0.47\textwidth}
       \begin{tikzpicture}
       \draw (0,1.35)--(0,1.65);
       \draw (0,2) ellipse (0.7cm and 0.35cm);
       \draw (0,1) ellipse (0.7cm and 0.35cm);
       \draw[fill=yellow!30] (-1.75,0) ellipse (0.7cm and 0.35cm);
       \draw[fill=yellow!30] (1.5,0) ellipse (0.7cm and 0.35cm);
       
       \draw[fill=red!30] (-2.5,-1) ellipse (0.7cm and 0.35cm);
       \draw[fill=red!30] (-0.85,-1) ellipse (0.7cm and 0.35cm);
       \draw[fill=red!30] (0.8,-1) ellipse (0.7cm and 0.35cm);
       \draw[fill=red!30] (2.5,-1) ellipse (0.7cm and 0.35cm);
         \node[]  at (-2.5,-1) (a) {$r$};
         \node[]  at (0.8,-1) (a) {$r$};
         \node[]  at (2.5,-1) (a) {$r$};
         \node[]  at (-.85,-1) (a) {$r$};
         \node[]  at (0,1) (a) {$r$};
         \node[]  at (-1.75,0) (a) {$r$};
         \node[]  at (1.5,0) (a) {$r$};
         \node[]  at (0,2) (a) {$\emptyset$};
        \draw (-2.5,-0.65)--(-2.25,-0.2);
        \draw (2.5,-0.65)--(2,-0.2);
        \draw (-1,-0.65)--(-1.25,-0.2);
        \draw (0.75,-0.65)--(1.1,-0.25);       
        \draw (-1.75,0.34)--(-0.7,1);
        \draw (1.75,0.34)--(0.7,1);
       \draw[fill=green!30] (-2.5,-2) ellipse (0.7cm and 0.35cm);
       \draw[fill=green!30] (-0.85,-2) ellipse (0.7cm and 0.35cm);
       \draw[fill=green!30] (0.8,-2) ellipse (0.7cm and 0.35cm);
       \draw[fill=green!30] (2.5,-2) ellipse (0.7cm and 0.35cm);
       \draw[fill=blue!20] (-2.5,-3) ellipse (0.7cm and 0.35cm);
       \draw[fill=blue!20] (-0.85,-3) ellipse (0.7cm and 0.35cm);
       \draw[fill=blue!20] (0.8,-3) ellipse (0.7cm and 0.35cm);
       \draw[fill=blue!20] (2.5,-3) ellipse (0.7cm and 0.35cm);
       \draw (-2.5,-2.35) -- (-2.5,-2.65);
        \draw (-0.85,-2.35) -- (-0.85,-2.65);
        \draw (0.8,-2.35) -- (0.8,-2.65);
        \draw (2.5,-2.35) -- (2.5,-2.65);        
        \draw (-2.5,-1.35) -- (-2.5,-1.65);
        \draw (-0.85,-1.35) -- (-0.85,-1.65);
        \draw (0.8,-1.35) -- (0.8,-1.65);
        \draw (2.5,-1.35) -- (2.5,-1.65);
        \node[]  at (-2.8,-2) (a) {$v_1$};
         \node[]  at (0.5,-2) (a) {$v_3$};
         \node[]  at (2.2,-2) (a) {$v_4$};
         \node[]  at (-1.05,-2) (a) {$v_2$};
         \node[]  at (-2.2,-2) (a) {$r$};
         \node[]  at (1.1,-2) (a) {$r$};
         \node[]  at (2.8,-2) (a) {$r$};
         \node[]  at (-.45,-2) (a) {$r$};
         \node[]  at (-2.5,-3) (a) {$v_1$};
         \node[]  at (0.8,-3) (a) {$v_3$};
         \node[]  at (2.5,-3) (a) {$v_4$};
         \node[]  at (-.8,-3) (a) {$v_2$};
       \end{tikzpicture}
        \caption{A tree decomposition ${\cal D}=(T_{{\cal D}},\beta_{{\cal D}})$ of width $<2$ of the  graph $B$ in Figure~\ref{fig:corrplanargraph}. The nodes colored blue, green, red, and yellow are labelled with {\bf introduce}($v_i$), {\bf introduce}($r$), {\bf forget}($v_i$) and {\bf join}, respectively.}
        \label{fig:tdB}
    \end{subfigure}
    ~ 
    \begin{subfigure}[b]{0.47\textwidth}
       \begin{tikzpicture}
       \draw (0,1.65)--(0,1.95);
       \draw (0,2.3) ellipse (0.7cm and 0.35cm);
       \draw (0,1.3) ellipse (1.5cm and 0.35cm);
       \draw[fill=red!30] (-2.5,-1) ellipse (0.7cm and 0.35cm);
       \draw[fill=red!30] (-0.85,-1) ellipse (0.7cm and 0.35cm);
       \draw[fill=red!30] (0.8,-1) ellipse (0.7cm and 0.35cm);
       \draw[fill=red!30] (2.5,-1) ellipse (0.7cm and 0.35cm);
       \draw[fill=yellow!30] (-1.75,0) ellipse (0.7cm and 0.35cm);
       \draw[fill=yellow!30] (1.5,0) ellipse (0.7cm and 0.35cm);
         \node[]  at (-2.5,-1) (a) {$v_1$};
         \node[]  at (0.8,-1) (a) {$v_3$};
         \node[]  at (2.5,-1) (a) {$v_4$};
         \node[]  at (-.85,-1) (a) {$v_2$};
         \node[]  at (0,2.3) (a) {$\emptyset$};
         \node[]  at (-1,1.25) (a) {$v_1$};
         \node[]  at (0.3,1.25) (a) {$v_3$};
         \node[]  at (1,1.25) (a) {$v_4$};
         \node[]  at (-0.3,1.25) (a) {$v_2$};
                 \draw (-2.5,-0.65)--(-2.25,-0.2);
        \draw (2.5,-0.65)--(2,-0.2);
        \draw (-1,-0.65)--(-1.25,-0.2);
        \draw (0.75,-0.65)--(1.1,-0.25);       
        \draw (-1.75,0.34)--(-0.7,1);
        \draw (1.75,0.34)--(0.7,1);                
                 \node[]  at (-2,0) (a) {$v_1$};
         \node[]  at (1.2,0) (a) {$v_3$};
                \node[]  at (-1.5,0) (a) {$v_2$};
         \node[]  at (1.8,0) (a) {$v_4$};

       \draw[fill=green!30] (-2.5,-2) ellipse (0.7cm and 0.35cm);
       \draw[fill=green!30] (-0.85,-2) ellipse (0.7cm and 0.35cm);
       \draw[fill=green!30] (0.8,-2) ellipse (0.7cm and 0.35cm);
       \draw[fill=green!30] (2.5,-2) ellipse (0.7cm and 0.35cm);
       \draw[fill=blue!20] (-2.5,-3) ellipse (0.7cm and 0.35cm);
       \draw[fill=blue!20] (-0.85,-3) ellipse (0.7cm and 0.35cm);
       \draw[fill=blue!20] (0.8,-3) ellipse (0.7cm and 0.35cm);
       \draw[fill=blue!20] (2.5,-3) ellipse (0.7cm and 0.35cm);
               \node[]  at (-2.5,-2) (a) {$v_1$};
         \node[]  at (0.8,-2) (a) {$v_3$};
         \node[]  at (2.5,-2) (a) {$v_4$};
         \node[]  at (-.85,-2) (a) {$v_2$};
         \node[]  at (-2.5,-3) (a) {$v_1$};
         \node[]  at (0.8,-3) (a) {$v_3$};
         \node[]  at (2.5,-3) (a) {$v_4$};
         \node[]  at (-.8,-3) (a) {$v_2$};
                \draw (-2.5,-2.35) -- (-2.5,-2.65);
        \draw (-0.85,-2.35) -- (-0.85,-2.65);
        \draw (0.8,-2.35) -- (0.8,-2.65);
        \draw (2.5,-2.35) -- (2.5,-2.65);        
        \draw (-2.5,-1.35) -- (-2.5,-1.65);
        \draw (-0.85,-1.35) -- (-0.85,-1.65);
        \draw (0.8,-1.35) -- (0.8,-1.65);
        \draw (2.5,-1.35) -- (2.5,-1.65);
       \end{tikzpicture}
        \caption{An {\em \STD{$(\ell,{\cal D})$}} derived from ${\cal D}$ using Definition~\ref{def:specialTreeDecomp}.}
        \label{fig:tdG}
    \end{subfigure}
    \caption{Figure~\ref{fig:tdB} represents a tree decomposition of $B$ (drawn in Figure~\ref{fig:corrplanargraph}). In fact it is a tree decomposition obtained after deleting the leaves of a nice tree decomposition.}\label{fig:TDderivation}
\end{figure}

Next, we prove that the \STD{$(\ell,{\cal D})$} $(T,\beta_{{\cal D}'})$ in Definition~\ref{def:specialTreeDecomp} is a tree decomposition of~$G$. We remark that if we replace the term $N_B(s)\cap \gamma_{{\cal D}}(t)$ by the term $N_B(s)$ in the second item of Definition~\ref{def:specialTreeDecomp}, then we  still derive a tree decomposition, but then some of the properties proved later do not hold true.

\begin{lemma}
Let $G$ be a map graph with a corresponding planar bipartite graph $B$.  
Let ${\cal D}'=(T,\beta_{{\cal D}'})$ be an \STD{$(\ell,{\cal D})$} where ${\cal D}=(T,\beta_{{\cal D}})$ is a tree decomposition of $B$ of width less than $\ell$.
Then, ${\cal D}'$ is  a tree decomposition of $G$. 
\end{lemma}
\begin{proof}
We first prove that every vertex of $G$ is present in at least one bag. Towards this, notice that Property  $(a)$ of  Definition~\ref{def:treeDecomp} of the  tree decomposition ${\cal D}$ of $B$ implies that $\bigcup_{t\in V(T)}\beta_{{\cal D}}(t)=V(B)=V(G)\cup S(G)$. 
Therefore, since $\beta_{{\cal D}'}(t)\supseteq \beta_{{\cal D}}(t)\cap V(G)$ for any $t\in V(T)$, we conclude that $\bigcup_{t\in V(T)}\beta_{{\cal D}'}(t)=V(G)$.  Now, we prove that for any edge $\{u,v\}\in E(G)$, there exists a bag $\beta_{{\cal D}'}(t)$ for some $t\in V(T)$ 
such that $\{u,v\}\subseteq \beta_{{\cal D}'}(t)$. Because  $\{u,v\}\in E(G)$, there exists a special vertex $s\in S(G)$ such that $\{u,s\},\{v,s\}\in E(B)$. 
By Property $(c)$  of  ${\cal D}$, the set of nodes $Q=\{t\in V(T)~:~s\in \beta_{{\cal D}}(t)\}$ induces a (connected) subtree of $T$. Let $z\in Q$ be a node such that the distance from $z$ to the root of $T$ is minimized. 
Therefore, the choice of $z$ is unique. 
Since $\{u,s\},\{v,s\}\in E(B)$, 
by Property $(b)$ of ${\cal D}$ and the definition of $Q$, there exist $x,y\in Q$ such that $\{u,s\}\in \beta_{\D}(x)$ and $\{v,s\}\in \beta_{\D}(y)$. Then, because $x$ and $y$ must be descendants of $z$ in $T$, 
we have that $\{u,v\}\subseteq \gamma_{{\cal D}}(z)$. This implies that $\{u,v\}\subseteq \beta_{{\cal D}'}(z)$. So we have proved Properties $(a)$ and $(b)$ of Definition~\ref{def:treeDecomp}. 

To prove Property $(c)$ of Definition~\ref{def:treeDecomp} with respect to $\D'$, we pick an arbitrary vertex $u\in V(G)$, and prove that the set of nodes $R=\{t\in V(T)~:~u\in \beta_{{\cal D}'}(t)\}$ induces a (connected) subtree of $T$. 
Observe 
that $R=O_u\cup (\bigcup_{s\in N_B(u)} R'_s)$
where 
$O_u=\{t\in V(T)~:~u\in \beta_{{\cal D}}(t)\}$, and $R'_s=\{t\in V(T)~:~s\in \beta_{{\cal D}}(t) \mbox{ and } u\in \gamma_{{\cal D}}(t)\}$ for each $s\in N_B(u)$. To prove that $T[R]$ is connected, it is enough to prove that 
$(i)$ $T[O_u]$ is connected, $(ii)$ $R_s'\cap O_u\neq \emptyset$ for all $s\in N_B(u)$, and $(iii)$ 
$T[R'_s]$ is connected for all $s\in N_B(u)$. Statement $(i)$ follows from Property $(c)$  of the  tree decomposition ${\cal D}$ of $B$.  For any $s\in N_B(u)$, since $\{u,s\}\in E(B)$ and by Property $(b)$ of $\D$, we have that $R'_s\cap O_t\neq \emptyset$, and hence Statement $(ii)$ follows. 

The proof of the lemma will be complete with the proof of Statement $(iii)$. Towards this, let $R_s=\{t\in V(T)~:~s\in \beta_{{\cal D}}(t)\}$ for all $s\in N_B(u)$. Clearly $R_s'\subseteq R_s$. By Property $(c)$ of 
${\cal D}$, we know that for any $s\in N_B(u)$, $R_s$ induces a (connected) subtree $T_s$ of $T$. We claim that  for any $s\in N_B(u)$, $T[R'_s]$ (which is a subgraph of $T_s$) is a (connected) subtree of $T_s$. Towards a contradiction, suppose that $T[R'_s]$ is not connected. Then, let $C$ and $C'$ be two connected components of $T[R_s']$ such that there exists a path $P$ in $T_s$ from a vertex in $C$ to a vertex in $C'$ whose internal vertices all belong to  $V(T_s)\setminus R_s'$. Then, there is an internal vertex $w$ of $P$ such that $w$ is an ancestor of one of the end-vertices of $P$ (because in a rooted tree any internal vertex of a path is an ancestor of an end-vertex of the path and $P$ has at least one internal vertex, else $C$ and $C'$ form one connected component). This implies that $u\in \gamma_{{\cal D}}(w)$, because $w\in V(T_s)=R_s$ and $w$ is an ancestor a vertex 
in $R'_s$ (where by the definition of $R_s'$, $u$ must belong to the bag of that vertex). This is a contradiction to the fact that $w\notin R_s'$.  Therefore, we conclude that $T[R'_s]$ is connected. This completes the proof of the lemma. 
\end{proof}

To simplify statements ahead, from now on, we have the following notation. 

\noindent\fbox{%
    \parbox{.98\textwidth}{%
Throughout the section, we fix a map graph $G$, a corresponding planar bipartite graph $B$ of $G$, an integer $\ell\in\mathbb{N}$,  
a nice tree decomposition ${\cal D}$ of $B$ of width less than $\ell$ and 
an {\em \STDlong{$\ell$}} ${\cal D}'$  of $G$ derived from ${\cal D}$ using Definition~\ref{def:specialTreeDecomp}
    }%
}

\smallskip

Recall that $T=T_{\cal D}=T_{{\cal D}'}$ and that for each node $t\in V(T)$, $\beta_{{\cal D}'}(t)$ was obtained from 
$\beta_{\D}(t)$ by replacing every special vertex $s\in S(G)$ with $N_B(s)\cap \gamma_{\cal D}(s)$. 

\begin{definition}
For a node $t\in V(T)$, we use \org($t$) to denote the set $\beta_{{\cal D}}(t)\cap \beta_{{\cal D}'}(t)$,  \fake$(t)$ to denote the set $\beta_{{\cal D}'}(t)\setminus \beta_{{\cal D}}(t)$, and  \cliques$(t)$ to denote 
the set $\{N_B(s)\colon s\in S(G)\cap \beta_{{\cal D}}(t) \}$ of special cliques of $G$.  
\end{definition}
Informally, for a node $t\in V(T)$, \org($t$) denotes the set of vertices of $V(G)$ present in the bag $\beta_{{\cal D}}(t)$, \fake$(t)$ denotes the set of ``new'' vertices added to $\beta_{{\cal D}'}(t)$ while replacing special vertices in $\beta_{{\cal D}}(t)$, and \cliques$(t)$ is the set of special cliques in $G$ that consist of one for each special vertex $s\in \beta_{{\cal D}}(t)$. For example,  let $t$ be the node in Figure~\ref{fig:TDderivation} that is labelled with {\bf forget}($v_1$) by ${\cal D}$. 
Then, \org$(t)=\emptyset$, \fake$(t)=\{v_1\}$ and \cliques$(t)=\{\{v_1,\ldots,v_4\}\}$.

In the remainder  of this section we prove properties related to ${\cal D}$ and ${\cal D}'$, which we use later in the paper to design some of our subexponential-time parameterized algorithms.  Towards the formulation of the first property,  
consider the tree decomposition ${\cal D'}$ in Figure~\ref{fig:TDderivation} and the set of its nodes whose bags contain  the vertex $v_1$  as a ``fake'' vertex.
This set of nodes forms a path with one end-vertex being the unique node $t_{v_1}$ of $T$ labelled with {\bf forget}$(v_1)$ by ${\cal D}$ and the other end-vertex being an ancestor of 
$t_{v_1}$. In fact, the set of nodes $Q=\{t \in V(T) \colon v_1\in \fake(t) \mbox{ and }r\in \beta_{{\cal D}}(t)\}$ forms  the unique path in $T$ from $t_{v_1}$ to $t_r$ where $t_r$ is the unique child of the node labelled with {\bf forget}$(r)$ by  ${\cal D}$. This observation is abstracted and formalized in the following lemma.

%


\begin{lemma}
\label{lem:twstruc1}
Let $v\in V(G)$ and $s\in S(G)$ such that $v\in N_B(s)$ and $Q=\{t \in V(T) \colon v\in \fake(t) \mbox{ and }s\in \beta_{{\cal D}}(t)\}\neq \emptyset$. Let $x$ be the  node in $T$ labelled with {\bf forget}$(v)$ by $\D$, and $y$ be the unique child of the node labelled with {\bf forget}$(s)$ by  ${\cal D}$. Then, 
$y$ is an ancestor of $x$, and  $Q$ induces a path in $T$ which is the unique path between $x$ and $y$ in $T$. 
\end{lemma}

\begin{proof}
First, we prove that $Q$ induces a (connected) subtree of $T$.  
Suppose not. Then, there exist two connected components $C_1$ and $C_2$ of $T[Q]$ such that there exists a path $P$  in $T$ from a vertex in $C_1$ to a vertex in $C_2$ whose internal vertices all belong to $V(T)\setminus Q$. By Property $(c)$ of the tree decomposition ${\cal D}$, 
we have that  $s\in \beta_{{\cal D}}(t)$ for any $t\in V(P)$. 
Moreover, there is an internal vertex $w$ of $P$ such that $w$ is an ancestor of one of the end-vertices of $P$. This implies that $v\in \gamma_{{\cal D}}(w)$, because $v$ belong to  the bags of the endpoints of $P$ (by the definition of $Q$ and $\fake$).  As we have also shown that $s\in \beta_{\D}(t)$ for all $t\in V(P)$, this implies that $w\in Q$, which is a contradiction. 
Hence, we have proved that $T[Q]$ is connected. 

Next, we prove that $T[Q]$ is a path such that one of its endpoints is a descendant of the other. Towards this, it is enough to prove that $(i)$ for any distinct $t,t'\in Q$, either $t$ is a descendant of $t'$ or $t'$ is a descendant of $t$. For the sake of contradiction, assume that there exist $t,t'\in Q$ such that neither $t$ is a descendent of $t'$ nor 
$t'$ is a descendent of $t$. By the definition of $Q$ and because $t,t'\in Q$, we have that $v\in \gamma_{{\cal D}}(t)$ and  $v\in \gamma_{{\cal D}}(t')$. Thus by Property $(c)$ of the tree decomposition ${\cal D}$, we have that  
$v\in \beta_{{\cal D}}(t)$ and $v\in \beta_{{\cal D}}(t')$. Because $v\in \fake(t)$ and $v\in \fake(t')$, this is a contradiction to the definition of $\fake$.  

It remains to prove that $y$ is an ancestor of $x$  and that $x$ and $y$ are endpoints of $T[Q]$. Towards this, recall that $x$ is the node in $T$ labelled with {\bf forget}$(v)$ by the  tree decomposition ${\cal D}$. 
First, we prove that $x$ is an end-vertex of the path $T[Q]$. Let $x'$ be the only child of $x$. 
To prove $x$ is an end-vertex of the path $T[Q]$, it is enough to show that $x\in Q$ and $x'\notin Q$. 
Since $x$ is labelled with {\bf forget}$(v)$ by $\D$, we have that $v\notin \beta_{{\cal D}}(x)$, $v\in \beta_{{\cal D}}(x')$,  and $v\in \gamma_{{\cal D}}(x)$. This implies that $v\in \org(x')$ and hence $x'\notin Q$. Now, we prove that $x\in Q$. For this purpose, let $R=\{t \in V(T) \colon s\in \beta_{{\cal D}}(t)\}$. Clearly, $Q\subseteq R$. By Property $(c)$ of the tree decomposition ${\cal D}$, we have that $T[R]$ is connected. We have already proved that $T[Q]$ is a path and since $Q\subseteq R$, $T[Q]$ is a path in $T[R]$. 
Since $x$ is labelled with {\bf forget}$(v)$ by ${\cal D}$, for any node $x''$ in the subtree rooted at $x$ and $x''\neq x$, by Observation~\ref{obs:forgetsubtree}, either $v\in \beta_{\D}(x'')$ or $v\notin \gamma_{\D}(x'')$. 
This implies that $Q$ contains no node in the subtree of $T$ rooted at $x$ and not equal to $x$. 
Moreover, observe that  there exists a node $x^\star$ in the subtree of $T$ rooted at $x$ such that $\{s,v\}\subseteq \beta_{{\cal D}}(x^\star)$ and hence $x^{\star}\in R$.
Now, since $Q$ is non-empty and $T[Q]$ is connected, we have that $s\in \beta_{{\cal D}}(x)$. Since $v\notin \beta_{{\cal D}}(x)$, $v\in \gamma_{{\cal D}}(x)$ and $s\in \beta_{{\cal D}}(x)$, we conclude that $x\in Q$. 
Thus, we have proved that $x$ is an end-vertex of the the path $T[Q]$.  

Next we prove that $y$ is the other end-vertex of the path $T[Q]$ and $y$ is an ancestor of $x$. Since $y$ is the only child of the node $y'$ labelled with {\bf forget}$(s)$, we have that $s\in \beta_{\cal D}(y)$ and $s\notin \beta_{\cal D}(y')$. 
This implies that $y'\notin Q$. Thus to prove that $y$ is an end-vertex of  the path $T[Q]$, it is enough to prove that $y\in Q$. Since $s\in \beta_{\cal D}(x)$, $s\in \beta_{\cal D}(y)$, $s\notin \beta_{\cal D}(y')$ and $y'$ is the parent of $y$, by Property $(c)$ of ${\cal D}$, 
we have that $y$ is an ancestor of $x$. This also implies that $v\in \gamma_{{\cal D}}(y)$ and $v\notin \beta_{{\cal D}}(y)$. 
Hence, $y\in Q$. This completes the proof of the lemma. 
%
\end{proof}

In the next lemma we show that for any special vertex $s\in S(G)$ and any node $t$  in $T$ labelled with {\bf introduce}$(s)$ by ${\cal D}$, it holds that $t$ and its child carry the ``same information''.

\begin{lemma}
\label{lem:twstruc2}
Let $s\in S(G)$ and $t$ be a node in $T$ labelled with {\bf introduce}$(s)$ by ${\cal D}$. Let $t'$ be the only child of $t$. 
Then, $\org(t)=\org(t')$ and $\fake(t)=\fake(t')$. 
\end{lemma}

\begin{proof}
We know that $\beta_{{\cal D}}(t)\setminus \{s\}=\beta_{{\cal D}}(t')$.  
This implies that $\org(t)=\org(t')$.  To prove that $\fake(t)=\fake(t')$,  
it is enough to show that $\fake(t)\cap N_B(s)=\emptyset$ (because any special vertex $s'$ in $\beta_{\D}(t')$ and not equal to $s$, is also belongs to $\beta_{\D}(t)$). 
Suppose by way of contradiction that  $\fake(t)\cap N_B(s)\neq \emptyset$ and let $u\in \fake(t)\cap N_B(s)$. Then,  there is a descendent $t_1$ of $t$ such that 
$t_1\neq t$ and $u\in \beta_{{\cal D}}(t_1)$. Thus, since $\{u,s\}\in E(B)$ and by Properties $(b)$ and $(c)$ of the tree 
decomposition ${\cal D}$, we get that $\{u,s\}\subseteq \beta_{{\cal D}}(t)$. This  implies that $u\in \org(t)$, which is a contradiction to the 
assumption that   $u\in \fake(t)\cap N_B(s)$. 
\end{proof}


Next, we see a property of  nodes $t\in V(T)$ labelled with {\bf join}. 

\begin{lemma}
\label{lem:twstruc3}
Let $t$ be a node in $T$ labelled with {\bf join} by {\cal D}, and $t_1$ and $t_2$ are its children.  
Then,  $\org(t)=\org(t_1)=\org(t_2)$, 
$\cliques(t)=\cliques(t_1)=\cliques(t_2)$, 
$\fake(t_1)\cap \fake(t_2)=\emptyset$,  and $\fake(t)=\fake(t_1)\cup\fake(t_2)$.
\end{lemma}
\begin{proof}
Since $t$ is a node in $T$ labelled with {\bf join} by ${\cal D}$ and $t_1$ and $t_2$ are its children, we have that $\beta_{{\cal D}}(t)=\beta_{{\cal D}}(t_1)=\beta_{{\cal D}}(t_2)$. 
This implies that $\org(t)=\org(t_1)=\org(t_2)$, 
$\cliques(t)=\cliques(t_1)=\cliques(t_2)$ and $\fake(t)=\fake(t_1)\cup\fake(t_2)$. For any $v\in \fake(t_i), i\in \{1,2\}$, we know that 
$v\notin \beta_{{\cal D}}(t_i)$, but there is a descendent $t'_i$ of $t_i$ such that $v\in \beta_{{\cal D}}(t_i')$. 
Thus, by Property $(c)$ of the tree decomposition ${\cal D}$, 
we have that $v\notin \gamma_{{\cal D}}(t_j)$ where $j\in \{1,2\}\setminus \{i\}$. 
This implies that $\fake(t_1)\cap \fake(t_2)=\emptyset$.  
\end{proof}


    \begin{figure}[t]
\centering
       \begin{tikzpicture}
              \draw (0,1.65)--(0,1.95);
       \draw (0,2.3) ellipse (0.7cm and 0.35cm);
       \draw (0,1.3) ellipse (1.5cm and 0.35cm);
       \draw[fill=red!30] (-2.5,-1) ellipse (0.7cm and 0.35cm);
       \draw[fill=red!30] (-0.85,-1) ellipse (0.7cm and 0.35cm);
       \draw[fill=red!30] (0.8,-1) ellipse (0.7cm and 0.35cm);
       \draw[fill=red!30] (2.5,-1) ellipse (0.7cm and 0.35cm);
       
              \draw[fill=yellow!30] (-1.75,0) ellipse (0.7cm and 0.35cm);
       \draw[fill=yellow!30] (1.5,0) ellipse (0.7cm and 0.35cm);
                \node[]  at (0,2.3) (a) {$\emptyset$};
         \node[]  at (-2.5,-1) (a) {$v_1$};
         \node[]  at (0.8,-1) (a) {$v_3$};
         \node[]  at (2.5,-1) (a) {$v_4$};
         \node[]  at (-.85,-1) (a) {$v_2$};
         \node[]  at (-1,1.25) (a) {$v_1$};
         \node[]  at (0.3,1.25) (a) {$v_3$};
         \node[]  at (1,1.25) (a) {$v_4$};
         \node[]  at (-0.3,1.25) (a) {$v_2$};
                 \draw (-2.5,-0.65)--(-2.25,-0.2);
        \draw (2.5,-0.65)--(2,-0.2);
        \draw (-1,-0.65)--(-1.25,-0.2);
        \draw (0.75,-0.65)--(1.1,-0.25);       
        \draw (-1.75,0.34)--(-0.7,1);
        \draw (1.75,0.34)--(0.7,1);                
                 \node[]  at (-2,0) (a) {$v_1$};
         \node[]  at (1.2,0) (a) {$v_3$};
                \node[]  at (-1.5,0) (a) {$v_2$};
         \node[]  at (1.8,0) (a) {$v_4$};

        \draw (-2.5,-1.35) -- (-2.5,-2.65);
        \draw (-0.85,-1.35) -- (-0.85,-2.65);
        \draw (0.8,-1.35) -- (0.8,-2.65);
        \draw (2.5,-1.35) -- (2.5,-2.65);

       \draw[fill=green!30] (-2.5,-2) ellipse (0.7cm and 0.35cm);
       \draw[fill=green!30] (-0.85,-2) ellipse (0.7cm and 0.35cm);
       \draw[fill=green!30] (0.8,-2) ellipse (0.7cm and 0.35cm);
       \draw[fill=green!30] (2.5,-2) ellipse (0.7cm and 0.35cm);
  
       \draw[fill=blue!20] (-2.5,-3) ellipse (0.7cm and 0.35cm);
       \draw[fill=blue!20] (-0.85,-3) ellipse (0.7cm and 0.35cm);
       \draw[fill=blue!20] (0.8,-3) ellipse (0.7cm and 0.35cm);
       \draw[fill=blue!20] (2.5,-3) ellipse (0.7cm and 0.35cm);
               \node[]  at (-2.5,-2) (a) {$v_1$};
         \node[]  at (0.8,-2) (a) {$v_3$};
         \node[]  at (2.5,-2) (a) {$v_4$};
         \node[]  at (-.85,-2) (a) {$v_2$};
         \node[]  at (-2.5,-3) (a) {$v_1$};
         \node[]  at (0.8,-3) (a) {$v_3$};
         \node[]  at (2.5,-3) (a) {$v_4$};
         \node[]  at (-.8,-3) (a) {$v_2$};
         \node[]  at (5,-1) (a) {{\bf fake introduce}$(v_i)$};
         \node[]  at (4.4,-3) (a) {{\bf introduce}$(v_i)$};
         \node[]  at (2.7,0) (a) {{\bf join}};
         \node[]  at (4.4,-2) (a) {{\bf redundant}};
         \node[]  at (2.1,1.3) (a) {{\bf join}};
         \node[]  at (2.5,2.3) (a) {{\bf forget}$(\{v_1,v_2,v_3,v_4\})$};
       \end{tikzpicture}
        \caption{Labeling the nodes in the \NSTDlong{$2$} ${\cal D}'$ derived from the nice tree decomposition ${\cal D}$ in Figure~\ref{fig:TDderivation}.
        }
        \label{fig:nicetdG}
        \end{figure}
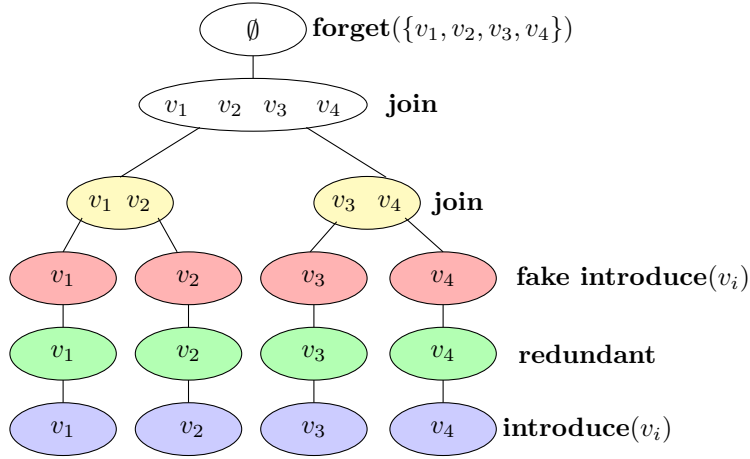

Now, we define a notion of {\em \NSTDlong{$\ell$}} of  
$G$ as the tree decomposition of $G$ derived from a nice tree decomposition ${\cal D}$ of $B$ of width less than $\ell$ (see Definition~\ref{def:specialTreeDecomp}) with additional labeling of nodes. 
In what follows, we describe this additional labeling of nodes.  Towards this, observe that 
because of Lemma~\ref{lem:twstruc2},  for any special vertex $s\in S(G)$ and any node $t\in V(T)$ labelled with {\bf introduce}($s$) by the nice tree decomposition ${\cal D}$, the bags $\beta_{D'}(t)$ and $\beta_{D'}(t')$ carry the ``same information'' where $t'$ is the only child of $t$. Informally, one may choose to handle these nodes by contracting them.  However, to avoid redundant proofs ahead, instead of getting rid of such nodes, we label  them with {\bf redundant} in ${\cal D'}$. 
Next, we explain how to label  other nodes of $T$ in the decomposition ${{\cal D}'}$ (see Figure~\ref{fig:nicetdG}).  To this end, let $t\in V(T)$. 
\begin{itemize}
\item If $t$ is labelled with {\bf leaf} by ${{\cal D}}$, then we label $t$ with {\bf leaf}. Here, $\beta_{{\cal D}'}(t)=\emptyset$. 
\item If $t$ is labelled with {\bf introduce}$(v)$ by ${{\cal D}}$ for some $v\in V(G)$, then we label $t$ with  {\bf introduce}$(v)$. 
In this case, $t$ has only one child $t'$ in $T$ (because any node labelled with {\bf introduce} by ${\cal D}$ has only one child) and $\beta_{{\cal D}'}(t)\setminus \{v\}=\beta_{{\cal D}'}(t')$.  
\item If $t$ is labelled with {\bf forget}$(v)$ by ${{\cal D}}$ for some $v\in V(G)$ and $v\in \fake(t)$, then we label 
$t$ with {\introducefakevertex}$(v)$. 
In this case, $t$ has only one child $t'$  and $\beta_{{\cal D}'}(t)=\beta_{{\cal D}'}(t')$, but $\org(t)=\org(t')\setminus \{v\}$ and $\fake(t)=\fake(t')\cup\{v\}$. 
\item If $t$ is labelled with {\bf forget}$(v)$ by ${{\cal D}}$ for some $v\in V(G)$ and $v\notin \fake(t)$, then we label $t$ with {\bf forget}$(v)$. 
In this case, $t$ has only one child $t'$, $\beta_{{\cal D}'}(t)=\beta_{{\cal D}'}(t')\setminus \{v\}$, $\org(t)=\org(t')\setminus \{v\}$ and $\fake(t)=\fake(t')$. 
\item Suppose $t$ is labelled with {\bf forget}$(s)$ by ${{\cal D}}$ for some $s\in S(G)$. Then, $t$ has only one child $t'$. Here, we label $t$ with {\bf forget}$(\beta_{{\cal D}'}(t')\setminus \beta_{{\cal D}'}(t))$.  In this case, $\fake(t)\subseteq \fake(t')$ and $\org(t)=\org(t')$. 
\item If $t$ is labelled with {\bf join} by ${{\cal D}}$, then we label $t$ with {\bf join}. Let $t_1$ and $t_2$ be the children of $t$. 
Then, $\org(t)=\org(t_1)=\org(t_2)$, $\cliques(t)=\cliques(t_1)=\cliques(t_2)$, 
$\fake(t_1)\cap \fake(t_2)=\emptyset$,  and $\fake(t)=\fake(t_1)\cup\fake(t_2)$. (See Lemma~\ref{lem:twstruc3}).
\item If $t$ is labelled with {\bf introduce}$(s)$ for some $s\in S(G)$, then we label $t$ with 
{\bf redundant} in ${\cal D}'$  
\end{itemize}
This completes the definition of the \NSTDlong{$\ell$} of $G$ derived from $\D$, to which we simply call an  \NSTD{($\ell,{\cal D}$)}.  
Notice that for each node $t$ in $T$,  $\vert \org(t)\vert+\vert \cliques(t)\vert \leq\ell$. That is, for any node $t\in V(T)$, there exist $i,j\in {\mathbb N}$ such that $i+j\leq \ell$, the cardinality of $\org(t)$ is at most $i$, and  the vertices in $\beta_{{\cal D}'}(t)\setminus \org(t)$ were obtained from at most $j$ special cliques.

%
%
%
%
%
%

%

Since the number of nodes with label {\bf forget}$(v)$ in the tree 
decomposition ${\cal D}$ is exactly one for any $v\in V(B)$ (see Observation~\ref{obs:onlyoneforget}), {\em at most one} node in $T$ is labelled with  {\introducefakevertex}$(v)$ in ${{\cal D}'}$. This is formally stated in the following observation.

\begin{observation}
\label{obs:fakeprop}
Let ${\cal D}'=(T,\beta_{{\cal D}'})$ be an \NSTD{$(\ell, {\cal D})$} of a map graph $G$, for some $\ell\in {\mathbb N}$, 
derived from a nice tree decomposition ${\cal D}$ of a corresponding planar bipartite graph of $G$.  
 Let $t\in V(T)$ and $v\in \fake(t)$. Then, 
\begin{itemize}
\item[$(i)$] there is a unique node $t'\in V(T)$ such that $t'$ is labelled with  {\introducefakevertex}$(v)$ in ${{\cal D}'}$,
\item[$(ii)$] $t$ is an ancestor of $t'$ or $t=t'$, and 
\item[$(iii)$] for any node $t''$ in the unique path between $t$ and $t'$, we have that $v\in \fake(t'')$. 
\end{itemize}
\end{observation}

The correctness of Observation~\ref{obs:fakeprop} follows from  Observation~\ref{obs:onlyoneforget} 
and Lemma~\ref{lem:twstruc1}. The discussion  above, along with Corollary~\ref{cor:planargridtws} and Proposition~\ref{prop:nice}, implies the following lemma. 

\begin{lemma}\label{lem:maindecomp}
Given a map graph $G$ with a corresponding planar bipartite graph $B$ and an integer $\ell\in\mathbb{N}$, in time 
$\OO(n^{2})$, one can either correctly conclude that $B$ contains an $\ell \times\ell$ grid as a minor, or compute a 
nice tree decomposition ${\cal D}$ of $B$ of width less than $5\ell$ and a \NSTD{$(5\ell,{\cal D})$} 
of $G$.  
\end{lemma}

Lastly, we prove an important property of a tree decomposition ${\cal D'}$ of a map graph $G$ that is derived from a tree decomposition ${\cal D}$ of  a corresponding bipartite planar graph of $G$.   In particular, the edges considered in the following lemma are precisely those that connect the vertices ``already seen'' (when we use dynamic programming (DP)) with vertices to ``see in the future''.  

\begin{lemma}
\label{lem:edgetwokinds}
For 
any node $t\in V(T)$, the edges with one endpoint in $\gamma_{{\cal D}'}(t)$ and other in $V(G)\setminus \gamma_{{\cal D}'}(t)$ are of 
two kinds. 
\begin{itemize}
\item Edges incident with vertices in $\org(t)$.
\item Edges belonging to some special clique in $\cliques(t)$ (these edges are incident to vertices in $\fake(t)$). 
\end{itemize}
\end{lemma}

\begin{proof}
Fix $t\in V(T)$. Since ${\cal D}'$ is a tree decomposition of $G$, for any edge $e\in E(G)$ with one endpoint in $\gamma_{{\cal D}'}(t)$ and other in $V(G)\setminus \gamma_{{\cal D}'}(t)$, the endpoint of $e$ in $\gamma_{{\cal D}'}(t)$ should belong to $\beta_{{\cal D}'}(t)$. 
Let $u$ be the endpoint of $e$ that belongs to $\beta_{{\cal D}'}(t)$, and $v$ be the other endpoint of $e$.  Notice that the set $\beta_{{\cal D}'}(t)$ is partitioned into $\org(t)$ and $\fake(t)$, so $u$  belongs to either $\org(t)$ or $\fake(t)$, and in the former case we are done. We now assume that $u\in \fake(t)$. 

Since $\{u,v\}=e\in E(G)$, there is a special vertex 
$s\in S(G)$ such that $\{u,s\},\{v,s\}\in E(B)$. If $s\in \beta_{{\cal D}}(t)$, then the edge $\{u,v\}$ belongs to the special clique 
$K=N_B(s)$ in $G$ and $K\in \cliques(t)$.  We claim that indeed $s\in \beta_{{\cal D}}(t)$. Towards this, notice that 
$u\in \fake(t)$. This implies that $u\notin \beta_{{\cal D}}(t)$, but $u$ is present in a bag $\beta_{{\cal D}}(t')$ of 
some descendant $t'$ of $t$. Moreover, since $\{u,s\}\in E(B)$, we further know that $\{u,s\}\subseteq \beta_{{\cal D}}(t_1)$ for some descendent $t_1$ of $t$.
Since $\{v,s\}\in E(B)$ and $v\notin \gamma_{{\cal D}'}(t)$, there is a bag $\beta_{{\cal D}}(t_2)$ such that 
$\{v,s\}\subseteq \beta_{{\cal D}}(t_2)$ and $t_2$ is not a descendent of $t$. Thus, since 
$s\in  \beta_{{\cal D}}(t_1)\cap  \beta_{{\cal D}}(t_2)$, by Property $(c)$ of  the tree decomposition ${\cal D}$, we have that $s\in  \beta_{{\cal D}}(t)$. 
This completes the proof of the lemma. 
\end{proof}




%% file: FVS.tex

\section{Feedback Vertex Set 
}\label{sec:fvs}
In this section, we give  some simple applications of the computation of an \STD{$(\ell,{\cal D})$}  in designing 
subexponential-time parameterized  algorithms on map graphs.  
We exemplify our approach by developing a subexponential-time parameterized algorithm for  \probFVS.
This approach can be used to design subexponential-time parameterized algorithms for the more general {\sc Connected Planar $\cal F$-Deletion} problem (discussed below) as well as {\sc Connected Vertex Cover} and {\sc Connected Feedback Vertex Set}. These simple applications already substantially improve upon the known algorithms for these problems on map graphs~\cite{FominLS12,FominLS17}.
We first prove the following theorem, and then show how the idea of the proof can be generalized to 
{\sc Connected Planar $\cal F$-Deletion}. 

\begin{theorem}\label{thm:fvs}
\probFVS{} on map graphs can be solved in time $2^{\OO(\sqrt{k}\log k)}\cdot n^{\OO(1)}$.
\end{theorem}

The starting point of the algorithm is that the existence of a large grid as a minor in the corresponding planar bipartite graph $B$ or a large clique implies that the given instance is a \No-instance.  Indeed we now observe that if we find a large grid as a minor in $B$, then we can find many vertex disjoint cycles in the map graph, and hence we can answer \No.

\begin{observation}\label{obs:fvsLargeGrid}
Let $(G,B, k)$ be an instance of  \probFVS. If $B$ contains a $3(\sqrt{k}+1)\times 3(\sqrt{k}+1)$ grid as a minor, then $(G,B,k)$ is a \No-instance. 
\end{observation}

\begin{proof}
From the existence of a $3\sqrt{k}\times 3\sqrt{k}$ grid as a minor in $B$, we can conclude that $B$ contains $k$ vertex-disjoint cycles 
of length at least $8$ each. For any cycle $C$ of length $\ell\geq 6$ in $B$, there is a cycle of length $\ell/2$ in $G$ whose set of vertices is a subset of $V(C)$. 
This implies that if $B$ has $k$ vertex-disjoint cycles of length at least $8$ each, then there are $k$ vertex-disjoint 
cycles in $G$.  
\end{proof}

This observation leads us to the following lemma.

\begin{lemma}
\label{lem:fvsFirstPhase}
There is an algorithm that given an instance  $(G,B,k)$ of \probFVS, runs in time $\OO(n^2)$, and 
either correctly concludes that the minimum size of a feedback vertex set of $G$ is more than $k$, or outputs a nice tree decomposition ${\cal D}$ of $B$ and  a  \NSTD{$(15(\sqrt{k}+1),{\cal D})$}  ${\cal D}'$ of $G$ such that  for each $t\in V(T)$, 
$\vert \beta_{\D'}(t)\vert \leq 15 (k+2)(\sqrt{k}+1)$ and $\beta_{\D'}(t)$ is a union of $15(\sqrt{k}+1)$ many cliques of size at most $k+2$ each.
\end{lemma}
\begin{proof}
If $\vert N_B(s)\vert  \geq k+3$ for some special vertex $s\in S(G)$, then $G$ has a clique of size $k+3$, and hence 
$(G,B,k)$ is a \No\ instance.  Thus, we now suppose that this is not the case. Now, we apply Lemma~\ref{lem:maindecomp} with $\ell=3(\sqrt{k}+1)$. If the output is 
a $3(\sqrt{k}+1) \times 3(\sqrt{k}+1)$ grid minor of $B$, then by Observation~\ref{obs:fvsLargeGrid}, 
$(G,B,k)$ is a \No-instance. Otherwise, we have a nice tree decomposition $\D$ of 
$B$ of width less than $15(\sqrt{k}+1)$ and a \NSTDlong{$(15(\sqrt{k}+1))$}  $\D'$ of $G$. 
In this case, since  $\vert N_B(s)\vert \leq k+2$ for every $s\in S(G)$, we have that for any $t\in V(T)$, $\vert \beta_{\D'}(t)\vert \leq 15(k+2)(\sqrt{k}+1)$. The bound on the number of cliques follows from the width of $\D$. 
\end{proof}
%


Because of  Lemma~\ref{lem:fvsFirstPhase}, to prove Theorem~\ref{thm:fvs}, we can focus on 
\probFVS{} on map graphs where the input is accompanied with a \NSTDlong{$(15(\sqrt{k}+1))$}   $\D'$  of $G$ such that  for each $t\in V(T)$, 
$\vert \beta_{\D'}(t)\vert \leq 15 (k+2)(\sqrt{k}+1)$ and $\beta_{\D'}(t)$ is a union of $15(\sqrt{k}+1)$ many cliques of size at most $k+2$ each.   The proof of Theorem~\ref{thm:fvs} is by a dynamic programming (DP) algorithm using the fact that for any $t\in V(T)$, $\beta_{\D'}(t)$ is a union of $15(\sqrt{k}+1)$ many cliques of size at most $k+2$. Observe that for any $t\in V(T)$, any feedback vertex set must contain all but two vertices in each clique. Thus, for each clique we have at most $\OO(k^2)$ 
choices of which vertices of the clique belong to a solution. Briefly, for any node  $t\in V(T)$, a subset $S\subseteq \beta_{\D'}(t)$ 
such that $S$ contains all but at most $2$ vertices from each clique in the bag $\beta_{\D'}(t)$, a partition ${\cal P}$ of $\beta_{\D'}(t)\setminus S$ and $k'\leq k$, we have DP table entry ${\cal A}[t,S,{\cal P},k']$ which stores a Boolean value. The table entry ${\cal A}[t,S,{\cal P},k']$ is set to $1$ if and only if there is a feedback vertex set $F$ of $G[\gamma_{\D'}(t)]$ of size $k'$ such that $F\cap \beta_{\D'}(t)=S$ and for any block $P$ of ${\cal P}$, all the vertices of $P$ belong to a connected component of $G[\gamma_{\D'}(t)]-F$. 
Notice that the cardinality of $\beta_{\D'}(t)\setminus S$ is upper bounded by $\OO(\sqrt{k})$. 
This allows us to bound the number of states by 
$2^{\OO(\sqrt{k} \log k)}$.
 After this observation the dynamic programming is identical to the one made for \probFVS on graphs of bounded treewidth. See the book~\cite{cygan2015parameterized} for further details on the dynamic programming algorithm for \probFVS on graphs of bounded treewidth. 
 
By following similar lines as in the case of the above algorithm for \probFVS, we can design subexponential-time parameterized  algorithms for {\sc Connected Feedback Vertex Set}  and  {\sc Connected Vertex Cover} on map graphs. 
 

\begin{theorem} \label{thm:confvsvc}
{\sc Connected Feedback Vertex Set}  and  {\sc Connected Vertex Cover}  on map graphs can be solved in time $2^{\OO(\sqrt{k}\log k)}\cdot n^{\OO(1)}$.
\end{theorem}


Our algorithm for \probFVS can be generlized to a large class of problems, namely, the class of 
{\sc Connected Planar $\cal F$-Deletion} problems. In this class, each problem is defined by family $\cal F$ of connected graphs that contains  at least one planar graph.
Here, the input is a graph $G$ and an integer parameter $k$. The goal is to find a set $S$ of size at most $k$ such that $G - S$ does not contain any of the graphs in $\cal F$ as a minor. This definition captures problems such as {\sc Vertex Cover}, {\sc Feedback Vertex Set}, {\sc \probTWdel}, \probPWdel, \probTDepthdel, {\sc Diamond Hitting Set} and {\sc Outerplanar Vertex Deletion}. Theorem~\ref{thm:fvs} can be generalized to the following general theorem. 

%
\begin{theorem} \label{thm:fdeletion}
 Every {\sc Connected Planar $\cal F$-Vertex Deletion} problem  on map graphs can be solved in time $2^{\OO(\sqrt{k}\log k)}\cdot n^{\OO(1)}$.
\end{theorem}

Similarly to \probFVS, we can prove that there is a constant $c$ (depending only on ${\cal F}$) such that if there is a $c\sqrt{k}\times c\sqrt{k}$-grid minor in $B$, then the given instance is a \No-instance. Moreover, if there is a clique of size at least $k+d+1$ in $G$, where $d$ is the size of the smallest graph in ${\cal F}$, then also the given instance is a \No-instance. These two arguments imply that there is an algorithm which given an instance  $(G,B,k)$ of {\sc Connected Planar $\cal F$-Vertex Deletion}, runs in time $\OO(n^2)$, and 
either correctly concludes that $(G,B,k)$ is a \No-instance or outputs a \NSTDlong{$\OO(\sqrt{k})$} 
$\D'$  of $G$ such that  for each $t\in V(T)$, 
$\vert \beta_{\D'}(t)\vert \leq \OO(k^2)$ and $\beta_{\D'}(t)$ is a union of $\OO(\sqrt{k})$ many cliques of size at most $k+d$ each. 
As in the case of \probFVS, any solution of {\sc Connected Planar ${\cal F}$-Deletion} contains all but at most $d-1$ vertices from any clique. Thus, for each clique of size $k'$ in $\beta_{\D'}(t)$ we have at most $\OO(\binom{k'}{k'-d-1})=\OO(k^d)$ (because $k'\leq k+d+1$) 
choices of which vertices of the clique belong to a solution. 
This allows us to bound the number of ``states'' by $(k^d)^{\OO(\sqrt{k})}=2^{\OO(\sqrt{k} \log k)}$. After this observation the dynamic programming is identical to the one made for {\sc Connected Planar $\cal F$-Vertex Deletion} on graphs of bounded treewidth. That is, given a tree decomposition of width $w$, there is an algorithm solving {\sc Connected Planar ${\cal F}$-Vertex Deletion} in time $2^{\OO(w \log w)}n^{\OO(1)}$~\cite{baste_et_al:LIPIcs:2018:8555}. Following this algorithm with our observation results in an algorithm with time complexity $2^{\OO(\sqrt{k}\log k)}\cdot n^{\OO(1)}$ for {\sc Connected Planar $\cal F$-Vertex Deletion} problem  on map graphs. 


%% file: cycle.tex

\section{Longest Cycle}\label{sec:exactCyc}


In the last section we saw  simple applications of the computation of an \STD{$(\ell,{\cal D})$}
on map graphs. In this section as well as Section~\ref{sec:cycPack},  we will see more involved applications of \STD{$(\ell,{\cal D})$}. Specially, in this section we prove that \probKCycle\ admits a subexponential-time parameterized algorithm on map graphs.  

\begin{theorem}
\label{thm:cycle}
\probKCycle\  on map graphs  
can be solved  in $2^{\OO(\sqrt k\log k)} \cdot n^{\OO(1)}$ time. 
\end{theorem}

Towards the proof of Theorem~\ref{thm:cycle}, we prove that if there is a solution (i.e., a cycle of length at least $k$), then there is one for which a ``sublinear crossing lemma'' holds. Informally, the sublinear crossing lemma asserts the existence of a solution such that at any separator (bag) of a given \NSTD{$(\ell,{\cal D})$}, the number of edges crossing the separation  is $\OO(\sqrt{k})$.  This lemma lies at the heart of the proof and is one of the main technical contributions of the paper. 



Towards proving Theorem~\ref{thm:cycle}, we design an algorithm that given a map graph $G$ along with a corresponding bipartite planar graph $B$ and $k\in {\mathbb N}$,  
runs in time $2^{\OO(\sqrt k\log k)}n^{\OO(1)}$ and decides whether $G$ has a cycle of length at least $k$. 
Notice that if there 
is a special vertex $s\in S(G)$ such that $\vert N_B(s)\vert \geq k$, then $G$ has a cycle of length at least $k$, 
because  $N_B(s)$ forms a clique in $G$. 
Moreover, observe that if there is a ``large enough'' grid in $B$, then we can answer \Yes. 
%
%
These observations lead to the following lemma.

\begin{lemma}
\label{lem:cyclealgo1step}
There is an algorithm that given an instance  $(G,B,k)$ of \probKCycle, runs in time $\OO(n^2)$, and 
either correctly concludes that $G$ has a cycle of length at least $k$, or outputs a nice tree decomposition ${\cal D}$ 
of $B$ of width $<5\sqrt{2k}$ and a \NSTD{$(5\sqrt{2k},{\cal D})$}  $\D'$ of $G$ such that  for each node $t\in V(T)$, 
$\vert \beta_{\D'}(t)\vert \leq 5\sqrt{2} \cdot k^{1.5}$.  
\end{lemma}
\begin{proof}
As mentioned before, if $\vert N_B(s)\vert \geq k$ for some special vertex $s\in S(G)$, then $G$ has a cycle of length $k$. 
Thus, we now suppose that this is not the case. 
Now, we apply Lemma~\ref{lem:maindecomp} with $\ell=\sqrt{2k}$. If the output is 
a $\sqrt{2k} \times \sqrt{2k}$ grid minor of $B$, 
then $B$ has a cycle of length at least $2k$, and this implies that $G$ has a cycle of length at least $k$. 
Otherwise, we have a nice tree decomposition $\D$ of 
$B$ of width less than $5\sqrt{2k}$ and a nice \NSTD{$(5\sqrt{2k},{\cal D})$}  $\D'$ of $G$. 
In this  case, since  $\vert N_B(s)\vert < k$ for every $s\in S(G)$, we have that for  each node $t\in V(T)$, 
$\vert \beta_{\D'}(t)\vert \leq 5\sqrt{2} \cdot k^{1.5}$.   
\end{proof}
Because of Lemma~\ref{lem:cyclealgo1step}, to prove Theorem~\ref{thm:cycle}, it is enough to prove the following 
lemma. 

\begin{restatable}{lemma}{cyclealgores}
\label{lem:cyclealgo}
There is an algorithm that 
given an instance $(G,B,k)$ of \probKCycle,  and  a \NSTD{$(5\sqrt{2k},{\cal D})$}  $\D'$ of $G$ (derived from a nice tree decomposition $\D$ of $B$) such that  for each node $t\in V(T)$, 
$\vert \beta_{\D'}(t)\vert \leq 5\sqrt{2} \cdot k^{1.5}$, 
runs in time $2^{\OO(\sqrt k\log k)} \cdot n^{\OO(1)}$, and outputs a longest cycle in $G$. 
\end{restatable}

Towards proving Lemma~\ref{lem:cyclealgo}, the main ingredient is to prove the following 
claim: if $G$ has a cycle of length $\ell$, then there is a cycle $C$ of length $\ell$, with 
the following property. 
\begin{framed}
\noindent 
For each node $t\in V(T)$, the number of edges of $E(C)$ with one endpoint in $\beta_{\D'}(t)$ and the other in 
$V(G)\setminus \gamma_{\D'}(t)$ is upper bounded by $\OO(\sqrt k)$.      
\end{framed}
The above mentioned property is encapsulated in the following sublinear crossing lemma.  

\begin
{restatable}
[Sublinear Crossing Lemma]{lemma}{subcrosslem}
\label{lem:cyclecross}
Let $(G,B,k)$ be an instance of \probKCycle. Let ${\cal D}$ be a nice tree decomposition of $B$ and ${\cal D}'$ be a  \NSTD{$(5\sqrt{2k},{\cal D})$}  of $G$. 
Let $C$ be a cycle 
in $G$. 
Then there is a cycle $C'$ of the same length as $C$ such that  for any node $t\in V(T)$, the 
number of edges in $E(C')$ with one endpoint in $\beta_{\D'}(t)$ and the other in $V(G)\setminus 
\gamma_{\D'}(t)$ is at most $20\sqrt{2k}$.  
\end{restatable}

%

Towards proving   Lemma~\ref{lem:cyclecross}, we first prove the following lemma. 

\begin{restatable}{lemma}{lemsubcrosscycle}
\label{lem:cyclecrossone}
Let $(G,B,k)$ be an instance of \probKCycle. 
Let ${\cal D}$ be a nice tree decomposition of $B$ and ${\cal D}'$ be a  \NSTD{$(5\sqrt{2k},{\cal D})$}  of $G$. 
Let $C$ be a cycle 
in $G$ and $K$ be a special clique in $G$. Then, there is a cycle $C'$ of the same length as $C$ such that 
$E(C')\setminus E(K)=E(C)\setminus E(K)$ and for any node $t\in V(T)$, the 
number of edges of $E(C')\cap E(K)$ with one endpoint in $\fake(t)\cap K$ and the other in $V(G)\setminus 
\gamma_{\D'}(t)$ is at most $4$.  
\end{restatable}

Before formally proving Lemma~\ref{lem:cyclecrossone}, we give a high level overview of the proof 
and an auxiliary lemma which we use in the proof of Lemma~\ref{lem:cyclecrossone}. The proof idea is to change 
the edges of $E(K)\cap E(C)$ in $C$ (because in Lemma~\ref{lem:cyclecrossone} our objective is to bound the ``crossing edges'' from a subset of $E(K)$ for each node $t\in V(T)$) to obtain a new cycle $C'$ of the same length as $C$ that satisfies the following property: $(i)$ for any node $t\in V(T)$, the number of edges of $E(C')\cap E(K)$ with one endpoint in $\fake(t)\cap K$ and the other in $V(G)\setminus \gamma_{\D'}(t)$ is at most $4$. For the ease of presentation, assume that $K\subseteq V(C)$. Now, consider the graph ${\cal P}$ obtained from the cycle $C$ after deleting the edges in $E(K)$. Without loss of generality assume that $E(C)\cap E(K)\neq \emptyset$. Otherwise,  Lemma~\ref{lem:cyclecrossone} is true where $C'=C$. We consider  ${\cal P}$ as a collection of 
vertex-disjoint paths where the end-vertices of the paths are in $K$. Some paths in  ${\cal P}$ may be of length $0$. 
Let $Z$ be the set of end-vertices of the paths in ${\cal P}$. Clearly, $Z\subseteq K$. We will ``complete'' the collection of paths ${\cal P}$ to a cycle by adding edges from $E(Z)$ satisfying Statement $(i)$.  Any cycle $C'$ with $V(C')=V({\cal P})$ has the same length as $C$.  So, all the work that is required for us is to  complete the collection of paths ${\cal P}$ to a cycle by adding edges from $E(Z)$ satisfying Statement $(i)$. Towards that, let $\widehat{\sigma}=v_1,\ldots,v_{k'}$ be an arbitrary sequence of vertices in $Z$.  We show (in Claim~\ref{claim:lesscross}) that $(ii)$ there is a subset of edges $F\subseteq E(Z)$ such that $E({\cal P})\cup F$ forms a cycle $C'$ with vertex set $V({\cal P})$ and  for any $j\in [k']$, the number of edges in $F$ with one endpoint in $\{v_1,\ldots,v_{j}\}$ and the other in $\{v_{j+1},\ldots,v_{k'}\}$ is at most $2$. This implies that for any $1\leq i\leq j\leq k'$, the number of edges in $F$ with one endpoint in $\{v_i,\ldots,v_j\}$ and the other in $Z\setminus \{v_i,\ldots,v_j\}$ is at most $4$.   In the light of Statement $(ii)$, our aim will be to prove that $(iii)$ for any node $t\in V(T)$, there exist $1\leq i\leq j\leq k'$ such that $\fake(t)\cap Z\subseteq \{v_i,\ldots,v_j\}$ and no vertex in $Z\setminus \gamma_{\D'}(t)$ belongs to $\{v_i,\ldots,v_j\}$. Then, Statement $(i)$ will follow (because edges of $C'$ incident with vertices in $K\setminus Z$ are from $E(G)\setminus E(K)$ and will not be counted in Statement $(i)$).  In fact, we will prove that there a sequence $\sigma$ on $Z$ (derived from a postorder transversal of $T$) such that Statement $(iii)$ is true (see Claim~\ref{claim:segment}).  The proof of Statement $(ii)$ is encapsulated in the following lemma (which we will use in the proof of Lemma~\ref{lem:cyclecrossone}). 

\begin{lemma}
\label{lem:stamentA}
Let $\ell\geq 3$ be an integer. Let $u_1,\ldots,u_{\ell}$ be a sequence of vertices in a graph $H$ where $X=\{u_1,\ldots,u_{\ell}\}$ is a clique in $H$. Let 
${\cal Q}$ be a family of vertex disjoint paths in $H$ (which possibly contains paths of length $0$) such that each $v\in X$ is an end-vertex of a path in ${\cal Q}$ and  
$E({\cal Q})\cap E(X)=\emptyset$.
Then, there is a set $F\subseteq E(X)$ 
such that the following conditions hold. 
\begin{itemize}
\item[$(a)$] $E({\cal Q})\cup F$ forms a cycle containing all the vertices of $V({\cal Q})$, 
\item[$(b)$] For any $j\in [\ell]$, the 
number of edges in $F$  with one endpoint in $\{u_1,\ldots,u_j\}$ and the other in $\{u_{j+1},\ldots,u_{\ell}\}$ is at most $2$. 
\item[$(c)$] If the degree of $u_1$ is one in ${\cal Q}$ (i.e., $u_1$ is an end-vertex of a path in ${\cal Q}$), then the 
number of edges in $F$  with  $u_1$ as an endpoint  is exactly $1$. 
\end{itemize}
\end{lemma}

\begin{proof}
We prove the lemma using induction on the length of the sequence $\ell$. 
By slightly abusing the notation, we also use ${\cal Q}$ as a subgraph of $H$ where each connected component is a path in the family ${\cal Q}$. 
Notice that for any vertex $u\in X$, $d_{{\cal Q}}(u)\in \{0,1\}$. 
Additionally,  notice that for any set $F\subseteq E(H)$ such that  $E({\cal Q})\cup F$ forms a cycle containing all the vertices of $V({\cal Q})$, if $d_{{\cal Q}}(u_1)=1$, then the number of edges in $F$  with $u_1$ as an endpoint is exactly $1$ (because the 
degree of each vertex in a cycle is $2$).  So to prove the lemma, it is enough to prove  
conditions $(a)$ and $(b)$ of the lemma.

The base case is when $\ell=3$. Towards the proof of the base case, suppose that ${\cal Q}=\{[u_1],[u_2],[u_3]\}$. Then,  the set of edges $F=\{\{u_1,u_2\},\{u_2,u_3\},\{u_3,u_1\}\}$ 
is a set as required to  satisfy  the lemma. Otherwise, ${\cal Q}=\{[xy],[z]\}$, where $\{x,y,z\}=\{u_1,u_2,u_3\}$. 
Then, $F=\{\{x,z\},\{y,z\}\}$  is a set as required to  satisfy  the lemma.

Now, we consider the induction step. 
For this purpose, we assume the lemma for any sequence of length at most $\ell-1$, and consider a sequence of length  $\ell>3$.  The proof consist of three cases as follows, depending on the degrees of $u_1$ and $u_2$ in ${\cal Q}$.

\medskip
\noindent
{\bf Case 1: $d_{\cal Q}(u_1)=0$ and $d_{\cal Q}(u_2)=1$.} 
Let $A=\{\{u_1,u_2\}\}$. Let ${\cal Q}'=(V({\cal Q}),E({\cal Q})\cup A)$ and $X'=X\setminus \{u_2\}$. 
The subgraph ${\cal Q}'$ is a collection of vertex disjoint paths  in $H$ such that each $v\in X'$ is an end-vertex of a path in ${\cal Q}'$ and $d_{{\cal Q}'}(u_1)=1$. Thus, by induction hypothesis, there is a set $F'\subseteq E(X')$ such that $(i)$ 
$E({\cal Q}')\cup F'$ forms a cycle containing all the vertices of $V({\cal Q}')$,  $(ii)$ for any $i\in [\ell]\setminus \{1\}$, the 
number of edges in $F'$  with one endpoint in $\{u_1,\ldots,u_i\}\setminus \{u_2\}$ and the other in $\{u_{i+1},\ldots,u_{\ell}\}\setminus \{u_2\}$ is at most $2$,  and $(iii)$  the number of edges in $F'$  with $u_1$ as an endpoint   is exactly $1$.
We claim that $F=F'\cup \{\{u_1,u_2\}\} $ is the required set of edges. Since $E({\cal Q})\cup F=E({\cal Q}')\cup F'$, 
by $(i)$, $E({\cal Q})\cup F$ forms a cycle and it uses all the vertices in $V({\cal Q})$ because $\{u_1,u_2\}\in F$. 
For any $j\in [\ell]\setminus \{1\}$, the 
edges in $F$  with one endpoint in $\{u_1,\ldots,u_j\}$ and other in $\{u_{j+1},\ldots,u_{\ell}\}$ are also edges  in $F'$,  
and $\{u_2,u_r\}\notin F'$ for all $r\in [\ell]$. Thus,  
by $(ii)$, condition $(b)$ of the statement holds for $j\in [\ell]\setminus \{1\}$. 
Lastly, notice that the number of edges in $F$  with $u_1$ as one endpoint  is exactly~$2$.

\medskip
\noindent
{\bf Case 2: $d_{\cal Q}(u_1)=d_{{\cal Q}}(u_2)=0$.} 
Let $A=\{\{u_1,u_2\},\{u_1,u_3\}\}$. 
Let ${\cal Q}'=(V({\cal Q}),E({\cal Q})\cup A)$.
%
The subgraph ${\cal Q}'$ is a collection of vertex disjoint paths  in $H$.  
Let $X'$ is the set of end-vertices of paths in ${\cal Q}'$.  Clearly, $X'\subseteq X$. 
Since the degree of $u_1$ in ${\cal Q}'$ is $2$, we have 
that $u_1\notin X'$ and hence $\vert X'\vert <\vert X\vert=\ell$. Since $d_{{\cal Q}}(u_2)=0$  
and only one edge in $A$ is incident with $u_2$, we have that $d_{{\cal Q}'}(u_2)=1$ and hence $u_2\in X'$. 
More precisely, $X'=X\setminus \{u_1\}$ if $d_{{\cal Q}'}(u_3)=1$ (equivalently $d_{{\cal Q}}(u_3)=0$) 
and $X'=X\setminus \{u_1,u_3\}$ otherwise. 
Thus, by induction hypothesis, there is a set $F'\subseteq E(X')$ such that $(i)$ 
$E({\cal Q}')\cup F'$ forms a cycle containing all the vertices of $V({\cal Q}')$,  $(ii)$ for any $i\in [\ell]\setminus \{1\}$, the 
number of edges in $F'$  with one endpoint in $\{u_2,\ldots,u_i\}$ and the other in $\{u_{i+1},\ldots,u_{\ell}\}$ is at most $2$,  
and $(iii)$  the number of edges in $F'$  with $u_2$ as an endpoint   is exactly $1$.
We claim that $F=F'\cup \{\{u_1,u_2\},\{u_1,u_3\}\} $ is the required set of edges. Since $E({\cal Q})\cup F=E({\cal Q}')\cup F'$, 
by $(i)$, $E({\cal Q})\cup F$ forms a cycle and it uses all the vertices in $V({\cal Q})$ because $\{u_1,u_2\}\in F$. 
For any $j\in [\ell]\setminus \{1,2\}$, the 
edges in $F$  with one endpoint in $\{u_1,\ldots,u_j\}$ and the  other in $\{u_{j+1},\ldots,u_{\ell}\}$ are also edges  in $F'$, and hence  
by $(ii)$, condition $(b)$ of the statement holds for $j\in [\ell]\setminus \{1,2\}$. 
%
By $(iii)$ and the fact that $F=F' \cup \{\{u_1,u_2\},\{u_1,u_3\}\}$, 
we have that  the 
number of edges in $F$  with one endpoint in $\{u_1,u_2\}$ and the other in $\{u_{3},\ldots,u_{\ell}\}$ is exactly $2$.
Lastly, notice that the number of edges in $F$  with one endpoint $u_1$  is exactly $2$ 
(these edges are $\{u_1,u_2\}$ and $\{u_1,u_3\}$).

\medskip
\noindent
{\bf Case 3: $d_{\cal Q}(u_1)=1$.}  Let $P$ be the path in ${\cal Q}$ such that $u_1$ is its end-vertex and let $z$ be the other end-vertex of $P$. Let $x$ be the first vertex in $u_2,\ldots,u_{\ell}$ that is not equal to $z$. That is, $x=u_2$ if $z\neq u_2$ and $x=u_3$ if $z= u_2$. 
Let $A=\{\{u_1,x\}\}$ and ${\cal Q}'=(V({\cal Q}),E({\cal Q})\cup A)$. Notice that $d_{{\cal Q}'}(u_1)=2$ and $d_{{\cal Q}'}(x)\in \{1,2\}$.  
If $d_{{\cal Q}'}(x)=1$, then denote $X'=X\setminus \{u_1\}$, and otherwise denote $X'=X\setminus \{u_1,x\}$.  The subgraph ${\cal Q}'$ is a collection of vertex disjoint paths  in $H$ such that  each $v\in X'$ is an end-vertex of a path in ${\cal Q}'$.  
Thus, by induction hypothesis, there is a set $F'\subseteq E(X')$ such that $(i)$ 
$E({\cal Q}')\cup F'$ forms a cycle containing all the vertices of $V({\cal Q}')$,  $(ii)$ for any $i\in [\ell]$, the 
number of edges in $F'$  with one endpoint in $X'\cap \{u_1,\ldots,u_i\}$ and the other in $X'\cap \{u_{i+1},\ldots,u_{\ell}\}$ is at most $2$, 
and $(iii)$  if 
$u_2\in X'$, then  $d_{{\cal Q}'}(u_2)=1$ and the number of edges in $F'$  with $u_2$ as an endpoint   is exactly $1$. We claim that $F=F'\cup \{\{u_1,x\}\} $ is the required set of edges. 
Since $E({\cal Q})\cup F=E({\cal Q}')\cup F'$, 
by $(i)$, $E({\cal Q})\cup F$ forms a cycle and it uses all the vertices in $V({\cal Q})$ because $\{u_1,x\}\in F$. 
Since $x\in \{u_2,u_3\}$ and $F=F'\cup \{\{u_1,x\}\}$, we have that for any $j\in [\ell]\setminus \{1,2\}$, the 
edges in $F$  with one endpoint in $\{u_1,\ldots,u_j\}$ and the other in $\{u_{j+1},\ldots,u_{\ell}\}$ are also edges in $F'$ and hence by $(ii)$,  condition $(b)$ holds for $j\in [\ell]\setminus \{1,2\}$. 
 If $x=u_2$, then the number of edge in $F$ with one endpoint in $\{u_1,u_2\}$ and the other in $\{u_3,\ldots,u_{\ell}\}$ is at most $1$ (because $d_{\cal Q}(u_1)=1$ and $\{u_1,u_2\}\in F$). 
If $x\neq u_2$, then $z=u_2$ and hence 
the number of edge in $F$ with one endpoint in $\{u_1,u_2\}$ and the other in $\{u_3,\ldots,u_{\ell}\}$ is exactly $2$.
So, condition $(b)$ holds for $j=2$. 
Lastly, notice that the number of edges in $F$  with $u_1$ as an endpoint   is exactly $1$ (because $d_{\cal Q}(u_1)=1$).  

This completes the proof of of the lemma. 
\end{proof}

Next, we move to a formal proof of Lemma~\ref{lem:cyclecrossone}. For the convenience of the reader we restate the lemma. 

\lemsubcrosscycle*

\begin{proof}[Proof of Lemma~\ref{lem:cyclecrossone}]
Without loss of generality, assume that $K\subseteq V(C)$. Otherwise, we can consider the statement of 
the lemma for cycle $C$ in the graph $G'=G-(K\setminus V(C))$ and special clique $K \cap V(C)$ of $G'$.  
We also assume that $E(C)\cap E(K) \neq \emptyset$, else the correctness is trivial because we can take $C'$ as $C$. 

Recall that $T=T_{\cal D}=T_{{\cal D}'}$. 
Let $\pi'$ be a postorder transversal of the  nodes in the rooted binary tree $T$, and let 
$\pi$ be the restriction of $\pi'$ where we only keep the nodes that are labeled with {\introducefakevertex}$(v)$ for some $v\in K$. Denote $\pi=t_1,\ldots,t_{k''}$ such that  each $t_i$, $i\in [k'']$, is labelled with    {\introducefakevertex}$(x_i)$ where $x_i\in K$. Notice that $\bigcup_{t\in V(T)}\fake(t)\cap K=\{x_1,\ldots,x_{k''}\}$ 
(by Observation~\ref{obs:fakeprop}). Let $\sigma_1$ be the sequence $x_1,\ldots,x_{k''}$ and $U=\{x_1,\ldots,x_{k''}\}$. 
Let $\sigma_2$ be a fixed arbitrary sequence of $K\setminus U$, i.e., all the vertices of $K$ that are never ``fakely introduced''. 
Let $\sigma$ be the sequence which is a concatenation of $\sigma_1$ and $\sigma_2$. 

Let ${\cal P}=(V(C),E(C)\setminus E(K))$. That is, ${\cal P}$ is the graph obtained by deleting edges of $E(K)$ from the cycle $C$. Notice that each connected component of ${\cal P}$ is a path (may be of length $0$) with end-vertices in $K$.
Let $Z$ be the set of end-vertices of the paths in ${\cal P}$. Notice that 
for any vertex $u\in K\setminus Z$, both edges of $C$ incident with $u$ are from $E(C)\setminus E(K)$ (see the left part of Figure~\ref{figureK}). That is, $E(C)\setminus E(K)=E(C)\setminus E(Z)=E({\cal P})$.  
Since we seek a cycle $C'$ in which $E(C')\setminus E(K)=E(C)\setminus E(K)$, 
no edge of $C'$ incident with $u$ for any vertex $u\in K\setminus Z$, is in $E(K)$. That is, all the edges of $E(C')\cap E(K)$ will belong to $E(Z)$. 
This leads to the  following simple observation.

\begin{observation}
\label{obs:Zenough}
Let $C'$ be a cycle in $G$ such that $E(C')\setminus E(K)=E(C)\setminus E(K)$
and $t\in V(T)$.  The number of edges of $E(C')\cap E(K)$ with one endpoint in 
$\fake(t)\cap K$ and the other in $V(G)\setminus \gamma_{\D'}(t)$ is equal to the  number of edges of $E(C')\cap E(Z)$ with one endpoint in $\fake(t)\cap Z$ and the other in $V(G)\setminus \gamma_{\D'}(t)$.
\end{observation}

  \begin{figure}[t]
   \centering
\begin{tikzpicture}[scale=1.4]

\node[draw, circle, scale=0.7] (a1) at (-1.8,0) {{$v_1$}};
\node[draw, circle, scale=0.7] (a2) at (-1.1,0) {{$v_2$}};
\node[draw, circle, scale=0.7] (a3) at (-0.4,0) {{$v_3$}};
\node[draw, circle, scale=0.7] (a4) at (0.3,0) {{$v_4$}};
\node[draw, circle, scale=0.7] (a5) at (1,0) {{$v_5$}};
\node[draw, circle, scale=0.7] (a6) at (1.7,0) {{$v_6$}};
\node[draw, circle, scale=0.7] (a7) at (2.4,0) {{$v_7$}};

\draw [thick,green] (a2) to[out=-80,in=180] (-0.65,-0.5) to [out=0, in=-90]  (a3);
\draw [thick,green] (a1) to[out=-95,in=180] (0,-1.2) to [out=0, in=-90]  (a7);

\draw [thick,green] (a5) to[out=-60,in=180] (1.45,-0.5) to [out=0, in=-90]  (a6);

\draw [thick,green] (a5) to[out=-80,in=180] (1.65,-0.75) to [out=0, in=-120]  (a7);

\draw [thick,red] (a3) to[out=60,in=180] (0.5,0.5) to [out=0, in=120]  (a6);

\draw [thick,red] (a1) to[out=-80,in=180] (-0.5,-0.85) to [out=0, in=-90]  (a4);

\draw [thick,red] (a2) to[out=-100,in=180] (-0.5,-0.65) to [out=0, in=-100]  (a4);

\node[draw, circle, scale=0.7] (b1) at (3.2,0) {{$v_1$}};
\node[draw, circle, scale=0.7] (b2) at (3.9,0) {{$v_2$}};
\node[draw, circle, scale=0.7] (b3) at (4.6,0) {{$v_3$}};
\node[draw, circle, scale=0.7] (b4) at (5.3,0) {{$v_4$}};
\node[draw, circle, scale=0.7,fill=yellow] (b5) at (6,0) {{$v_5$}};
\node[draw, circle, scale=0.7] (b6) at (6.7,0) {{$v_6$}};
\node[draw, circle, scale=0.7,fill=yellow] (b7) at (7.4,0) {{$v_7$}};

\draw [thick,green] (b2) to[out=-80,in=180] (4.35,-0.5) to [out=0, in=-90]  (b3);
\draw [thick,green] (b1) to[out=-80,in=180] (5,-1) to [out=0, in=-90]  (b7);

\draw [thick,green] (b5) to[out=-60,in=180] (6.45,-0.5) to [out=0, in=-90]  (b6);

\draw [thick,green] (b5) to[out=-80,in=180] (6.65,-0.75) to [out=0, in=-120]  (b7);

\draw [thick,blue] (b1) -- (b2);
\draw [thick,blue] (b3) -- (b4);
\draw [thick,blue] (b4) to[out=80,in=180] (6,0.5) to [out=0, in=90]  (b6);


\end{tikzpicture}
\caption{Left part illustrates a cycle $C$ interacting with a special clique $K=\{v_1,\ldots,v_7\}$. The red curves represent edges in $E(C)\cap E(K)$ and green curves represent paths in $C$ with endpoints in $K$ and (at least one) internal vertices in $V(G)\setminus K$. Thus, ${\cal P}$ is the collection of paths that is a union of  the set of two ``green'' paths ($(v_2-v_3)$ and $(v_1-v_7-v_5-v_6)$) and $\{[v_4]\}$. 
%
Here $Z=\{v_1,v_2,v_3,v_4,v_6\}$. Any edge of $E(C)$ incident with $v_5$ and $v_7$ (i.e., vertices in $K\setminus Z$) are from $E(C)\setminus E(K)$. 
The right part illustrates the proof of Claim~\ref{claim:lesscross}. The edges of $E(C')\setminus E(C)$ mentioned in the proof of Claim~\ref{claim:lesscross} are colored blue.}
\label{figureK}
 \end{figure}
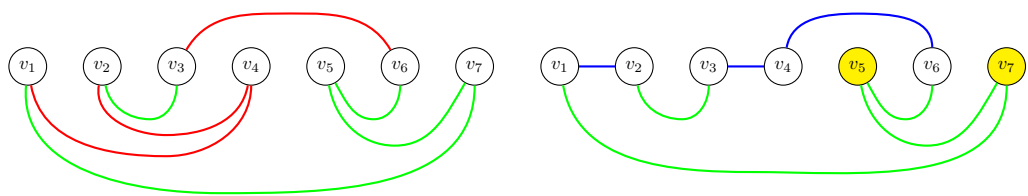


Let $Z=\{v_1,\ldots,v_{k'}\}$ and $\sigma'=\sigma|_Z=v_1,\ldots,v_{k'}$.
%
%
%
The main ingredients of the proof  are 
the following two claims.

\begin{claim}
\label{claim:lesscross}
There is a cycle $C'$ of the same length as $C$ such that $(i)$ $E(C')\setminus E(Z)=E(C)\setminus E(Z)$,  and $(ii)$ for any $j\in [k']$, the number of edges of $E(C')\cap E(Z)$ with one endpoint in 
$\{v_1,\ldots,v_j\}$ and the other in $\{v_{j+1},\ldots,v_{k'}\}$ is at most $2$.  
\end{claim}

\begin{proof}
Clearly when $k'\leq 2$, $\vert E(Z)\vert\leq 1$ and $C'=C$ satisfies the conditions of the claim. 
To prove the claim for $k'\geq 3$, 
we apply Lemma~\ref{lem:stamentA}.  
Recall that $\{v_1,\ldots,v_{k'}\}=Z\subseteq K$ and hence $Z$ forms a clique in $G$.  Additionally, recall that ${\cal P}$ is a collection of paths 
such that $Z$ is the set of end-vertices of the paths in ${\cal P}$. 
Thus, we apply Lemma~\ref{lem:stamentA} for the sequence $v_1,\ldots,v_{k'}$ of vertices in $G$ and 
family of paths ${\cal P}$. Then, by Lemma~\ref{lem:stamentA}, there is a subset $F\subseteq E(Z)$ such that 
$(a)$ $E({\cal P})\cup F$ forms a cycle $C'$ containing all the vertices of $V({\cal Q})$ and $(b)$ for any $j\in [k']$, the number of edges in $F$ with one endpoint in $\{v_1,\ldots,v_j\}$ and the other in $\{v_{j+1},\ldots,v_{k'}\}$ is at most $2$. Since $V(C)=V({\cal P})=V(C')$, the lengths of cycles $C'$ and $C$ are same. 
Because of statement $(a)$ and $E({\cal P})=E(C)\setminus E(K)=E(C)\setminus E(Z)$, we have that $E(C')\setminus E(Z)=E(C)\setminus E(Z)$. Finally, condition $(ii)$ in the claim follows from statement $(b)$. This completes the proof of the claim.
\end{proof}



Recall that for a sequence $\sigma'=u_1u_2\ldots u_{\ell}$ and any $1\leq i\leq j\leq \ell$, the sequence 
$\sigma'' =u_i\ldots u_j$  is called a segment of $\sigma'$.

\begin{claim}
\label{claim:segment}
For any node $t\in V(T)$, there is a segment $\sigma''$ of $\sigma'$ such that 
each vertex in $\fake(t)\cap Z$ appears in $\sigma''$,  and each vertex in 
$Z\setminus \gamma_{\D'}(t)$
does not appear in $\sigma''$.  
\end{claim}

\begin{proof}
Fix a node $t\in V(T)$. 
Recall that $\sigma=\sigma_1\sigma_2$  and $\sigma'=\sigma|_Z$.
Here, the set of vertices present in $\sigma_1$  is  $U=(\bigcup_{t\in V(T)}\fake(t)\cap K) \supseteq (\bigcup_{t\in V(T)}\fake(t)\cap Z)$  (because $Z\subseteq K$), and 
no vertex in $\sigma_2$ is from $U$. 
This implies that all the vertices of $\fake(t)\cap Z$ are in the sequence $\sigma_1$.  
That is, the sequence $\sigma''$ we seek is also a sequence of $\sigma_1|_Z$ and this is the reason we defined 
$\sigma$ to be $\sigma_1\sigma_2$.  
Thus, to prove the claim it is enough to prove that 
there is a segment $\sigma_1'$ of $\sigma_1|_Z$ such that 
each vertex in $\fake(t)\cap Z$ appears in $\sigma_1'$ 
and each vertex in 
$Z\setminus \gamma_{\D'}(t)$ 
does not appear in $\sigma_1'$.

Recall that $\sigma_1=x_1\ldots x_{k''}$ is obtained from the sequence $\pi=t_1,\ldots,t_{k''}$. In turn, recall that $\pi$ is the restriction of 
the postorder transversal $\pi'$  of $T$, where for each $i\in [k'']$, $t_i$ is labelled with {\introducefakevertex}$(x_i)$ for $x_i\in K$. 
Let $W_t$ be the nodes of  the subtree of $T$ rooted at $t$,  
and $$V_t=\{v\in K\colon \mbox{there is  $t'\in W_t$ such that $t'$ is labelled with {\introducefakevertex}}(v)\}.$$ 

The vertices in $W_t$ appear consecutively in $\pi$. Thus, we can let $\pi_t$ be the minimal segment of $\pi$ that  
contains all the nodes in $V_t$.  
Let $i,j\in [k'']$ be such that $\pi_t=t_i,\ldots,t_j$. Now, we define $\sigma_t$ be the segment $x_i,\ldots,x_j$ of $\sigma_1$.   
Now we prove the claim. By conditions $(i)$ and $(ii)$ in Observation~\ref{obs:fakeprop},  $\fake(t)\cap Z\subseteq V_t$.  Clearly, no vertex in $Z\setminus \gamma_{\D'}(t)$  is in $V_t$. 
This implies that each vertex in $\fake(t)\cap Z$ appears in 
$\sigma_t$ and no vertex from $Z\setminus \gamma_{\D'}(t)$  appears in $\sigma_t$. In turn, this implies that $\sigma_t|_Z$ 
is the required segment $\sigma''$ of $\sigma'=\sigma|_Z$. 
\end{proof}

Now, having the above two claims, we are ready to prove the lemma. 
By Claim~\ref{claim:lesscross}, we have that there is a cycle $C'$ such that 
$(i)$ $E(C')\setminus E(Z)=E(C)\setminus E(Z)$, and $(ii)$ for any $j\in [k']$, the number of edges of $E(C)\cap E(Z)$ with one endpoint in $\{v_1,\ldots,v_j\}$ and other in $\{v_{j+1},\ldots,v_{k'}\}$ is at most $2$. By Claim~\ref{claim:segment}, we know that for any $t\in V(T)$, 
there is a segment $\sigma''$ of $\sigma'$ such that 
all vertices in $\fake(t)\cap Z$ appear in a segment $\sigma''$ and no vertex from 
$Z\setminus \gamma_{\D'}(t)$ appears in $\sigma''$. That is, there exist $i,j\in [k']$ such that 
$\fake(t)\cap Z\subseteq \{v_i,\ldots,v_j\}$ and $(Z\setminus \gamma_{\D'}(t))\cap \{v_i,\ldots,v_j\}=\emptyset$. 
 Therefore, by $(ii)$, 
 the number of edges of $E(C')\cap E(Z)$ with one endpoint in $\fake(t)\cap Z$ and the other in $V(G)\setminus \gamma_{\D'}(t)$ is at most $4$.  Then, by Observation~\ref{obs:Zenough}, the proof of the lemma is complete. 
\end{proof}

Now we are ready to prove Lemma~\ref{lem:cyclecross}.  We restate the lemma below. 

\subcrosslem*


\begin{proof}[Proof of Lemma~\ref{lem:cyclecross} ]
Let $S(G)=\{s_1,\ldots,s_{\ell}\}$, and let $K_i=N_B(s_i)$ for all $i\in [\ell]$.  For any $i\in [\ell]$, 
let $Z_i=\bigcup_{j\in [i]} K_j$ and $F_i=\bigcup_{j\in [i]} E(K_j)$. 
We remind that  $T=T_{\D}=T_{\D'}$. 
Towards the proof of the lemma we first 
prove the following claim using induction on $i$. 
\begin{claim}
\label{claim:inductionfakeedges}
Let $S$ be a cycle 
in $G$. Then, for any $i\in [\ell]$, there is a cycle $S_i$ of the same length as $ S$ such that  
$E(S_i)\setminus F_i= E(S)\setminus F_i$, and for any $t\in V(T)$, the 
number of edges in $E(S_i)\cap F_i$ with one endpoint in $\fake(t)\cap Z_i$ and the other in $V(G)\setminus 
\gamma_{\D'}(t)$ is at most $4r$ where $r= \vert \beta_{\D}(t)\cap \{s_1,\ldots,s_i\}\vert$. 
\end{claim}
\begin{proof}
The base case is when $i=1$, and it follows from Lemma~\ref{lem:cyclecrossone} 
(by substituting $C=S$ and $K=K_1$). Now, we consider the induction step for $i>1$. 
By induction hypothesis, we have that the claim is true for $i-1$. That is,  there is a cycle 
$S_{i-1}$ of the same length as $S$  such that  $(i)$ $E(S_{i-1})\setminus F_{i-1}= E(S)\setminus F_{i-1}$, and $(ii)$ for any $t\in V(T)$, the number of edges in $E(S_{i-1})\cap F_{i-1}$ with one endpoint in $\fake(t)\cap Z_{i-1}$ and the other in $V(G)\setminus \gamma_{\D'}(t)$ is at most $4r'$ where $r'= \vert \beta_{\D}(t)\cap \{s_1,\ldots,s_{i-1}\}\vert$.
Now we apply Lemma~\ref{lem:cyclecrossone} (by substituting $C=S_{i-1}$ and $K=K_i$). Then, there is a cycle $S_i$ of the same length as $ S_{i-1}$ such that 
$(a)$ $E(S_i)\setminus E(K_i)=E(S_{i-1})\setminus E(K_i)$, and $(b)$ for any $t\in V(T)$, the 
number of edges in $E(S_i)\cap E(K_i)$ with one endpoint in $\fake(t)\cap K_i$ and the other in $V(G)\setminus 
\gamma_{\D'}(t)$ is at most $4$.  

Now, we prove that $S_i$ satisfies the conditions in the claim. We begin by proving that  $E(S_i)\setminus F_i= E(S)\setminus F_i$. 
\begin{eqnarray*}
E(S)\setminus F_i&=&E(S)\setminus (F_{i-1}\cup E(K_i))\\
&=&(E(S)\setminus F_{i-1})\setminus E(K_i)\\
&=&(E(S_{i-1})\setminus F_{i-1})\setminus E(K_i) \qquad\qquad\qquad (\mbox{By $(i)$})\\
&=&(E(S_{i-1})\setminus E(K_i))\setminus F_{i-1} \\
&=&(E(S_{i})\setminus E(K_i))\setminus F_{i-1} \qquad\qquad\qquad\quad (\mbox{By $(a)$})\\
&=&E(S_i)\setminus F_i \qquad\qquad\qquad\qquad\qquad\qquad\;\; (\mbox{Because }F_i=F_{i-1}\cup E(K_i))
\end{eqnarray*}

Next, we prove that for any $t\in V(T)$, the 
number of edges in $E(S_i)\cap F_i$ with one endpoint in $\fake(t)\cap Z_i$ and the other in $V(G)\setminus 
\gamma_{\D'}(t)$ is at most $4r$ where $r= \vert \beta_{\D}(t)\cap \{s_1,\ldots,s_i\}\vert$.
Fix a node $t\in V(T)$. First, suppose $s_i\notin \beta_{\D}(t)$. Then, by $(ii)$, we have that  
the number of edges in $E(S_{i-1})$ with one endpoint in $\fake(t)\cap Z_{i-1}=\fake(t)\cap Z_{i}$ and the other in $V(G)\setminus \gamma_{\D'}(t)$ is at most $4r'$ where  $r'= \vert \beta_{\D}(t)\cap \{s_1,\ldots,s_{i-1}\}\vert=\vert \beta_{\D}(t)\cap \{s_1,\ldots,s_i\}\vert=r$. Moreover, since $E(S_i)\setminus E(K_i)=E(S_{i-1})\setminus E(K_i)$, we have that the number of edges in $E(S_{i})$ with one endpoint in $\fake(t)\cap Z_{i}$ and the other in $V(G)\setminus 
\gamma_{\D'}(t)$ is at most $4r$.

Second, suppose $s_i\in \beta_{\D}(t)$. Then, by $(ii)$, we have that  
the number of edges in $E(S_{i-1})\cap F_{i-1}$ with one endpoint in $\fake(t)\cap Z_{i-1}$ and the other in $V(G)\setminus \gamma_{\D'}(t)$ is at most $4r'$ where 
$r'= \vert \beta_{\D}(t)\cap \{s_1,\ldots,s_{i-1}\}\vert$.  
By $(a)$, we have that $E(S_i)\setminus E(K_i)=E(S_{i-1})\setminus E(K_i)$, and by $(b)$, we have that  
the  number of edges of $E(S_i)\cap E(K_i)$, with one endpoint in $\fake(t)\cap K_i$ and other in $V(G)\setminus 
\gamma_{\D'}(t)$ is at most $4$.  
Thus, we have that 
the number of edges in $E(S_{i})\cap (F_{i-1}\cup E(K_i))=E(S_{i})\cap F_i$ with one endpoint in $\fake(t)\cap Z_{i}$ and the other in $V(G)\setminus 
\gamma_{\D'}(t)$ is at most $4(r'+1)$,  
and $r'+1= \vert \beta_{\D}(t)\cap \{s_1,\ldots,s_{i}\}\vert$.  This completes the proof of the claim. 
\end{proof}

By applying Claim~\ref{claim:inductionfakeedges} with $S=C$,  we get that there is a cycle $C'=S_{\ell}$ of the same length as $C$ such that  $(iii)$
for any $t\in V(T)$, the 
number of edges in $E(C')\cap F_{\ell}=E(C')$ with one endpoint in $\fake(t)\cap Z_{\ell}=\fake(t)$ and the other in $V(G)\setminus 
\gamma_{\D'}(t)$ is at most $4r$ where $r= \vert \beta_{\D}(t)\cap \{s_1,\ldots,s_{\ell}\}\vert=\vert \beta_{\D}(t)\cap S(G)\vert$. 

We claim that 
$C'$ has the required property. Towards the proof, fix a node $t\in V(T)$. By Lemma~\ref{lem:edgetwokinds}, 
we know that for any edge $e$ with one endpoint $u$ in $\beta_{\D'}(t)$ and the other in $V(G)\setminus 
\gamma_{\D'}(t)$, we have that either $u\in \org(t)$ or $e$ belongs to some special clique $K\in \cliques(t)$. 
By $(iii)$, the 
number of edges of $C'$ with one endpoint in $\fake(t)$ and the other in $V(G)\setminus 
\gamma_{\D'}(t)$ is at most $4r$ where $r= \vert \beta_{\D}(t)\cap S(G)\vert$. 
Notice that, 
since $C'$ is a cycle,  the 
number of edges of $C'$ with one endpoint in $\org(t)$ and the other in $V(G)\setminus 
\gamma_{\D'}(t)$ is at most $2\cdot \vert \org(t)\vert$. That is, the number of 
edges of $C'$ with one endpoint in $\beta_{\D'}(t)$ and the other in $V(G)\setminus 
\gamma_{\D'}(t)$ is at most $2\cdot \vert \org(t)\vert + 4r\leq 4(\vert \org(t)\vert+r)=4\vert \beta_{\D}(t)\vert=20\sqrt{2k}$, because $\D$ is a tree decomposition of width $<5\sqrt{2k}$.  This completes the proof of the lemma. 
\end{proof}


Now we ready to give a proof sketch of Lemma~\ref{lem:cyclealgo}.  For the convenience of the reader we restate the lemma.  


\cyclealgores*

\begin{proof}[Proof Sketch of Lemma~\ref{lem:cyclealgo}] 
Recall that we are given an instance $(G,B,k)$ of \probKCycle, a nice tree decomposition $\D$ of $B$, and  a \NSTD{$(5\sqrt{2k},{\cal D})$}   $\D'$ of $G$, such that  for each $t\in V(T)$,  $\vert \beta_{\D'}(t)\vert \leq 5\sqrt{2} \cdot k^{1.5}$. 
Lemma~\ref{lem:cyclecross} ensures that if $(G,B,k)$ is  a\Yes\ instance of \probKCycle, then there is a cycle $C$ of length at least $k$ 
such that  for any $t\in V(T)$, the number of edges with one endpoint in $\beta_{\D'}(t)$ and the other in $V(G)\setminus 
\gamma_{\D'}(t)$ is at most $20\sqrt{2k}$.  We give a dynamic programming (DP) algorithm, called ${\cal A}$, to find a cycle satisfying properties described in  
Lemma~\ref{lem:cyclecross}. 

Algorithm ${\cal A}$ is a DP algorithm over  the given \NSTD{$(5\sqrt{2k},{\cal D})$}  $\D'$ of $G$. 
For any node $t\in V(T)$, we define $G_{t}$ as the induced subgraph $G[\gamma_{\D'}(t)]$ of $G$. 
Let ${\cal C}$ be the set of maximum length cycles in $G$ such that for any node  $t\in V(T)$, the number of edges with one endpoint in $\beta_{\D'}(t)$ and the other in $V(G)\setminus \gamma_{\D'}(t)$ is at most $20\sqrt{2k}$. 
This allows us to keep only $2^{\OO(\sqrt k\log k)}$ states for any node in our DP  
algorithm. 
Algorithm ${\cal A}$ will construct a cycle $C\in {\cal C}$.  
For a set $Q$ of paths (of length $0$ or more) and cycles, define $\widehat{Q} =\{\{u,v\} \colon \mbox{ there is a $u$-$v$ path $P$ in $Q$}\}$. Let $C\in {\cal C}$.  
 For any $t\in V(T)$, define 
${C}_t$ to be the set of connected components when we restrict 
$C$ to $G_t$. That is, each element in $C_t$ is a path (maybe of length $0$) or $C$ itself (in that case $C_t = \{C\}$).
We also use $C_{t}$ to denote the subgraph $G_t[E({C})]$ of $G_t$. 
Notice that $\bigcup_{Y\in \widehat{C}_t} Y$  is the set of vertices of degree $0$ or $1$  
in  ${C}_t$ and  $\bigcup_{Y\in \widehat{C}_t} Y \subseteq \beta_{\D'}(t)$ 
(recall that ${\widehat{C}_t}=\{\{u,v\} \colon \mbox{ there is a $u$-$v$ path $P$ in $C_t$}\}$).  
We know that the number of edges with one endpoint in $\beta_{\D'}(t)$ and other in $V(G)\setminus 
\gamma_{\D'}(t)$ is at most $20\sqrt{2k}$. This implies that the cardinality of $\bigcup_{P\in \widehat{C}_t} P$ 
is at most $20\sqrt{2k}$.
In our DP algorithm,  we will have state indexed by $(t,\widehat{C}_t,\vert E(C_t) \vert)$,  which 
will be set to $1$.   
Formally, for any $t \in V(T)$, $\ell \in [n]$ and a family $\ZZ$ of vertex disjoint sets of size at most $2$ of $\beta_{\D'}(t)$ with the property that the cardinality of $\bigcup_{Z\in \ZZ} Z$ is at most  $20\sqrt 2k$, we will have a  table entry ${\cal A}[t, \ZZ, \ell]$. For each $t \in V(T)$, we maintain the following correctness invariant. 

\medskip
\noindent
{\bf Correctness Invariant:} $(i)$ 
For every $C\in {\cal C}$,  
${\cal A}[t,\widehat{C}_t, \vert E(C_t) \vert]=1$, 
$(ii)$ for any 
family $\ZZ$ of vertex-disjoint sets of size at most $2$ of $\beta_{\D'}(t)$ with  $0<\vert \bigcup_{Z\in \ZZ} Z \vert \leq 20\sqrt{2k}$, 
$\ell\in [n]$, and ${\cal A}[t,\ZZ,\ell]=1$, there is a set ${\cal Q}$ of $\vert \ZZ \vert$ vertex-disjoint paths in $G_t$
where the endpoints of each path are specified by a set in $\ZZ$ and $\vert E(\QQ)\vert=\ell$, 
and $(iii)$  if ${\cal A}[t,\emptyset ,\ell]=1$, then there is a cycle of length $\ell$ 
in $G_t$. 

The correctness of the our algorithm will follow from the correctness invariant.  The way we fill the table entries  is similar to the way it is done for DP algorithms over graphs of bounded treewidth. That is, we fill the table entries by considering various cases for bags (introduce, forget and join) and using the previously computed table entries. This part of our algorithm is similar to the algorithm for \probKCycle in~\cite{subexpudg} on a so called {\em special path decomposition}.  
\end{proof}

Theorem~\ref{thm:cycle} follows from 
Lemmata~\ref{lem:cyclealgo1step} and \ref{lem:cyclealgo}. 
The algorithm for \probKPath\ goes along the same lines as \probKCycle. Let $(G,B,k)$ be an instance of 
\probKPath. We first apply Lemma~\ref{lem:maindecomp} with $\ell=\sqrt{2k}$ and if we get a $\sqrt{2k}\times \sqrt{2k}$ grid minor 
of $B$, then we conclude that $G$ has a path of length $k$. Otherwise, we construct  a \NSTD{$(5\sqrt{2k},{\cal D})$} $\D'$ of $G$ where $\D$ is a nice tree decomposition of $B$,  and guess two 
end-vertices $u$ and $v$ of a path of length $k$ in $G$ (assuming it exists). Then, we add $u$ and $v$ to all the bags of $\D'$ as 
{\em original} vertices and let $G'=(V(G),E(G)\cup \{\{u,v\}\})$. Next, to prove the existence of a path of length $k$ in $G$, it is enough to check 
the existence of a cycle of length at least $k$ in $G'$ using the tree decomposition $\D'$ (where we added $\{u,v\}$ to all bags). This can be done by using Lemma~\ref{lem:cyclealgo}.

\begin{theorem}
\label{thm:path}
\probKPath\  on map graphs  
can be solved  in $2^{\OO(\sqrt k\log k)} \cdot n^{\OO(1)}$ time. 
\end{theorem}

%% file: cycle-packing.tex

\section{Cycle Packing}\label{sec:cycPack}
In this section, we prove that \probCycPacking admits a subexponential-time parameterized algorithm on map graphs. That is, we prove the following.

\begin{theorem}\label{thm:cycPack}
\probCycPacking on map graphs can be solved in time $2^{\OO(\sqrt{k}\log k)}\cdot n^{\OO(1)}$.
\end{theorem}

Let $(G,B,k)$ be an instance of \probCycPacking. Our first observation is that if $B$ has a large grid minor, then  $(G,B,k)$ is 
\Yes\ instance.

\begin{observation}\label{obs:cycPackLargeGrid}
Let $(G,B, k)$ be an instance of \probCycPacking on map graphs. If $B$ contains a $3\sqrt{k}\times 3\sqrt{k}$ grid as a minor, then $G$ has $k$ vertex-disjoint cycles.
\end{observation}
\begin{proof}
From a $3\sqrt{k}\times 3\sqrt{k}$ grid minor of $B$, we can conclude that $B$ contains $k$ vertex-disjoint cycles 
of length at least $8$ each. For any cycle of length $\ell\geq 6$ in $B$, there is a cycle of length $\ell/2$ in $G$. 
This implies that if $B$ has $k$ vertex-disjoint cycles of length at least $8$ each, then there are $k$ vertex-disjoint 
cycles in $G$.  
\end{proof}

Additionally, notice that if $\vert N_B(s)\vert \geq 3k$ for some $s\in S(G)$, then $G$ has $k$ vertex-disjoint cycles of length $3$ each,  because $N_B(s)$ forms a clique in $G$.  This fact along with Observation~\ref{obs:cycPackLargeGrid} leads to the following lemma. 


\begin{lemma}
\label{lem:cycPackFirstPhase}
There is an algorithm that given an instance  $(G,B,k)$ of \probCycPacking on map graphs, runs in time $\OO(n^2)$, and either correctly concludes that $(G,B,k)$ is a \Yes\ instance, or outputs 
a nice tree decomposition $\D$ of $B$ of width less than $15\sqrt{k}$ and a \NSTD{$(15\sqrt{k},{\cal D})$}  $\D'$ of $G$, such that  for each $t\in V(T)$, $\vert \beta_{\D'}(t)\vert \leq 45 \cdot k^{1.5}$.  
\end{lemma}
\begin{proof}
For any $s\in S(G)$, if $\vert N_B(s)\vert \geq 3k$, then $(G,B,k)$ is a \Yes\ instance.
Now we apply Lemma~\ref{lem:maindecomp} with $\ell=3\sqrt{k}$. If the output is 
a $3\sqrt{k} \times 3\sqrt{k}$ grid minor of $B$, then by Observation~\ref{obs:cycPackLargeGrid}, 
$(G,B,k)$ is a \Yes\ instance. Otherwise, we have a nice tree decomposition $\D$ of 
$B$ of width less than $15\sqrt{k}$ and a  \NSTD{$(15\sqrt{k},{\cal D})$}  $\D'$ of $G$. 
In this case, since  $\vert N_B(s)\vert < 3k$ for any $s\in S(G)$, we have that for  each $t\in V(T)$, 
$\vert \beta_{\D'}(t)\vert \leq 45 \cdot k^{1.5}$.   
%
\end{proof}

Due to Lemma~\ref{lem:cycPackFirstPhase}, to prove Theorem~\ref{thm:cycPack}, it is 
enough to prove the following lemma. 

\begin{lemma}
\label{lem:cycPacksecondPhase}
There is an algorithm that given an instance $(G,B,k)$ of \probCycPacking on map graphs, a nice tree decomposition $\D$ of $B$ of width less than $15\sqrt{k}$, and  a  \NSTD{$(15\sqrt{k},{\cal D})$}  $\D'$ of $G$ such that  for each $t\in V(T)$, 
$\vert \beta_{\D'}(t)\vert \leq 45 \cdot k^{1.5}$, 
runs in time $2^{\OO(\sqrt k\log k)} \cdot n^{\OO(1)}$, and correctly concludes whether $(G,B,k)$ is a \Yes\ instance or 
not.  
\end{lemma}

As in the case of \probKCycle, we want to bound the ``interaction'' of a solution (i.e., a cycle packing) across every bag  of $\D'$ to be $\OO(\sqrt k)$. That is, if $(G,B,k)$ is a \Yes\ instance of \probCycPacking,  
then there is a solution ${\cal C}$ with the following property. 

\begin{framed}
\noindent 
For any node $t\in V(T)$, the number of edges in $E({\cal C})$ with one endpoint in $\beta_{\D'}(t)$ and the other in 
$V(G)\setminus \gamma_{\D'}(t)$ is upper bounded by $\OO(\sqrt k)$.      
\end{framed}


The above mentioned property is encapsulated in Lemma~\ref{lem:cyclepackcross}.   
Moreover, notice that  given a set ${\cal C}$ of pairwise vertex-disjoint cycles in a graph $G$ and a cycle $C\in{\cal C}$ that is not an induced cycle in $G$, by replacing $C$ in ${\cal C}$ by an induced cycle in $G[V(C)]$, we obtain another set of pairwise vertex-disjoint cycles. So from now onwards we assume that our objective is to look for $k$ vertex-disjoint induced cycles in $G$.

\begin{lemma}[Sublinear Crossing Lemma]
\label{lem:cyclepackcross}
Let $(G,B,k)$ be a \Yes\ instance of \probCycPacking on map graphs, and let  $\D'$ be a  \NSTD{$(15\sqrt{k},{\cal D})$}  of $G$ where $\D$ is a nice tree decomposition of $B$. 
%
Then, there is a solution $\cal C$ such that each cycle in $\cal C$ is an induced cycle in $G$ and for any $t\in V(T)$, the 
number of edges with one endpoint in $\beta_{\D'}(t)$ and the other in $V(G)\setminus 
\gamma_{\D'}(t)$ is at most $360\sqrt{k}$.  
\end{lemma}



Towards the proof of Lemma~\ref{lem:cyclepackcross}, 
recall that for any $s\in S(G)$, $N_B(s)$ forms a clique in $G$ and we call it a special clique of $G$. 
If a solution $\cal C$ of $(G,B,k)$ contains a cycle with at least three vertices from $N_B(s)$, then it should be a triangle because we seek induced cycles. Towards 
finding a solution for $(G,B,k)$, we consider a solution that maximizes the number of  triangles it selects 
from the special cliques of $G$.  Let ${\mathscr S}$ be the set of solutions of $(G,B,k)$  with maximum number of  triangles from the special cliques of $G$ and which consist of induced cycles in $G$. Before proceeding to the proof of Lemma~\ref{lem:cyclepackcross}, 
we prove an analogous result for packing triangles from special cliques. 

\begin{lemma}
\label{lem:trianglepackcross}
Let $(G,B,k)$ be an instance of \probCycPacking on map graphs, and let  $\D'$ be a  \NSTD{$(15\sqrt{k},{\cal D})$}  of $G$  where $\D$ is a nice tree decomposition of $B$. 
Let ${\cal C}$ be a set of vertex-disjoint triangles from special cliques of $G$. Then, there is a set ${\cal C}'$ of vertex-disjoint triangles from special cliques of $G$ such that $\vert \cal C\vert=\vert {\cal C}'\vert$, $V({\cal C})=V({\cal C}')$, and 
for any $t\in V(T)$, the number of edges of $E({\cal C}')$ with one endpoint in $\beta_{\D'}(t)$ and the other in $V(G)\setminus 
\gamma_{\D'}(t)$ is at most $60\sqrt{k}$.  
\end{lemma}

\begin{proof}
Let $S(G)=\{s_1,\ldots,s_{\ell}\}$, and let $K_i$ be the special clique $N_B(s_i)$ in $G$ for any $i\in [\ell]$. 
Let ${\cal C}=\biguplus_{i\in [\ell]} {\cal C}_i$ be a set of vertex-disjoint triangles such that ${\cal C}_i$ is a set of triangles in $K_i$ for any $i\in [\ell]$, and ${\cal C}_i\cap {\cal C}_j=\emptyset$ for any distinct $i,j\in [\ell]$. First, we prove the following claim.

\begin{claim}
\label{claim:triangleinaclique}
Let $i\in [\ell]$. Then, there is a set ${\cal C}_i'$ of  vertex-disjoint triangles from the special clique $K_i$ of $G$ such that $\vert {\cal C}_i\vert=\vert {\cal C}_i'\vert$, $V({\cal C}_i)=V({\cal C}_i')$, and for any node $t\in V(T)$, the number of edges of $E({\cal C}_i')$ with one endpoint in $\fake(t)\cap K_i$ and the other in $V(G)\setminus \gamma_{\D'}(t)$ is at most~$4$. 
\end{claim} 
\begin{proof}
Let $\pi'$ be a postorder transversal of the  nodes in the binary tree $T$. Let 
$\pi$ be the restriction of $\pi'$ where we keep only the nodes  labelled with {\introducefakevertex}$(v)$ for some $v\in K_i$. Recall that, for any vertex $v\in V(G)$, there is at most one node in $T$ which is labelled with 
{\introducefakevertex}$(v)$ by $\D'$ (see Observation~\ref{obs:fakeprop}). 
Accordingly, denote $\pi=t_1,\ldots,t_{r}$ where  each $t_j$, $j\in [r]$, is a node labelled with  {\introducefakevertex}$(v_j)$ for some $v_j\in K_i$. Let $\sigma_1$ be the sequence $v_1,\ldots,v_{r}$, and $U=\{v_1,\ldots,v_{r}\}$. Let $\sigma_2$ be a fixed arbitrary sequence on $K_i\setminus U$. Let $\sigma$ be the sequence $(\sigma_1\sigma_2)|_{V({\cal C}_i)}$. That is, $\sigma$ is the sequence obtained by restricting $\sigma_1\sigma_2$ (the concatenation of $\sigma_1$ and $\sigma_2$) to $V({\cal C}_i)$.  Since ${\cal C}_i$ is a set of vertex-disjoint triangles from $K_i$, we have that the length of $\sigma$ is a multiple $3$. Thus, we can denote $\sigma=z_1z_2\ldots z_{3q}$ for some positive integer $q$. Notice that $z_j=v_j$ for any $j\in [r]$, and $\vert {\cal C}_i\vert=q$.  Now, we define the ``required'' set of triangles to be ${\cal C}_i'=\{[z_{3c-2}z_{3c-1}z_{3c}z_{3c-2}]~:~ c\in [q] \}$.   Clearly, ${\cal C}_i'$ is a set of vertex-disjoint triangles, $\vert {\cal C}_i'\vert=\vert {\cal C}_i\vert$ and $V({\cal C}_i)=V({\cal C}_i')$. The proof of the following statement easily follows from the definition ${\cal C}_i'$. 
\begin{itemize}
\item[$(a)$] For any $j\in [3q]$, the number of edges of $E({\cal C}_i')$   with one endpoint in $\{z_1,\ldots,z_j\}$ and the other in $\{z_{j+1},\ldots,z_{3q}\}$ is at most $2$. 
\end{itemize}
The proof of the following statement is similar in arguments to that of Claim~\ref{claim:segment}. 
\begin{itemize}
\item[$(b)$] For any  $t\in V(T)$, there is a segment $\sigma'$ of $\sigma$ such that 
each vertex in $\fake(t)\cap V({\cal C}_i')$ appears in $\sigma'$, and each vertex in $V({\cal C}_i')\setminus \gamma_{\D'}(t)$ does not appear in $\sigma'$.  
\end{itemize}

Now, we are ready to complete the proof of the claim.
Towards this, let us fix a node $t\in V(T)$. 
By statement $(b)$, 
there exist $j_1,j_2\in [3q]$ such that 
$\fake(t)\cap V({\cal C}_i')\subseteq \{z_{j_1},\ldots,z_{j_2}\}$ and $(V({\cal C}_i')\setminus \gamma_{\D'}(t))\cap \{z_{j_1},\ldots,z_{j_2}\}=\emptyset$. 
 Therefore, by statement $(a)$, we conclude that 
 the number of edges of $E({\cal C}_i')$ with one endpoint in $\fake(t)\cap K_i$ and other in $V(G)\setminus \gamma_{\D'}(t)$ is at most $4$.  
\end{proof}

%
%
%

To prove the lemma, we apply Claim~\ref{claim:triangleinaclique} for all $i\in [\ell]$ to obtain ${\cal C}_i'$ from ${\cal C}_i$,  and then let ${\cal C}'=\bigcup_{i\in[\ell]} {\cal C}_i'$. Clearly, by Claim~\ref{claim:triangleinaclique}, ${\cal C}'$ is a set of vertex-disjoint triangles, $\vert {\cal C}\vert=\vert {\cal C}'\vert$, and $V({\cal C})=V({\cal C}')$. Now, fix any $t\in V(T)$. By Claim~\ref{claim:triangleinaclique}, the number of edges of $E({\cal C}')$ with one endpoint in $\fake(t)\cap \beta_{\D'}(t)$ and the other in $V(G)\setminus \gamma_{\D'}(t)$ is at most $4r_1$  where $r_1=\vert \cliques(t)\vert$. Since the degree of each vertex in a graph consisting of only vertex-disjoint cycles is $2$, the number of edges  of $E({\cal C}')$ with one endpoint in $\org(t)\cap \beta_{\D'}(t)$ and the other in $V(G)\setminus \gamma_{\D'}(t)$ is at most $2r_2$ where $r_2=\vert \org(t)\vert$. Therefore, the number of edges  of $E({\cal C}')$ with one endpoint in $\beta_{\D'}(t)$ and the other in $V(G)\setminus \gamma_{\D'}(t)$ is at most $4r_1+2r_2\leq 4(r_1+r_2)=4\vert \beta_{\D}(t)\vert \leq 60\sqrt{k}$. 
\end{proof}

Recall that ${\mathscr S}$ is the set of solutions of $(G,B,k)$  with maximum number of  triangles from special cliques of $G$ which consists only of induced cycles. Next, we state another lemma needed for the proof of Lemma~\ref{lem:cyclepackcross}. The proof of this lemma will be the focus of most of the rest of this section.  


\begin
{restatable}
{lemma}{cyclepackcrossCtwo}
\label{lem:cyclepackcrossC2}
Let $(G,B,k)$ be an instance of \probCycPacking on map graphs, and let  $\D'$ be a  \NSTD{$(15\sqrt{k},{\cal D})$}  of $G$ where $\D$ is a nice tree decomposition of $B$. 
Let ${\cal C}\in {\mathscr S}$. Let ${\cal C}_2\subseteq {\cal C}$ be such that $C\in {\cal C}_2$ if and only if $C$ is not a triangle in a special clique of $G$.   Then, for any node $t\in V(T)$, the 
number of edges of ${\cal C}_2$ with one endpoint in $\beta_{\D'}(t)$ and the other in $V(G)\setminus 
\gamma_{\D'}(t)$ is at most $300\sqrt{k}$.  
\end{restatable}

Lemma~\ref{lem:cyclepackcross} follows immediately from Lemmata~\ref{lem:trianglepackcross} and 
\ref{lem:cyclepackcrossC2}. Our proof of  Lemma~\ref{lem:cyclepackcrossC2} requires the arguments of the following lemma (see Figure~\ref{figureK}).

\begin{lemma}
\label{lem:cyclepackingcontradiction}
Let $(G,B,k)$ be an instance of \probCycPacking on map graphs. 
Let ${\cal C}\in {\mathscr S}$, and denote ${\cal C}={\cal C}_1\uplus {\cal C}_2$ where $C\in {\cal C}_1$  if and only if  $C$ is a triangle in a special clique.  Let $K_1$ and $K_2$ be two special cliques. Then,  
there does not exist two vertices $u,v\in V({\cal C}_2)\cap K_1\cap K_2$ and four edges $e_1,e_2,e_3,e_4$
such that $(a)$ $e_1,e_2 \in E({\cal C}_2)\cap E(K_1)$,  $(b)$ $e_3,e_4 \in E({\cal C}_2)\cap E(K_2)$, $(c)$ $u$ is incident with $e_1$ and $e_3$, and $(d)$ $v$ is incident with $e_2$ and $e_4$.  
\end{lemma}
\begin{proof}
For the sake of contradiction, we assume that there exist two vertices $u,v\in V({\cal C}_2)\cap K_1\cap K_2$ and four edges $e_1,e_2,e_3,e_4$
such that $(a)$ $e_1,e_2 \in E({\cal C}_2)\cap E(K_1)$, $(b)$ $e_3,e_4 \in E({\cal C}_2)\cap E(K_2)$, $(c)$ $u$ is incident with $e_1$ and $e_3$, and $(d)$ $v$ is incident with $e_2$ and $e_4$. Let $e_1=\{u,w_1\}$, $e_2=\{v,w_2\}$, $e_3=\{u,z_1\}$ and 
$e_4=\{v,z_2\}$. We claim that all the four vertices $w_1$, $w_2$, $z_1$ and $z_2$ are distinct. 
We only prove that $w_1\notin \{w_2,z_1,z_2\}$. (All other cases are symmetric). 
Targeting a contradiction, suppose $w_1\in \{w_2,z_1,z_2\}$. 
If $w_1=w_2$, then all the four edges 
$e_1,e_2,e_3$ and $e_4$ are part of a single cycle $C\in {\cal C}_2$. 
Then, ${\cal C}'=({\cal C}\setminus \{C\})\cup \{[uvw_1u]\}$ is a solution to  $(G,B,k)$ and ${\cal C}'$ contains 
strictly more triangles from special cliques than ${\cal C}$. This a contradiction to the assumption that ${\cal C}\in {\mathscr S}$.  
The proof of the statement $w_1\neq z_2$ is the same in arguments to that of the statement $w_1\neq w_2$. 
Since $G$ is a simple graph and $e_1$ and $e_3$ are distinct edges, we have that $w_1\neq z_1$. 

So, now we have that $w_1$, $w_2$, $z_1$ and $z_2$ are distinct vertices. That is,  
$w_1uz_1$ is a subpath of a cycle $C_1$ in ${\cal C}_2$ and $w_2vz_2$ is a subpath of a cycle 
$C_2$ in ${\cal C}_2$.  
This implies that 
${\cal C}'=({\cal C}\setminus \{C_1,C_2\})\cup \{[uw_1w_2u], [vz_1z_2v]\}$ is a solution to $(G,B,k)$ and ${\cal C}'$ contains strictly more triangles from special cliques than  ${\cal C}$. This a contradiction to the assumption that ${\cal C}\in {\mathscr S}$ (see Figure~\ref{figureK} for an illustration). 
%
\end{proof}

  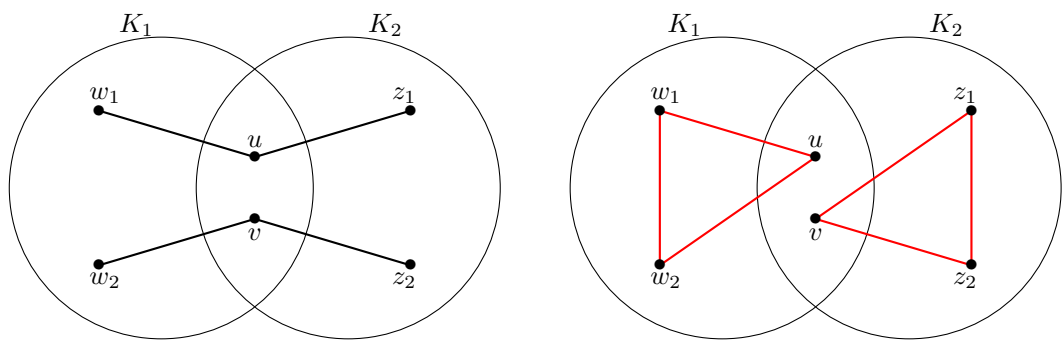
\begin{figure}[t]
   \centering
\begin{tikzpicture}[scale=0.82]

\node[draw, circle, minimum size=4cm] (a1) at (0,0) {};
\node[draw, circle, minimum size=4cm] (a1) at (3,0) {};

\node[scale=1] (u) at (1.5,0.5) {$\bullet$};
\node[scale=1] (u) at (1.5,0.75) {$u$};

\node[] (v) at (1.5,-0.5) {$\bullet$};
\node[] (v) at (1.5,-0.75) {$v$};

\node[] (x) at (-1,-1.25) {$\bullet$};
\node[] (y) at (4,-1.25) {$\bullet$};
\node[] (x) at (-0.9,-1.5) {$w_2$};
\node[] (y) at (3.9,-1.5) {$z_2$};

\node[] (w) at (-0.4,2.63) {$K_1$};
\node[] (w) at (3.6,2.63) {$K_2$};

\node[] (w) at (-0.9,1.5) {$w_1$};
\node[] (w) at (-1,1.25) {$\bullet$};
\node[] (x) at (4,1.25) {$\bullet$};
\node[] (x) at (3.9,1.5) {$z_1$};

\draw[thick] (-1,1.25)--(1.5,0.5)--(4,1.25);
\draw[thick] (-1,-1.25)--(1.5,-0.5)--(4,-1.25);


\node[draw, circle, minimum size=4cm] (a1) at (9,0) {};
\node[draw, circle, minimum size=4cm] (a1) at (12,0) {};

\draw[thick,red] (10.5,0.5)-- (8,1.25) --(8,-1.25) -- (10.5,0.5);
\draw[thick,red] (10.5,-0.5)--(13,1.25) --(13,-1.25) --(10.5,-0.5);

\node[scale=1] (u) at (10.5,0.5) {$\bullet$};
\node[] (v) at (10.5,-0.5) {$\bullet$};
\node[scale=1] (u) at (10.5,0.75) {$u$};
\node[] (v) at (10.5,-0.75) {$v$};

\node[] (x) at (8,-1.25) {$\bullet$};
\node[] (y) at (13,-1.25) {$\bullet$};
\node[] (w) at (8,1.25) {$\bullet$};
\node[] (x) at (13,1.25) {$\bullet$};

\node[] (x) at (8.1,-1.5) {$w_2$};
\node[] (y) at (12.9,-1.5) {$z_2$};
\node[] (w) at (8.1,1.5) {$w_1$};
\node[] (x) at (12.9,1.5) {$z_1$};

\node[] (w) at (8.4,2.63) {$K_1$};
\node[] (w) at (12.6,2.63) {$K_2$};

\end{tikzpicture}
\caption{Illustration for the proof of   Lemma~\ref{lem:cyclepackingcontradiction}. Relevant subpaths of $C_1$ and $C_2$ are  given in the left figure.}
\label{figureK}
 \end{figure}

We also require the following lemma in the proof of Lemma~\ref{lem:cyclepackcrossC2}. To state the lemma we need the following definition. We say that a collection of sets $\{A_1,\ldots,A_q\}$ has a system of distinct representatives if there exist distinct vertices  $a_1,a_2,\ldots,a_q$ such that $a_i\in A_i$ for all $i\in [q]$.

\begin{lemma}
\label{lem:sysrep}
Let $\{A_1,\ldots, A_{q}\}$ be a collection of  sets of size at least one such that all but at most one set have size $2$, and each element appears in at most two sets. Then, $\{A_1,\ldots,A_{q}\}$ has a system of distinct representatives.
\end{lemma}

\begin{proof} 
For the base case, $q=1$, the statement holds trivially. Consider the induction step $q>1$. There exists at most one set $A\in \{A_1,\ldots,A_{q}\}$ such that $\vert A \vert=1$. If all sets have size $2$, then choose $A$ to be an arbitrary set.   Select an element $z_A$ from $A$ as its representative. By the definition of  $\{A_1,\ldots,A_{q}\}$, there is at most one set $B$ in $\{A_1,\ldots,A_{q}\}$ such that $B\neq A$ and $z_A\in B$. Moreover $\vert B\vert=2$. Then, by induction hypothesis we have that $(\{A_1,\ldots,A_{q}\}\setminus \{A,B\})\cup \{B'\}$, where $B'=B\setminus \{z_A\}$, has a system of distinct representatives. This system of representatives along with $z_A\in A$ forms a system of distinct representatives for $\{A_1,\ldots,A_{q}\}$. 
\end{proof}

Now, we are ready to give a proof for Lemma~\ref{lem:cyclepackcrossC2}. For the convenience of the reader we restate the lemma. 

\cyclepackcrossCtwo*

\begin{proof}[Proof of Lemma~\ref{lem:cyclepackcrossC2}]
Fix a node $t\in V(T)$. 
Let $S(G)\cap \beta_{\D}(t)=\{s_1,\ldots,s_{\ell}\}$, and for each $i\in [\ell]$, let $K_i$ be the special 
clique $N_B(s_i)$ in $G$. Then, $\cliques(t)=\{K_1,\ldots,K_{\ell}\}$.
Let $G_t=G[\gamma_{\D'}(t)]$. Let ${\cal P}$ be the restriction of ${\cal C}_2$  to $G_t$. That is, ${\cal P}=(V(G_t)\cap V({\cal C}_2),E(G_t)\cap E({\cal C}_2))$. Notice that ${\cal P}$ is  a collection of pairwise vertex-disjoint cycles and paths (some paths could be of length $0$) such that the end-vertices of each path belong to $\beta_{\D'}(t)$. Observe that cycles in ${\cal P}$ are fully contained inside $G_t$ and hence the edges of these cycles do not contribute to the number of crossing edges we want to bound. Hence, we assume without loss of generality that ${\cal P}$ only contains pairwise vertex-disjoint paths  
 such that the end-vertices of each path belong to $\beta_{\D'}(t)$. 
 Let ${\cal O}\subseteq {\cal P}$ be such that $P\in {\cal O}$ if and only if  at least one end-vertex of $P$ is in $\org(t)$. Let ${\cal F}={\cal P}\setminus {\cal O}$.  Since $\vert \org(t)\vert \leq 15\sqrt{k}$, $\vert {\cal O}\vert \leq 15\sqrt{k}$.  By the definition of ${\cal F}$, for any path $P\in {\cal F}$, both the end-vertices of $P$ belong $\fake(t)$. We further classify ${\cal F}={\cal F}_0\uplus {\cal F}_1$ where 
${\cal F}_0$ contains the paths of length $0$ in ${\cal F}$, and ${\cal F}_1$ contains the paths of length 
at least $1$ in ${\cal F}$. Notice that ${\cal P}={\cal O}\uplus {\cal F}_0\uplus {\cal F}_1$. Moreover, notice that the set  $X$ of vertices in $\beta_{\D'}(t)$ that are endpoints of edges in $E({\cal C}_2)$ whose other endpoints in $V(G)\setminus \gamma_{\D'}(t)$, is the set of end-vertices of the paths in ${\cal P}$. If $\vert {\cal P}\vert \leq 75\sqrt{k}$, then$\vert X \vert \leq  150\sqrt{k}$, and hence the number of edges of ${\cal C}_2$ with one endpoint in $\beta_{\D'}(t)$ and the other in $V(G)\setminus \gamma_{\D'}(t)$ is at most $2\vert X\vert \leq 300\sqrt{k}$ (because ${\cal C}_2$ is a set of vertex-disjoint cycles). 
Thus, to complete the proof of the lemma, it is suffice to  prove that indeed $\vert {\cal P}\vert \leq 75\sqrt{k}$. 
Recall that $\vert {\cal O}\vert \leq 15\sqrt{k}$.
Thus, to prove that $\vert {\cal P}\vert \leq 75\sqrt{k}$, it is enough to prove that 
$\vert {\cal F}_0\vert \leq 45\sqrt{k}$ and  $\vert {\cal F}_1 \vert \leq  15\sqrt{k}$.

Before formally proving upper bounds on the cardinalities of ${\cal F}_0$ and ${\cal F}_1$, we give a high level overview of the proof. 
 Towards bounding  ${\cal F}_0$, we construct a planar bipartite subgraph $H$ of $B$ with bipartition ${\cal F}_0\uplus \{s_1,\ldots,s_{\ell}\}$ and the following property: for any $v\in {\cal F}_0$, $d_H(v)=2$. Therefore, $\vert {\cal F}_0\vert=\frac{\vert E(H)\vert}{2}$.   To upper bound $\vert {\cal F}_0\vert$, we construct  a minor $H'$ of $H$ on $\ell$ vertices and $\frac{\vert E(H)\vert}{2}$ edges and prove that $(a)$ $H'$ is a graph without self-loops and parallel edges. 
Since $H'$ is a minor of a planar graph $H$, $H'$ is also a planar graph. Thus, $H'$ is a planar graph without self-loops and parallel edges, and $\vert V(H')\vert =\ell$. This implies that the number of edges in $H'$ is at most $3\ell-6$ (because the number of edges in a simple planar graph on $N$ vertices is at most $3N-6$). This, in turn,  will imply that $\vert {\cal F}_0\vert \leq 3\ell-6\leq 45\sqrt{k}$. 
Towards upper bounding  $\vert {\cal F}_1 \vert $ by $15\sqrt{k}$, we construct a graph $H_1$ on the vertex set  $\{s_1,\ldots,s_{\ell}\}$. 
To construct the edge set of $H_1$, we add one edge for each path in ${\cal F}_1$. Thus, $\vert {\cal F}_1\vert = \vert E(H_1)\vert$. To upper bound $\vert E(H_1)\vert$, we prove that $(b)$ $H_1$ is a forest. This implies that $\vert {\cal F}_1\vert = \vert E(H_1)\vert\leq \ell-1<15\sqrt{k}$. The proofs of both the Statements $(a)$ and $(b)$ use the assumption that ${\cal C}\in {\mathscr S}$. 

Now, we move towards the formal proof of $\vert {\cal F}_0\vert \leq 45\sqrt{k}$. 
Notice that ${\cal F}_0$ is a set of paths of length $0$ 
 and each $v\in {\cal F}_0$ belongs to $\fake(t)$.   Any edge  in $E({\cal C}_2)$ incident with a vertex $v\in {\cal F}_0$ belongs to $\bigcup_{i\in [\ell]} E(K_i)$ (because of Lemma~\ref{lem:edgetwokinds}).  For any $v\in {\cal F}_0$, let $e_v=\{v,x_v\}$ and $e_v'=\{v,y_v\}$ be the edges of $E({\cal C}_2)$ incident with $v$.  


\begin{claim}
\label{claim:noloop}
For any $v\in {\cal F}_0$, there does not exist 
$i\in [\ell]$ such that both $e_v,e_v'\in E(K_i)$.  Moreover, $x_v\neq y_v$ and $x_v,y_v\notin V(G_t)$. 
\end{claim}
\begin{proof}
Let $v\in {\cal F}_0$. Since $G$ is a simple graph, we have that $x_v\neq y_v$. 
The condition $x_v,y_v\notin V(G_t)$ follows from the definition of ${\cal F}_0$. 
Towards a contradiction, 
suppose there exists  $i\in [\ell]$ such that $e_v,e_v'\in E(K_i)$. That is, $v,x_v,y_v\in K_i$. 
Let $C\in {\cal C}_2$ be such that $v,x_v,y_v\in V(C)$.   Since $v,x_v,y_v\in K_i$, $C'=[vx_vy_vv]$ forms a triangle in the special clique $K_i$. Then, ${\cal C}'=({\cal C}\setminus \{C\})\cup \{C'\}$ is a cycle packing such that $\vert {\cal C}'\vert=\vert {\cal C}\vert$ and 
 ${\cal C}'$ contain 
strictly more triangles from special cliques than ${\cal C}$. This is a contradiction to the fact that 
${\cal C}\in {\mathscr S}$ (see Figure~\ref{figureoneclique} for an illustration). 
\end{proof}

  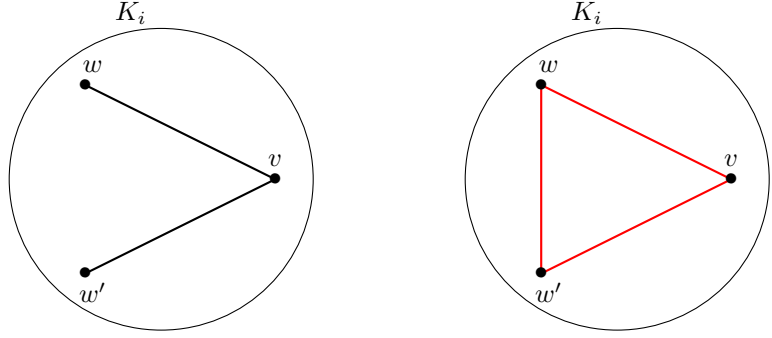
\begin{figure}[t]
   \centering
\begin{tikzpicture}[scale=1]

\node[draw, circle, minimum size=4cm] (a1) at (0,0) {};

\node[scale=1] (u) at (1.5,0) {$\bullet$};
\node[scale=1] (u) at (1.5,0.25) {$v$};

\node[] (x) at (-1,-1.25) {$\bullet$};
\node[] (x) at (-0.9,-1.5) {$w'$};

\node[] (w) at (-0.4,2.2) {$K_i$};

\node[] (w) at (-0.9,1.5) {$w$};
\node[] (w) at (-1,1.25) {$\bullet$};

\draw[thick] (-1,1.25)--(1.5,0)-- (-1,-1.25);

\node[draw, circle, minimum size=4cm] (a1) at (6,0) {};
\draw[thick,red] (5,1.25)--(7.5,0)-- (5,-1.25)--(5,1.25);

\node[scale=1] (u) at (7.5,0) {$\bullet$};
\node[scale=1] (u) at (7.5,0.25) {$v$};

\node[] (x) at (5,-1.25) {$\bullet$};
\node[] (x) at (5.1,-1.5) {$w'$};

\node[] (w) at (5.6,2.2) {$K_i$};

\node[] (w) at (5.1,1.5) {$w$};
\node[] (w) at (5,1.25) {$\bullet$};

%
%
%
%
%

\end{tikzpicture}
\caption{Illustration of Proof of Claim~\ref{claim:noloop}. $C'=[wvw'w]$ is a triangle in the special clique $K_i$.}
\label{figureoneclique}
 \end{figure}
 

We proceed to define a subgraph $H$ of $B$ on the vertex set $\{s_1,\ldots,s_{\ell}\}\cup {\cal F}_0$, and whose edge set will be defined immediately. Towards this, note that 
by Lemma~\ref{lem:edgetwokinds}, we have that for any $v\in {\cal F}_0$,  $e_v\in E(K_i)$ and $e_v'\in E(K_j)$ for some special cliques $K_i$ and $K_j$ in $\cliques(t)$. Moreover, by Claim~\ref{claim:noloop}, $e_v'\notin E(K_i)$ and $e_v\notin E(K_j)$.  For any $v\in {\cal F}_0$, we choose $i,j\in [\ell]$ such that $e_v\in E(K_i)$ and $e_v'\in E(K_j)$ (notice that $K_i\neq K_j$). Then, we add the edges $\{v,s_i\}$ and $\{v,s_j\}$ to $H$. Notice that since $e_v\in E(K_i)$ and $e_v'\in E(K_j)$, we have that $\{v,s_i\}$ and $\{v,s_j\}$ are edges in $B$. Therefore, $H$ is a subgraph of $B$. That is, $H$ is a planar bipartite graph with bipartition $\{s_1,\ldots,s_{\ell}\}\uplus {\cal F}_0$, and the degree of each $v\in {\cal F}_0$ is exactly $2$. 

Next, we construct a minor of $H$ by arbitrarily choosing, for each $v\in {\cal F}_0$, exactly one edge incident to $v\in {\cal F}_0$ and contracting it into its other endpoint (so we refer to the resulting vertex using the identity of the other endpoint). Let the resulting graph be $H_1$. Notice that 
$H_1$ is a planar graph on the vertex set $\{s_1,\ldots,s_{\ell}\}$. As each vertex in $v\in {\cal F}_0$ has degree exactly $2$ in $H$, and we contracted exactly one edge incident with $v$ to obtain $H_1$, we have  that $\vert {\cal F}_0\vert$ is equal to the number of edges in $H_1$. Since $H_1$ is a planar graph with vertex set $\{s_1,\ldots,s_{\ell}\}$, if there are no self-loops and parallel edges in $H_1$, then the number of edges in $H_1$ (and hence $\vert {\cal F}_0\vert$) is at most $3\ell-6$ (because the number of edges in a simple planar graph on $N$ vertices is at most $3N-6$).  Next, we prove that indeed $H_1$ is a graph without self-loops and parallel edges.

\begin{claim}
$H_1$ is a graph without parallel edges and self-loops. 
\end{claim}

\begin{proof}
Targeting a contradiction, suppose $H_1$ has a self-loop. Then, $H$ has two vertices $v\in {\cal F}_0$ and $s_i\in \{s_1,\ldots,s_{\ell}\}$ such that there are two edges between $v$ and $s_i$ in $H$. That is, both $e_v=\{v,x_v\}$ and $e_v'=\{v,y_v\}$ are edges in $E(K_i)$. However, this is not possible because of Claim~\ref{claim:noloop}. 


Now, we prove that $H_1$ does not contain two parallel  edges. For the sake of contradiction, assume that there exist $s_i$ and $s_j$ in $V(H_1)$ such that there are two parallel edges between $s_i$ and $s_j$ in $H_1$. This implies that there exist two distinct vertices $v,u\in {\cal F}_0$ such that $e_u=\{u,x_u\}$ and $e_v=\{v,x_v\}$ are edges in $K_i$ and $e_u'=\{u,y_u\}$, and $e_v'=\{v,y_v\}$ are edges in $K_j$. However, this  contradicts Lemma~\ref{lem:cyclepackingcontradiction}.  This completes the proof of the claim. 
\end{proof}

Thus, we have proved that $\vert {\cal F}_0\vert \leq 3\ell-6 \leq 3\cdot 15\sqrt{k}-6\leq 45\sqrt{k}$ (because $\ell\leq \vert \beta_{\D}(t)\vert$).

Next, we  prove that $\vert {\cal F}_1\vert \leq 15\sqrt{k}$. 
Let $\ell'=\vert {\cal F}_1\vert$ and  ${\cal F}_1=\{P_1,\ldots,P_{\ell'}\}$.
Notice that ${\cal F}_1$ is a set of paths of length at least $1$
and each end-vertex of any path in ${\cal F}_1$
belongs to $\fake(t)$.  
For any $i\in [\ell']$, let $u_i$ and $v_i$ be the end-vertices of $P_i$ 
(since $P_i$ is a path of length at least $1$, we have that $u_i\neq v_i$).
Moreover, for any $i\in [\ell']$, let $\{u_i,w_i\}$ and $\{v_i,z_i\}$ be the edges of $E({\cal C}_2)\setminus E({\cal F}_1)$ incident with $u_i$ and $v_i$, respectively.  
%

Now, we create an auxiliary graph $H'$ on the vertex set $\{s_1,\ldots,s_\ell\}$. 
By Lemma~\ref{lem:edgetwokinds}, we have that for 
any $i\in [\ell']$,  $u_i,w_i\in K$ and $v_i,z_i\in K'$ for some special cliques $K$ and $K'$ in $\cliques(t)$.   
Moreover, it is easy to see that $(i)$ there is no special clique $K''$ such that $K''\cap \{u_i,v_i,w_i,z_i\}\geq 3$. Otherwise,  
we replace the cycle containing $\{u_i,v_i,w_i,z_i\}$ in ${\cal C}_2$ with a triangle in   $K''\cap \{u_i,v_i,w_i,z_i\}$ and thereby  contradict the assumption that  ${\cal C}\in {\mathscr S}$. 
For any $i\in [\ell']$, we choose $j,j'\in [\ell]$ such that $u_i,w_i\in K_j$ and 
$v_i,z_i\in K_{j'}$ (notice that $j\neq j'$). Then, we add an edge between $s_j$ 
and $s_{j'}$ in $H'$, and we denote this edge with $g_i$. 
Clearly, $H'$ is a graph without self-loops.  

 \begin{claim}
\label{claim:Hforest}
$H'$ is a forest. 
\end{claim}

 \begin{proof}
 For the sake of contradiction, assume (without loss of generality) that there is a cycle $L=[s_{i_1}g_{1}s_{i_2}g_{2}\ldots s_{i_r}g_{r}s_{i_1}]$  in $H'$. Recall that $K_{i_j}=N_B(s_{i_j})$ is a special clique in $\cliques(t)$ for any $j\in [r]$.   
 From the definition of $L$,  
 we have that   $(a)$ for any $j\in [r]$, $u_j,w_j\in K_{i_j}$, $(b)$ for any $j\in [r-1]$, $v_j, z_j\in K_{i_{j+1}}$, 
 and $(c)$ $v_r,z_r\in K_{i_1}$. By the definition of ${\cal F}_1$, we have that $u_1,\ldots,u_r,v_1,\ldots,v_r$ are distinct vertices. See Figure~\ref{figure:forestHprime} for an illustration.   
Let ${\cal C}'_2\subseteq {\cal C}_2$ be such that 
 $C\in {\cal C}'_2$ if and only if $V(C)\cap \{u_1,\ldots,u_r\}\neq \emptyset$.  
 Notice that each cycle in ${\cal C}'_2$ contains at least one vertex in $\{u_1,\ldots,u_r\}$. 
  Therefore, $\vert {\cal C}'_2\vert \leq r$. 
 Moreover, any cycle in ${\cal C}_2$ that has non-empty intersection with $\{u_i,w_i,v_i,z_i ~\vert~i\in [r]\}$ also belongs to ${\cal C}_2'$.  
 
 We prove the claim using a proof by contradiction. 
 Towards that, we will prove that there are $r$ triangles in the special cliques $K_{i_1},\ldots, K_{i_r}$ using vertices from $\{u_i,v_i,w_i,z_i~\vert~ i\in [r]\}$;  then, by replacing ${\cal C}_2'$ with these triangles, we will reach in a contradiction to the assumption that ${\cal C}\in {\mathscr S}$.  
 Let $Z=\{w_j,z_j~\vert~ j\in [r]\}$. Let $S_j=\{w_j,z_{j-1}\mod r ~\vert~ j\in [r]\}$ for all $j\in [r]$ (here, $0\mod r$ is defined to be $r$). 
Now, we prove that $\{S_1,\ldots,S_r\}$ has a system of distinct representatives using Lemma~\ref{lem:sysrep}. 
%
%
Towards that, we first prove that $\vert S_j\vert =2$ for any $j\in [r]$. Suppose not. Then, $\vert S_j\vert =1$ for some $j\in [r]$. This implies that $w_j=z_{s}$, where $s=(j-1)\mod r$.  Then, $[v_sw_ju_j]$ is a subpath of a cycle $C$ in ${\cal C}_2$. Moreover, $v_s,w_j,u_j\in K_{i_j}$. Therefore, by replacing $C\in {\cal C}$ with the triangle  $[v_sw_ju_jv_s]$, we contradict the assumption that ${\cal C}\in {\mathscr S}$.  
 Notice that since the degree of each vertex in a set of vertex-disjoint cycles is exactly $2$, we have that each vertex 
 $z\in Z$ appears in at most two sets in $\{S_1,\ldots,S_r\}$. Thus, by Lemma~\ref{lem:sysrep}, $\{S_1,\ldots,S_r\}$ has a system of distinct representatives. 
 
 
%
%

    \begin{figure}[t]
   \centering
\begin{tikzpicture}[scale=1]

\node[draw, circle, minimum size=2cm] (a1) at (0,0) {};

\node[scale=1] (u1) at (0.4,-0.5) {$\bullet$};
\node[scale=1] (u) at (0.7,-0.45) {$u_1$};
\node[scale=1] (v4) at (-0.6,-0.5) {$\bullet$};
\node[scale=1] (u) at (-0.3,-0.5) {$v_4$};

\node[scale=1] (w1) at (0.4,0.5) {$\bullet$};
\node[scale=1] (w) at (0.4,0.7) {$w_1$};
\node[scale=1] (z4) at (-0.6,0.5) {$\bullet$};
\node[scale=1] (z) at (-0.4,0.7) {$z_4$};

\draw (0.4,-0.5)--(0.4,0.5);
\draw (-0.6,-0.5)--(-0.6,0.5);

%
%
%
%

\node[draw, circle, minimum size=2cm] (a1) at (3,0) {};

\node[scale=1] (u2) at (3.4,-0.5) {$\bullet$};
\node[scale=1] (u) at (3.7,-0.45) {$u_2$};
\node[scale=1] (v1) at (2.6,-0.5) {$\bullet$};
\node[scale=1] (v) at (2.9,-0.5) {$v_1$};

\node[scale=1] (w2) at (3.4,0.5) {$\bullet$};
\node[scale=1] (w) at (3.4,0.7) {$w_2$};
\node[scale=1] (z1) at (2.6,0.5) {$\bullet$};
\node[scale=1] (z) at (2.6,0.7) {$z_1$};

\draw (3.4,-0.5)--(3.4,0.5);
\draw (2.6,-0.5)--(2.6,0.5);

\node[draw, circle, minimum size=2cm] (a1) at (6,0) {};

\node[scale=1] (u3) at (6.4,-0.5) {$\bullet$};
\node[scale=1] (u) at (6.7,-0.45) {$u_3$};
\node[scale=1] (v2) at (5.6,-0.5) {$\bullet$};
\node[scale=1] (v) at (5.85,-0.5) {$v_2$};

\node[scale=1] (w3) at (6.4,0.5) {$\bullet$};
\node[scale=1] (w) at (6.4,0.7) {$w_3$};
\node[scale=1] (z2) at (5.6,0.5) {$\bullet$};
\node[scale=1] (z) at (5.6,0.7) {$z_2$};

\draw (6.4,-0.5)--(6.4,0.5);
\draw (5.6,-0.5)--(5.6,0.5);

\node[draw, circle, minimum size=2cm] (a1) at (9,0) {};

\node[scale=1] (u4) at (9.4,-0.5) {$\bullet$};
\node[scale=1] (u) at (9.7,-0.5) {$u_4$};
\node[scale=1] (v3) at (8.6,-0.5) {$\bullet$};
\node[scale=1] (v) at (8.85,-0.45) {$v_3$};

\node[scale=1] (w4) at (9.4,0.5) {$\bullet$};
\node[scale=1] (w) at (9.4,0.7) {$w_4$};
\node[scale=1] (z3) at (8.6,0.5) {$\bullet$};
\node[scale=1] (z) at (8.6,0.7) {$z_3$};

\draw (9.4,-0.5)--(9.4,0.5);
\draw (8.6,-0.5)--(8.6,0.5);

\draw[thick,red] (1,0)--(2,0);

\draw[thick,red] (4,0)--(5,0);

\draw[thick,red] (7,0)--(8,0);

\draw[thick,red] (0,1) to[out=20,in=160] (9,1);

\draw[thick,green] (0.4,-0.5) to[out=-40,in=220] (2.6,-0.5);

\draw[thick,green] (3.4,-0.5) to[out=-40,in=220] (5.6,-0.5);

\draw[thick,green] (6.4,-0.5) to[out=-40,in=220] (8.6,-0.5);

\draw[thick,green] (-0.6,-0.5) to[out=-30,in=210] (9.4,-0.5);

\node[] (u) at (-1.1,0.7) {$K_{i_1}$};
\node[] (u) at (1.9,0.7) {$K_{i_2}$};
\node[] (u) at (4.9,0.7) {$K_{i_3}$};
\node[] (u) at (7.9,0.7) {$K_{i_4}$};

%
%
%
\end{tikzpicture}
\caption{The red curves are edges in the cycle $L$. The green curves are some paths in ${\cal F}_1$ that give rise to the edges (drawn as red curves) in $L$. The vertices $u_1,\ldots,u_4,v_1,\ldots,v_4$ are distinct. 
$Z=\{w_1,\ldots,w_4,z_1,\ldots,z_4\}$, $S_1=\{w_1,z_4\}, S_2=\{w_2,z_1\},S_3=\{w_3,z_2\}$ and $S_4=\{w_4,z_3\}$. 
Even though the circles representing special cliques are disjoint in the illustration, they may not be disjoint in general.}
\label{figure:forestHprime}
 \end{figure}
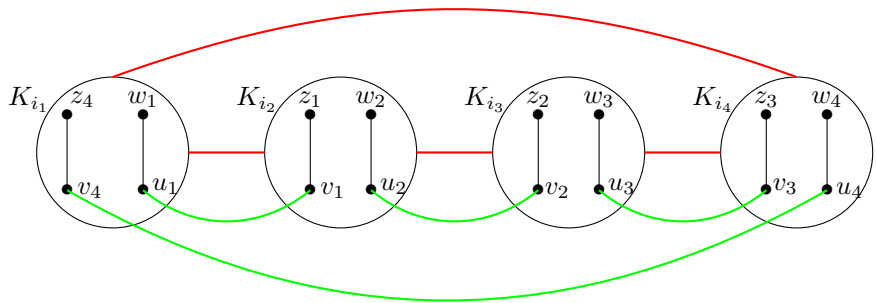

 Now, we construct a set of $r$ triangles  as follows. For each $S_j$, let $a_j$ be its representative.  
 Notice that $a_j\in K_{i_j}$. 
Let  $\widehat{{\cal C}}=\{[u_ja_jv_{(j-1)\mod r}u_j] ~:~j\in [r]\}$. 
 Notice that $\widehat{\cal C}$ is a set of vertex-disjoint trangles in special cliques of $G$. 
 %
 Therefore, by replacing cycles in ${\cal C}_2'$ with cycles in $\widehat{\cal C}$, 
 we get a solution of $(G,B,k)$ with more triangles from special cliques than that of ${\cal C}$. This a contradiction to the assumption that ${\cal C}\in {\mathscr S}$. This completes the proof of the claim. 
\end{proof}

Since $H'$ is a forest (by Claim~\ref{claim:Hforest}), we have that $\vert E(H')\vert <\vert V(H')\vert $. 
Notice that for each path in ${\cal F}_1$, we added exactly one edge to $H'$ and hence $\vert E(H')\vert = \vert {\cal F}_1\vert$.  
 Therefore, 
 we have that  $\vert {\cal F}_1\vert=\vert E(H')\vert < \vert V(H')\vert \leq  \ell  \leq 15\sqrt{k}$. 
This completes the proof of the lemma. 
\end{proof}

Using Lemma~\ref{lem:cyclepackcross}, we can design an algorithm for \probCycPacking and prove Lemma~\ref{lem:cycPacksecondPhase}. This algorithm is a   dynamic programming algorithm over the \NSTD{$(15\sqrt{k},{\cal D})$}  $\D'$ of $G$, such that  for each $t\in V(T)$, 
$\vert \beta_{\D'}(t)\vert \leq 45 \cdot k^{1.5}$.  The number of states at any node $t\in V(T_{\cal D'})$ is upper bounded by $2^{\OO(\sqrt{k}\log k)}\cdot n^{\OO(1)}$, and thus resulting in an algorithm with running time  $2^{\OO(\sqrt{k}\log k)}\cdot n^{\OO(1)}$. 
For further details, we refer to the proof sketch of Lemma~9.3 in \cite{subexpudg}
for a similar algorithm on a path decomposition of the input graph with similar properties.

%% file: conclusion.tex
\section{Conclusion}\label{sec:conclusion}
In this paper, we gave  subexponential algorithms of running time $2^{\cO({\sqrt{k}\log{k}})} \cdot n^{\cO(1)}$ for \probKCycle, \probKPath, \probFVS and  \probCycPacking\ 
  on map graphs. The following are some open questions along the direction of our work.  
  \begin{itemize}
  \item 
  Is there a parameterized subexponential time  algorithm  for \probKexactCycle\ on map graphs? Here, we want to test whether the input graph has a cycle of length exactly $k$.
   \item Can we get a better running time (by shaving off the $\log{k}$ factor in the exponent) for 
  \probKCycle, \probKPath, \probFVS and  \probCycPacking\ on map graphs?
\item  It is noteworthy to remark that a simple disjoint union trick \cite{BodlaenderDFH09,fomin_lokshtanov_saurabh_zehavi_2019}  implies that \probKCycle\ and \probKPath  do not admit a polynomial kernel on map graphs. Can we get a Turing polynomial kernel for these problems on map graphs? That is, 
can we 
give a polynomial time algorithm that given an instance of \probKCycle\ or \probKPath on map graphs, produces polynomially many reduced instances of size polynomial in $k$ such that the input instance is a \Yes\ instance if and only if one of the reduced instances is? 
\item Can we get  a general characterization of parameterized problems admitting subexponential algorithms on map graphs like the bidimensionality theory for planar graphs?
  \end{itemize}